\documentclass{article}
\pdfoutput=1
\usepackage[a4paper,top=3cm, bottom=3cm]{geometry}
\usepackage{amscd}
\usepackage{float}
\usepackage{wrapfig}
\usepackage{mathtools}
\usepackage[english,german]{babel}
\usepackage[latin1]{inputenc}
\usepackage{amsmath}
\usepackage{amssymb}
\usepackage[T1]{fontenc}
\usepackage{graphicx}
\usepackage{soul}
\usepackage[usenames,dvipsnames,svgnames,table]{xcolor}
\usepackage{color}
\usepackage{xcolor}
\usepackage{tcolorbox}
\usepackage[section]{placeins}
\usepackage{cite}
\usepackage{multirow}
\usepackage{tabularx}
\usepackage{tabu}
\usepackage{tikz}
\usepackage[bookmarks=false]{hyperref}
\usetikzlibrary{arrows,intersections}
\usetikzlibrary{arrows.meta}
\usetikzlibrary{shapes,backgrounds}
\usetikzlibrary{matrix}
\usetikzlibrary{positioning}
\newcolumntype{L}[1]{>{\raggedright\arraybackslash}p{#1}} 
\newcolumntype{C}[1]{>{\centering\arraybackslash}p{#1}} 
\newcolumntype{R}[1]{>{\raggedleft\arraybackslash}p{#1}} 
\usepackage{amsthm}
\usepackage{algorithm}
\usepackage{algpseudocode}
\usepackage{verbatim}
\usepackage{theoremref}
\usepackage{adjustbox}
\usepackage{enumitem}
\usepackage{rotating}
\usepackage{caption}

\theoremstyle{definition}
\newtheorem*{definition}{Definition}
\newtheorem{proposition}{Proposition}

\def\real{{\mathbb{R}}}

\def\real{{\mathbb{R}}}

\def\mean{{\mathop{\mathsf{mean}}}}
\def\var{{\mathop{\mathsf{var}}}}

\def\simul{{\mathop{\mathsf{sim}}}}
\def\Pr{{\mathop{\mathsf{Prob}}}}

\DeclareMathOperator{\E}{\mathsf{E}}
\DeclareMathOperator{\Var}{\mathsf{Var}}
\DeclareMathOperator{\dd}{\mathsf{d}\!}

\usepackage{scalerel}

\begin{document}
\selectlanguage{english}

\title{Malware Epidemic Effects in a Kinetic Conflict Model}
\author{Joachim Draeger, Stephanie \"Ottl\\[1mm]
\em\small IABG mbH, Ottobrunn, Germany}
\date{}
\maketitle

\begin{abstract}
We propose a framework for examining the effects of infections with
self-replicating malware on military forces engaged in kinetic combat. The
framework uses models, in which kinetic attrition is represented by a Lanchester
model coupled with an SIR-like model describing the malware propagation across
the forces. Basic knowledge about the expected circumstances restricts the
set of scenarios to be analyzed using the model. Remaining uncertainties
are taken into account as random variations given by information-theoretic
principles. The situation assessment is realized by Monte-Carlo simulations
with the risk as a possible assessment measure. 
	
An application of the proposed framework to a simple exemplary situation
demonstrates its usage in practice. 
The assumed uncertainties about the considered situation lead to an
outcome statistics, which changes corresponding to the improving knowledge
about the situation. Large uncertainties may lead to results profoundly
different from point estimates. 
For assuring practicability, the paper provides options to
determine the values of important model parameters by measurement.
It also discusses how to 
utilize the assessment results calculated with help of the framework.
\end{abstract}


\section{Introduction}

Among all threats related to cyber security, self-replicating malware,
like viruses and worms, are some of the most important ones. Due to the 
capability of self-reproduction, the initial infection of a single, 
unimportant system component may cause catastrophic damage to the
overall system at the end. 
The disproportionality between the low effort and the potentially 
significant damage qualifies the usage of malware for large scale attacks 
and cyber warfare \cite{musman2010evaluating,halpin2006cyberwar}.
Indeed, the Cyber Conflict Studies Association 
CCSA\footnote{http://www.cyberconflict.org} lists numerous incidents at
national level. We mention some examples:
In 1998, the NATO attacked infrastructure and command \& control 
structures in Serbia during the Kosovo war with self-replicating malware 
\cite{grant1999airpower}, enabling an especially successful air campaign. 
In April 2007, Estonia was attacked at cyber level, presumably by Russia
\cite{czosseck2011estonia}. The attack has targeted ministries, banks, and 
media and caused injuries due to riots resulting from the effectiveness
of the attack.
Trojans have enforced a temporary shutdown of the computer network 
of the German Bundestag in 2015 \cite{martinjung2015fuer,reuters2015cyber}.
Other examples of attacks with self-replicating malware can be found in 
\cite{naraine2008coordinated, 
geers2015cyber, pagliery2016scary,case2016analysis}. 


An adequate analysis of malware infection 
effects must take the (non-)availability of functionalities
into account. The degradation of the availability caused by 
the malware infection leads to a natural loss function and allows an 
objective and quantitative assessment of the malware effects. 
This paper considers a conventional battle (i.e. kinetic combat) 
between two opposing military forces as an exemplary system 
that has been affected by malware.
A Lanchester model is used to describe
the two forces and the kinetic combat between them.
The propagation of malware across the 
forces is represented by a SIR-like epidemic model.
%
The analysis of the malware effects on system availability 
is executed in a framework allowing the portrayal of specific 
situations by individual models. This contributes to
the flexibility of the proposed approach, which paves the way
towards a more faithful representation of real world situations. 
The framework also takes into account that the common idealistic 
anticipation of perfect knowledge is not always met in practice.
Instead, it expresses missing knowledge as epistemic uncertainties. 
Simulation runs can be used for analyzing the model behavior.
When organized in a Monte Carlo method,
the sampling of simulation runs can also take uncertainties into account. 
The approach can especially be used for risk assessments.

Several papers are considering co-occurring kinetic combat and 
cyber warfare with malware.
Mishra and Prajapati \cite{mishra2013modelling} discuss 
cyber warfare based on a differential equation system, 
but they do not take availability
aspects into account; instead, they focus on a stability analysis.
McMorrow \cite{mcmorrow2010science} aims at developing a deeper understanding
of large scale cyber attacks based on a comparison of biological and 
malware epidemics.
Schramm \cite{schramm2013lanchester} has developed a corresponding
model combining a kinetic battle situation with a one-sided malware 
attack based on the SIR model \cite{kermack1927contribution}. Its 
symmetrization by Yildiz \cite{yildiz2014modeling} assures that both
forces have the same capabilities and vulnerabilities. Both 
\cite{yildizmodeling} and \cite{yildiz2014modeling} examine and expand 
the models of Schramm. 

The paper is structured as follows. The class of models used in our framework
for the investigation of malware effects is presented in
section~\ref{modelingstructure}. A simple homogeneous model belonging
to this class is discussed as an example in section~\ref{secexamplemodel}.
Section~\ref{sectwo} develops our concept of a simulation-based 
computational analysis. The method for including uncertainties
is elaborated in this section as well.
Section~\ref{secfour} analyzes an exemplary situation with help of the
framework, whereby the model of section~\ref{secexamplemodel} is reused.
These considerations are supplemented by section~\ref{secapplic}, which
discusses the application of the framework from the practice point of view.
The paper closes with an outlook in section~\ref{secdiscussion}.

\section{An Integrated Kinetic-Cyber Model}\label{modelingstructure}

\subsection{Modeling Strategy}

Due to specific circumstances like individual force structures and
malware attack vectors, the usage of an 'universal' representation 
will be doomed to failure. Individual models are required for
adapting the model to the situational characteristics. 
Thus, a {\em framework} is proposed, in which the underlying model
$M\in \mathcal{M}$ can be chosen from a model class $\mathcal{M}$. 
The restriction to models belonging to $\mathcal{M}$ assures
essential model properties like the fading dynamics for $t\rightarrow\infty$
in proposition~\ref{dynamicsatinf}
(see section~\ref{secmodelextensions} below). 

The framework is built on the work of Schramm \cite{schramm2013lanchester} and Yildiz \cite{yildiz2014modeling}. Accordingly, $\mathcal{M}$ consists of (generalized) system dynamics models composed of two main components: First, a submodel $M_K$ representing kinetic combat --- or, at least, the corresponding attrition --- by a Lanchester-like model \cite{das2019fitting}, and, second, a submodel $M_C$ representing cyber combat by a SIR-like model \cite{del2015mathematical} of malware epidemics. 
The experience gained so far confirms that the system dynamics paradigm 
allows a discussion of basic mechanisms and foundational aspects.
Concerning Lanchester, one may take a look at e.g. \cite{das2019fitting} and references therein. An application to a more recent conflict is given in \cite{kamaluregression}. Concerning malware epidemics, the realism of compartmental models is discussed in \cite{del2015mathematical}.
The system dynamics approach may be overstrained, though, when aiming at precise quantitative predictions.
Both kinetic combat and cyber warfare are essentially ill-defined \cite{kress2012modeling,davis1991base,tolk2012challenges,del2015mathematical}. The ill-definedness is caused by the unpredictability of individual human decisions and by the potential occurrence of unexpected events. Therefore, just relying on more advanced models will not bring about a higher predictive quality compared to simpler models. Additionally, simple models have some principal advantages. Their small computational complexity, for example, enables the sampling of large numbers of simulation runs --- we will make use of this option in section~\ref{secthree}. 

%

\subsection{The Kinetic Component of the Model}

Kinetic combat is modeled using Lanchester equations
\cite{lanchester1956mathematics,mackay2006lanchester},
which consider two opposing forces Blue and Red. 
Since the technical systems used in military forces have
usually quite different capabilities and vulnerabilities,
heterogeneous Lanchester models \cite{das2019fitting} are permitted. 
Although the correspondence between Lanchester model and real system may only be approximate, the basic idea of Lanchester models to reduce the dynamics of combat to attrition rates is notwithstanding considered as a useful concept \cite{ancker1987validity}. Accordingly, Lanchester attrition models are used as components of more encompassing models of kinetic combat~\cite{bathe1989development,romero1991new}. 

In the kinetic component, all model parameters can be chosen individually for 
Blue and Red. This is indicated by an index ${}_b$ resp.\ ${}_r$.
Since the effectiveness $\delta_{ij}$ of an attacking force element is not only
determined by its own capabilities, but also by the capabilities 
of the defender, the model parameters depend on the 
weapon system types $i,j$ of both attacker and defender. 
Accordingly, a general form of the heterogeneous Lanchester model is
\begin{equation}\label{heterogeneouslanchester}
\begin{array}{rcl}
	\frac{\dd b_j}{\dd t} &= \sum\limits_{i=1}^n -\delta_{r,ij} r_i^{p_{r,ij} }(t) b_j^{q_{r,ij} }(t)\\[1mm]
	\frac{\dd r_j}{\dd t} &= \sum\limits_{i=1}^n -\delta_{b,ij} b_i^{p_{b,ij} }(t) r_j^{q_{b,ij} }(t)
\end{array}
\end{equation}
with $j=1,\ldots,n$ and initial conditions $b_j(0),r_j(0)\ge 0$.
In these equations, $b_j(t)$ resp.\ $r_j(t)$ represents the 
number of the blue resp.\ red force elements of type $j$ at time $t$.
The parameters $p_{b/r,ij}$ determine the capabilities of the forces as 
an attacker, whereas $q_{b/r,ij}$ do so for the defender. Equation 
system~(\ref{heterogeneouslanchester}) is well-defined unless one
of $r_j$ or $b_j$ becomes zero during the evolution. This problem is 
discussed in more detail in section~\ref{modconslab}.

Following the idea of Lanchester's model, each set of force elements represented by an $r_j$ or $b_j$ should behave homogeneously concerning attrition. The heterogeneity of the model will thus typically respect the organizational structure of the involved forces. We may also take their spatial arrangement into account in order to distinguish between the parts of the forces that are engaged and those, which are not \cite{duffey2017dynamic}. Range-dependent attrition rates due to restrictions of the firing range of weapons can be encoded as well \cite{hughes1995two}.

\subsection{The Cyber Component of the Model}

The cyber component $M_C$ describes the malware propagation across a force by an SIR-like model. Kermack \& McKendrick \cite{kermack1927contribution} developed the SIR model originally for the description of biological epidemics, but later it was successfully applied to the propagation of malware (e.g. \cite{sung2013using,mishra2012differential,sneha2015survey}) as well. The notation reflects the three different infection states $S,I,R$ of the model representing {\em S}usceptible (here called vulnerable), {\em I}nfected, and {\em R}ecovered (here called patched) force elements. The model assumes that a vulnerable element in contact with an infected individual may become infected itself with a rate $\beta$.  Infected elements may be patched with a rate $\bar \gamma$, being immune against the used malware afterwards. This leads to the following equations system of the basic SIR-model:
\begin{equation}\label{sirequ}
\begin{array}{rccccl} 
	\displaystyle\frac{\dd S}{\dd t} &=& -&\beta S(t)I(t)\\[2mm]
	\displaystyle\frac{\dd I}{\dd t} &=& &\beta S(t)I(t)&-&\bar \gamma I(t)\\[2mm]
	\displaystyle\frac{\dd R}{\dd t} &=& &&&\bar \gamma I(t)
\end{array}
\end{equation}
Concerning initial conditions, it holds $S(0),I(0),R(0)\ge 0$. 
In the proposed framework, refinements
\cite{del2015mathematical,peng2013smartphone,sneha2015survey} 
of the SIR model such as the SEIR model \cite{wang2013seir} may be used. 
The large potential variety of cyber components corresponds to 
the large spectrum of possible infection and propagation mechanisms
\cite{schramm2013lanchester}. If required, the cyber model component 
may be provided for each weapon system type individually;
this option will not be elaborated further, though.

The general form of the cyber component is described as a
compartmental system. Each compartment $C_g$, $g=1,\ldots,m$
represents a specific malware infection state of a weapon system.
We identify $C_g$ with the number of force elements being in the
state represented by $C_g$. These numbers are changing due to flows
$\phi_{gl}(C)$ from $C_g$ to $C_l$ dependent on $C:=(C_1,\ldots,C_m)$.
The flows portray malware propagation as well as malware removal 
and other countermeasures for fighting the malware epidemics.
The dynamics is thus given by a system
\begin{equation}\label{compsysequ}
\displaystyle\frac{\dd C_l}{\dd t} = \sum\limits_{g=1}^m \phi_{gl}(C) 
\end{equation}
of differential equations with $C_g(t)\ge 0$ for all times $t\ge 0$
and the values of $C_g(0)$ \cite{sandberg1978mathematical} as
initial conditions.
For being well-defined, we assume that all flows $\phi_{gl}$
are bounded. 
Of course, not every term that meets the given technical constraints will
describe a malware propagation process in a realistic way. The constraints
are necessary, however, to ensure the proper function of the framework. 
Antisymmetry $\phi_{gl}=-\phi_{lg}$ assures that the outflow 
of the stock of force 
elements in infection state $g$ given by $\phi_{gl}$
equals the inflow of the stock of elements in state $l$.
This means that the compartmental system~(\ref{compsysequ}) is closed, 
i.e. that the population $\sum_g C_g(t)$ summed up over all states $g$ 
is constant over time $t$.
For preserving this property after integrating kinetic and cyber 
submodels, i.e. in presence of an annihilation
of force elements due to kinetic combat, we assume normalized 
compartment levels $C_g(t) =\tilde C_g(t)/(\sum_g C_g(t))$.
This assures $\sum_g C_g(t)=1$. Correspondingly, flow terms describing
malware propagation processes have to preserve normalization, whereas
for flows related to other processes the absolute numbers
of force elements are of interest. Accordingly, a normalization
has to be omitted for the latter. The example of a kinetic-cyber 
model presented in section~\ref{secexamplemodel} demonstrates 
an application of the normalization rule.

\subsection{The Coupling of the Model Components}

In the overall model, kinetic and cyber component are coupled with
each other for taking the degradation of kinetic capabilities by
the propagating malware into account.
A malware infection may lead both to a reduced effectiveness
of attacking force elements \cite{schramm2013lanchester,yildiz2014modeling} 
and to a higher loss rate of defending force elements. For representing
these effects, additional indices $k,l$ are introduced as malware
infection states of attacker and defender.
The kinetic effectivenesses $\delta_{ikjl}$ have to be chosen correspondingly.
This leads to a Lanchester equation system with extended parameterization.
Concerning the levels $b_{jl}$ and $r_{jl}$ of the force elements
of weapon system type $j$, the infection state $l$ is indicated
as well.
Accordingly, the kinetic interaction term for Blue has the form
\begin{equation}
\mathcal{K}_{b,jl} := -\sum\limits_{i=1}^n \sum\limits_{k=1}^m  \delta_{r,ikjl} \cdot r_{ik}^{p_{r,ikjl} }(t) \cdot b_{jl}^{q_{r,ikjl} }(t).
\end{equation}
The corresponding term for Red is designated as $\mathcal{K}_{r,jl}$.
We define the numbers of existing force elements according to
\begin{equation}\label{definitionn}
	N_{b,j}(t) := \sum\limits_{l=1}^m  b_{jl}(t) \,\,\,\text{and}\,\,\,
	N_{r,j}(t) := \sum\limits_{l=1}^m  r_{jl}(t)
\end{equation}
We also define an overall force size
$N_{b/r}(t):= \sum\limits_{j=1}^n \zeta_j N_{b/r,j}(t)$ with $\zeta_j> 0$ 
as positive summation weights. Using the weights, different
levels of importance can be distinguished; an aircraft carrier,
for example, contributes more to the overall force than a small
patrol boat.

For including the evolving malware propagation (and the fight against 
it) in the equation system (\ref{heterogeneouslanchester}),
flow terms $\phi_{b/r,glj}$ are added representing the transitions
between the malware infection states $g$ and $l$ for the 
specific weapon system type $j$. As for equation~(\ref{compsysequ}),
we assume that all flow terms $\phi_{b/r,glj}$ are bounded
\label{assumptbound} for reasons of well-definedness. 
The malware propagation terms are, thus,
\begin{equation}
\mathcal{P}_{b,jl} := \sum\limits_{g=1}^m \phi_{b,glj}(b_{j1},\ldots,b_{jm})
\end{equation}
for Blue and $\mathcal{P}_{b,jl}$ for Red, respectively.

The initial malware infection starting the malware epidemic
is caused by a malware attack of the adversarial force. 
The effect of a malware attack triggered by red force elements
of type $i$ in malware infection state $k$ on blue force elements
of type $j$ in state $g$ is represented by a flow 
\begin{equation}
\mathcal{A}_{b,jl} := \sum\limits_{i=1}^n \sum\limits_{k=1}^m  \sum\limits_{g=1}^m  \alpha_{r,ikjgl} (r_{ik},b_{j1},\ldots,b_{jm})
\end{equation}
for Blue and, analogously, by $\mathcal{A}_{r,jl}$ for Red.
The blue force elements being attacked are subject to a 
state transition $g\rightarrow l$. Correspondingly,
an antisymmetry condition $\alpha_{b/r,ikjgl} = - \alpha_{b/r,ikjlg}$
holds in analogy to $\phi$. 
Again, we assume that all flow terms 
$\alpha_{b/r,ikjgl}(r_{ik},b_{j1},\ldots,b_{jm})$ are bounded
for reasons of well-definedness. 
The precise shape of $\alpha$ depends on the infection mechanism
used by the malware. Examples are given in \cite{schramm2013lanchester}.
The attack rate $\alpha_{b/r,ikjgl}$ is equal to zero,
if the malware infection state $k$ prohibits a malware attack 
on enemy force elements or if the weapon system type $i$ has
no malware attack capability. 
Taken together, this gives the following equation system. 
\begin{equation}\label{intermediatestep3}
\begin{array}{rl}
	\frac{\dd b_{jl}}{\dd t} =& \mathcal{K}_{b,jl} + 
	 \mathcal{P}_{b,jl} +
	 \mathcal{A}_{b,jl}
\\[1mm]
	\frac{\dd r_{jl}}{\dd t} =& \mathcal{K}_{r,jl} + 
	 \mathcal{P}_{r,jl} +
	 \mathcal{A}_{r,jl}
\end{array}
\end{equation}
We assume non-negativity $b_{jl}(0),r_{jl}(0)\ge 0$ of all
compartments as initial condition.
The equation system can be extended by the annihilation losses 
$ D_{j}'(t)= - \sum_l \mathcal{K}_{jl}(t)$ for bookkeeping purposes 
with $D_{j}(0)\ge 0$ as initial condition. 
Additionally, we define $D(t):= \sum_j \zeta_j D_{j}(t)$.

\subsection{Adding an Event Mechanism} 

We permit a finite number of timed events modifying parameter values 
of equation system~(\ref{intermediatestep3}).
This extension 
completes the construction of the underlying model class $\mathcal{M}$.
In general, the event mechanism can be used for representing various kinds of 
singular events and for a temporal structuring of the course of action 
\cite{stodd2017black,burnsimproving,ghanea2012cyber,wangen2015quantitative}.
An exemplary temporal structuring may consist, say, of a preparation, an 
assault, and an exploitation phase. In these phases, different force
elements will have prominent roles. Accordingly, the values of the 
attrition coefficients will usually differ in the individual phases 
\cite{hartley1995validating,hughes1995two}.
We will make use of events in section~\ref{modextlab} for
choosing individual start times of the kinetic combat,
of the malware attacks, and of the patching actions. 

A typical example of a singular event is a change of the force elements
opposing each other due to troop movements \cite{ancker1987validity}. 
Such re-assignments can be represented by events activating and
deactivating attrition between units. Indeed, 
the Lanchester model is only a model of attrition, but not a full model
of combat \cite{sfikas2017model}. 
It is the event mechanism, which extends the component $M_K$ 
from a model of attrition to a general model of kinetic combat.
Though several important aspects of conflict 
have to be taken into account outside of (\ref{lanchmalware}), their 
effects can usually be incorporated in our model by events as additional input. 
These aspects include not only rather elementary activities like 
troop movements, but
the decisions of the involved military leaders in general. We let the 
user of the framework experiment with them based on corresponding what-if 
assumptions. 

\subsection{Model Consistency}\label{modconslab}


The kinetic interaction terms $\delta_{r,ikjl} \cdot r_{ik}^{p_{r,ikjl} }(t) \cdot b_{jl}^{q_{r,ikjl} }(t)$ and $\delta_{b,ikjl} \cdot b_{ik}^{p_{b,ikjl} }(t) \cdot r_{jl}^{q_{b,ikjl} }(t)$ in (\ref{intermediatestep3}) require further discussion. 
In order to improve readability, the indices $i,j,k,l$ are omitted in this section in the following.
For $q_{r} \rightarrow 0$, one gets
$b^{q_{r} }(t)\rightarrow 1$ independently of the value of 
$b$. This means, the losses inflicted on $b(t)$
by Red become independent of the current size of $b(t)$.
For $b$ close to $0$, the level of $b$ 
thus becomes negative. Similarly, for $p_{r}\rightarrow 0$
the term $r^{p_{r} }$ becomes approximately
equal to $1$ independently of the value of $r$.
Thus, a red force already annihilated can inflict arbitrarily 
large losses on Blue. For $p_{r}=1$, $q_{r}=1$,
which is the special case usually discussed in literature,
such irregularities do not occur.

For avoiding unphysical behavior, a barrier function 
$f(x)=(\exp(-1/x))^{0.0001}$ is introduced with $f(x)\rightarrow 0$ 
for $x\rightarrow 0$ and $f(x) \approx 1$ otherwise.
Replacing the original term
\begin{equation}\label{originteract}
\delta_{r} \cdot r^{p_{r} }(t) \cdot b^{q_{r} }(t)
\end{equation}
by 
\begin{equation}\label{newinteract}
\delta_{r} \cdot f(r) \cdot r^{p_{r} }(t) \cdot f(b) \cdot b^{q_{r} }(t)
\end{equation}
and analogously for the other kinetic interaction term in the 
equation system~(\ref{intermediatestep3}) avoids the occurrence of 
compartments with negative levels but preserves the system dynamics 
for $b >0$ approximately. For avoiding bulky expressions, we will use
the designation $\overline{x^a} := f(x)x^a$. This simplifies
term~(\ref{newinteract}) to 
\begin{equation}\label{modinteract}
\delta_{r} \cdot \overline{r^{p_{r} }(t)} \cdot \overline{b^{q_{r} }(t)}.
\end{equation}
For simplifying the notation, the temporal derivative of 
e.g. a compartment level $C(t)$ is designated as $C'(t)$
in the future.

\begin{proposition}[Consistency of the Modified Equation System]
\label{veryfirstpropos}
~
\begin{enumerate}[label=\emph{\alph*}),nosep,leftmargin=*]
\item For $b \rightarrow 0$ or $r\rightarrow 0$ it holds
$\delta_{r} \cdot \overline{r^{p_{r} }(t)} \cdot \overline{b^{q_{r} }(t)} \rightarrow 0$ and
$$\delta_{r} \cdot \overline{r^{p_{r} }(t)} \cdot \overline{b^{q_{r} }(t)} \approx \delta_{r} \cdot r^{p_{r} }(t) \cdot b^{q_{r} }(t)$$ otherwise.
		An analogous statement is valid for the corresponding interaction term $\delta_{b} \cdot \overline{b^{p_{b} }(t)} \cdot \overline{r^{q_{b} }(t)}$ for Red.
\item \label{labnoinfluence}
$N_{b} \rightarrow 0 \Rightarrow N_{b}'(t), N_{r}'(t) \rightarrow 0$ and
$N_{r} \rightarrow 0 \Rightarrow N_{b}'(t), N_{r}'(t) \rightarrow 0$.
\end{enumerate}
\end{proposition}

\begin{proof}[Sketch of Proof]
~
\begin{enumerate}[label=\emph{\alph*}),nosep,leftmargin=*]
\item According to the definition of the operator $\overline{\phantom{a}}$,
it holds $\overline{b^{q_{r} }(t)} \rightarrow 0$
for $b \rightarrow 0$ and $b^{q_{r}}(t)\approx 
\overline{b^{q_{r} }(t)} $ otherwise.
In the first case, this is caused by $f(b) \rightarrow 0$ for the 
additional factor $f(b)$ occurring in 
$\overline{b^{q_{r} }(t)} $; in the second case, this factor 
can be neglected due to $f(b) \approx 1$.
A corresponding statement is valid for $r^{p_{r} }(t)$. 
The claim is an immediate consequence.
\item The condition $N_b \rightarrow 0$ leads immediately to
$b\rightarrow 0$ due to the non-negativity of all
force element type levels. Thus it holds
$\delta_{b} \cdot \overline{b^{p_{b} }(t)} \cdot \overline{r^{q_{b} }(t)} \rightarrow 0$
and
$\delta_{r} \cdot \overline{r^{p_{r} }(t)} \cdot \overline{b^{q_{r} }(t)} \rightarrow 0$; 
this results from part a). Since no other term can cause annihilation losses ---
		$\phi_{glj}$ and $\alpha_{ikjgl}$ 
are just modifying the malware states of force elements ---
this gives $N_{b}'(t), N_{r}'(t) \rightarrow 0$, which is the claim.
\end{enumerate}
\end{proof}

One has to note that a vanishing dynamics $N_{b}'(t), N_{r}'(t) \rightarrow 0$
may not only be caused by a vanishing force size $N_{b}$ or $N_{r}$,
but also by vanishing kinetic effectiveness $\delta_{b/r}$.

If we speak about the equation system (\ref{intermediatestep3}) in the
following, we always mean the system modified by the 
barrier function $f(x)$ as described above. 
The preceding proposition shows that the introduction of the
barrier function $f(x)$ excludes the two types of unphysical behavior 
mentioned at the beginning of this section.
%
Such negativities can also be generated at the computational level
due to overshooting effects, if the size of the time steps is not
small enough. A less extreme barrier function $f(x)$ ---
i.e. a barrier function with a smaller exponent --- may 
avoid negativities more reliably, but its influence on the
dynamics of the system will be stronger. 

\subsection{Model Properties} \label{secmodelextensions}

We derive several properties, which assure applicability and validity
of the intended risk assessment procedure.

\begin{proposition}[Properties of the Equation System]\label{sympropos}
~
\begin{enumerate}[label=\emph{\alph*}),nosep,leftmargin=*]
\item \label{proposfirst}
The equation system~(\ref{intermediatestep3})
is symmetric w.r.t. Blue and Red, i.e. an
interchange of the initial conditions and parameters for Blue and Red
corresponds to an interchange of the dynamics of the compartment
levels $b_{jl}(t)$, $D_{b,j}(t)$ and $r_{jl}(t)$, $D_{r,j}(t)$.
\item \label{propos_ndconst} 
The quantities $N_j(t) + D_j(t)$, $j=1,\ldots,n$ and $N(t) + D(t)$ are preserved over time. For an initial condition $D_j(0)=0$ resp. $D(0)=0$, one gets $N_j(0) = N_j(t) + D_j(t)$ resp. $N(0) = N(t) + D(t)$ for all times $t$.
\end{enumerate}
\end{proposition}

\begin{proof}[Sketch of Proof]
~
\begin{enumerate}[label=\emph{\alph*}),nosep,leftmargin=*]
\item According to the structure of (\ref{intermediatestep3}).
\item Combining (\ref{intermediatestep3}) and the definition of 
$D_{j}'(t)$ leads to $ N_j'(t) + D_j'(t) = 0 $,
because all terms in the equation system cancel out.
In this respect, one has to take the antisymmetry of $\phi_{b/r,glj}$ and $\alpha_{b/r,ikjgl}$ into account;
since we sum up in $N_j+D_j$ resp.\ $N+D$ over both indices $g,l$ of 
$\phi_{glj}$ and $\alpha_{ikjgl}$
subject to the antisymmetry constraint, it holds
$\textstyle \sum_l \left( \sum_g \phi_{glj} \right) =0 $ and
$\textstyle \sum_l \left( \sum_{ikg} \alpha_{ikjgl} \right) =0.$
Thus, $N_j(t) + D_j(t)$ is constant over time.
The statements concerning $N(t) + D(t)$ and $N(0) = N(t) + D(t)$ are an immediate consequence. 
\end{enumerate}
\end{proof}

With help of proposition \ref{sympropos}.\ref{propos_ndconst}, 
the maximum possible annihilation losses are known.
Another important aspect is the existence of a well-defined 
'final' outcome. As soon as the simulation of the evolution 
of a given scenarion has estimated the final outcome 
with sufficient accuracy, it may stop (see section~\ref{secsimhor}). 
In reality, a battle does of course not necessarily stop with such 
final outcome, which usually corresponds to the annihilation of one 
of the forces. If a force is experiencing
many losses when compared with the enemy, it might for example to
decide to break off combat and retreat. 
Since we are primarily interested in risk assessments and other
evaluations and not in predictions of events to be expected, 
these additional options are considered as negligible.
The definition of the final outcome is based on the following proposition.

\begin{proposition}[Dynamics at Infinity] \label{dynamicsatinf}
It holds
$ 
\lim\limits_{t\rightarrow \infty} N_j'(t) = 0 $ and $
\lim\limits_{t\rightarrow \infty} D_j'(t) = 0
$
\end{proposition}

\begin{proof}[Sketch of Proof]
Due to the annihilation losses given by the terms $\delta_{r,ikjl} \cdot r_{ik}^{p_{r,ikjl} }(t) \cdot b_{jl}^{q_{r,ikjl} }(t)$ and $\delta_{b,ikjl} \cdot b_{ik}^{p_{b,ikjl} }(t) \cdot r_{jl}^{q_{b,ikjl} }(t)$ in the equation system~(\ref{intermediatestep3}), $N_j(t)$ is monotonically decreasing. In this respect it is important to note that the other terms $\phi_{glj}$ and 
$\alpha_{ikjgl}$ 
in (\ref{intermediatestep3}) are just modifying the malware
states of force elements, but they do not change $N_j(t)$. 
According to $N_j''(t) < 0$ --- derivable directly from (\ref{intermediatestep3}) --- the decrease of $N_j'(t)$ is monotonic as well; thus, an already small flow rate can not increase again. 
Due to the monotonic decrease of $N_j(t)$ and $N_j'(t)$ on the one hand 
and the limitation given by $N_j(t)\ge 0$ on the other, the decrease 
must be fading out. 
This proves the claim $N_j'(t)\rightarrow 0$.
The claim $D_j'(t)\rightarrow 0$ follows immediately using proposition~\ref{sympropos}.\ref{propos_ndconst}.
\end{proof}


\section{Example: A Simple Homogeneous Model} \label{secexamplemodel}

\subsection{Setting of the Example}

We give a simple example --- continued in section~\ref{secfour}
--- how a specific kinetic-cyber
situation can be represented in a model $M\in \mathcal{M}$.
It serves for illustrative purposes and makes
the idealistic assumption of homogeneous forces. This assumption may
be approximatively fulfilled in situations like large tank battles
or combat without heavy equipment 
\cite{engel1954verification,kamaluregression,lucas2004fitting,fricker1998attrition,hartley1995validating}. 
As a consequence of the assumption of homogeneity, the dependence
of $\delta$, $p$, and $q$ on the weapon system type is abandoned.
For the Lanchester parameters $p$ and $q$, a distinction of 
individual parameter values for Blue and Red is omitted.
Accordingly, the equations (\ref{heterogeneouslanchester}) simplify to
\begin{equation} \label{lanchequref}
\begin{array}{rl}
	b'(t)  &= - \delta_r r^{p}(t) b^{q}(t)\\[1mm]
	r'(t)  &= -\delta_b b^{p}(t) r^{q}(t)
\end{array}
\end{equation}
For the representation of the malware infection, we make use of the
SIR-model given in (\ref{sirequ}). The basic model is slightly
generalized by distinguishing between patching of infected and
of vulnerable systems with rate $\tilde \gamma$ and $\gamma$.
Patching of infected systems requires an interaction with
vulnerable or patched systems. Infected systems can not patch
other systems or themselves.

When coupling a homogeneous Lanchester model and a SIR model,
the force elements have to be partitioned among the three different
malware infection states --- $S$, $I$, and $R$ --- of the SIR model. 
These compartments were supplemented by the
compartment $D$ of destroyed force elements for bookkeeping purposes.
As the result of the propagating malware, the kinetic effectiveness
$\delta$ of infected force elements is reduced by a factor
$\eta\in [0,1]$.
The product $\eta\delta$ can be interpreted as a modified kinetic
effectiveness of attacking force elements infected with malware.
It does not depend on the malware infection state of the force elements
being attacked. This means that losses are suffered independently
of their malware state, i.e. vulnerable, infected, and patched systems
are affected according to the same attrition rate.
The effectiveness is restored to its original value
after patching.
The option to distinguish between different values of $p$, $q$
in dependence of malware infection states is omitted. 

\subsection{Model Structure}

\begin{figure}[tbh!]
\begin{center}
\begin{tikzpicture}[scale=0.4, every node/.style={transform shape}]
\node [draw=blue, rounded corners, fill=blue!4, minimum width=28cm,
minimum height=11cm, line width=0.5] (b) at (10,-12.5) {};
\node [draw=red, rounded corners, fill=red!4, minimum width=28cm,
minimum height=11cm, line width=0.5] (r) at (10,2.5) {};
\node [draw=none, rounded corners, minimum width=5cm, minimum height=2cm, line width=0.5] (bx) at (-1,-17) {\Huge Blue Force};
\node [draw=none, rounded corners, minimum width=3cm, minimum height=2cm, line width=0.5] (vx) at (-1,7) {\Huge Red Force};
\node [draw=red, rounded corners, fill=red!10, minimum width=6cm,
minimum height=2cm, line width=0.5] (v) at (0,0) {\LARGE Vulnerable
$S_r$};
\node [draw=red, rounded corners, fill=red!10, minimum width=6cm,
minimum height=2cm, line width=0.5] (i) at (10,0) {\LARGE Infected
$I_r$};
\node [draw=red, rounded corners, fill=red!10, minimum width=6cm,
minimum height=2cm, line width=0.5] (p) at (20,0) {\LARGE Patched
$R_r$};
\node [draw=red, rounded corners, fill=red!10, minimum width=6cm,
minimum height=2cm, line width=0.5] (d) at (10,6) {\LARGE Destroyed
$D_r$};
\node [draw=blue, rounded corners, fill=blue!10, minimum width=6cm,
minimum height=2cm, line width=0.5] (vb) at (0,-10) {\LARGE Vulnerable
$S_b$};
\node [draw=blue, rounded corners, fill=blue!10, minimum width=6cm,
minimum height=2cm, line width=0.5] (ib) at (10,-10) {\LARGE Infected
$I_b$};
\node [draw=blue, rounded corners, fill=blue!10, minimum width=6cm,
minimum height=2cm, line width=0.5] (pb) at (20,-10) {\LARGE Patched
$R_b$};
\node [draw=blue, rounded corners, fill=blue!10, minimum width=6cm,
minimum height=2cm, line width=0.5] (db) at (10,-16) {\LARGE Destroyed
$D_b$};
\draw [-{Latex[scale=1.0]}, draw=black, dashed, line width=0.9] (vb)
-- node[above] {\LARGE $\beta_r$} (ib);
\draw [-{Latex[scale=1.0]}, draw=black, dashed, line width=0.9] (ib)
-- node[above] {\LARGE $\tilde \gamma_r$} node[above=35pt] {\LARGE $\gamma_r$} (pb);
\draw [-{Latex[scale=1.0]}, draw=black, dashed, line width=0.9] (vb)
-- ++(0,2.2cm) -| (pb);
\draw [-{Latex[scale=1.0]}, draw=black, line width=0.9] (vb) --
node[left=10pt] {\LARGE $p,q$} (db);
\draw [-{Latex[scale=1.0]}, draw=black, line width=0.9] (ib) --
node[right=5pt] {\LARGE $p,q$} (db);
\draw [-{Latex[scale=1.0]}, draw=black, line width=0.9] (pb) --
node[right=15pt] {\LARGE $p,q$} (db);
\draw [-{Latex[scale=1.0]}, draw=black, dashed, line width=0.9] (v) -- 
node[above] {\LARGE $\beta_b$} (i);
\draw [-{Latex[scale=1.0]}, draw=black, dashed, line width=0.9] (i) --
node[above] {\LARGE $\tilde \gamma_b$} node[below=35pt] {\LARGE $\gamma_b$} (p);
\draw [-{Latex[scale=1.0]}, draw=black, dashed, line width=0.9] (v) -- 
++(0,-2.2cm) -| (p);
\draw [-{Latex[scale=1.0]}, draw=black, line width=0.9] (v) --
node[left=10pt] {\LARGE $p,q$} (d);
\draw [-{Latex[scale=1.0]}, draw=black, line width=0.9] (i) --
node[right=5pt] {\LARGE $p,q$} (d);
\draw [-{Latex[scale=1.0]}, draw=black, line width=0.9] (p) --
node[right=15pt] {\LARGE $p,q$} (d);
%
\draw [draw=black, line width=0.9,<->] (7,-7) -- node[left]
{\LARGE $\delta_{b},\delta_{r}$} node[right]
{\LARGE $\eta_{b},\eta_{r}$} (7,-3);
\draw [draw=black, dashed, line width=0.9,<->] (13,-7) -- node[right]
{\LARGE $\alpha_b,\alpha_r$} (13,-3);
\draw [draw=black, line width=0.9] (18,-6) -- (21,-6);
\draw [draw=black, dashed, line width=0.9] (18,-5) -- (21,-5);
\node [draw=none, right] at (21.3,-6) {\LARGE Kinetic Effects};
\node [draw=none, right] at (21.3,-5) {\LARGE Malware Effects};
\end{tikzpicture}
\end{center}
\caption{
Red force and blue force fighting against each other on two levels: 
kinetic and cyber}
\label{compartmental_model}
\end{figure}

The setting of the example leads to the model structure 
in figure~\ref{compartmental_model}. 
We will now quantify the flows between the compartments belonging to 
the blue force; the flows for Red result from the intended symmetry 
of the model. We will use the designation
$N=S+I+R$ for the still existing force elements.
\begin{itemize}
\item $\beta_b S_b I_b / N_b$ is the 
flow from $S_b$ to $I_b$ representing the malware infection process.
\item $\gamma_b S_b$ is the flow from $S_b$ to $R_b$ due to
	the patching of vulnerable (i.e. non-infected) elements.
\item $\tilde \gamma_b I_b (S_b + R_b) / N_b $ is the flow from $I_b$ 
to $R_b$ representing a patch of the vulnerability for infected force
elements. The removal of the malware infection is considered
as part of the patching process.
\end{itemize}


Now, the flows related to kinetic combat are discussed. Following the
structure of the equation system~(\ref{lanchequref}), the attrition
rate of Red on Blue is $\delta_{r} (S_r+R_r+\eta_{r}I_r)^p$.
The factor $\eta_{r}$ in this term describes the reduced
effectiveness of infected force elements.
The affected elements of Blue are $S^q$, $I^q$, or $R^q$.

\begin{itemize}\itemsep0pt
\item $\delta_{r}(S_r^p+R_r^p+\eta_{r}I_r^p) \cdot S_b^{q}$ is the
flow from $S_b$ to $D_b$ due to kinetic combat losses.
\item $\delta_{r}(S_r^p+R_r^p+\eta_{r}I_r^p) \cdot I_b^{q}$ is the
flow from $I_b$ to $D_b$ due to kinetic combat losses.
\item $\delta_{r}(S_r^p+R_r^p+\eta_{r}I_r^p) \cdot R_b^{q}$ is the
flow from $R_b$ to $D_b$ due to kinetic combat losses.
\item $\alpha_r (S_r + R_r) \cdot S_b$ is the flow from $S_b$ 
to $I_b$ due to a malware attack of Red. It should be noted, 
that $\alpha$ is here used as a scalar coefficient instead 
of a flow function as in (\ref{intermediatestep3}).
\end{itemize}
In the end, we get four equations for each of the two forces.
We only give the equations
for Blue, since the equations for Red have symmetric form.
Additional remarks about
the basic structure of such an equation systems can be found in
\cite{schramm2013lanchester,yildiz2014modeling}. 
\begin{align}\label{lanchmalware}
\begin{split}
	S_b'(t) =& -\beta_b S_b I_b / N_b - \gamma_b S_b \\
	&\quad - \delta_{r}(S_r^p+R_r^p+\eta_{r}I_r^p) \cdot
						S_b^{q}
						\\
        &\quad - \alpha_r (S_r + R_r) \cdot S_b\\
	I_b'(t) =& \quad\beta_b S_b I_b / N_b - \tilde \gamma_b I_b (S_b + R_b) / N_b \\
	&\quad - \delta_{r}(S_r^p+R_r^p+\eta_{r}I_r^p) \cdot
						I_b^{q}
						\\
        &\quad + \alpha_r (S_r + R_r) \cdot S_b\\
	R_b'(t) =&\quad \gamma_b S_b + \tilde \gamma_b I_b (S_b + R_b) / N_b \\
	&\quad - \delta_{r}(S_r^p+R_r^p+\eta_{r}I_r^p) \cdot
						R_b^{q}
						\\
	D_b'(t) =&\quad
		\delta_{r}(S_r^p+R_r^p+\eta_{r}I_r^p) 
		\cdot ( S_b^{q} + I_b^{q} + R_b^{q} ) \\
\end{split}
\end{align}

\begin{figure}[htb!]
\centering
\includegraphics[width = 1\textwidth]{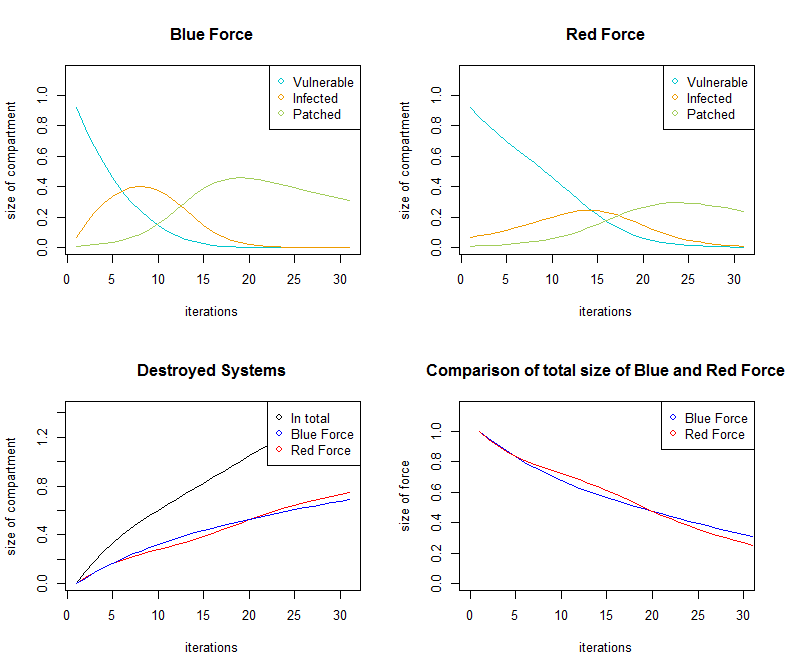}
\caption{
A typical behavior of the equation system (\ref{lanchmalware}).
Besides of the number of vulnerable, infected, and patched elements
for both forces, the overall numbers of available and destroyed
force elements are shown. 
Due to malware attacks applied by both forces,
the red force is stronger at the beginning but
will still lose in the end. 
}
\end{figure}

\subsection{Model Parameters}\label{modextlab}

\begin{table}[tb!]
\setlength{\tabcolsep}{3pt}
\tabulinesep=1.5mm
\centering
\begin{tabu} to \textwidth{| L{2.6cm} |  L{0.8cm}  X[1]  L{1.8cm} |}
	\hline
	\textbf{Parameter Group} & \multicolumn{2}{l}{\textbf{Parameter
with Description}} & \textbf{Parameter Range} \\
           \hline \hline
           \multirow{1}{*}{\parbox[c][1.0cm][c]{2cm}{{\textbf{Kinetic
	Combat}}}}  & $\delta_{}$ &  Kinetic effectiveness of vulnerable and
patched elements & $\delta\geq 0$  \\
          \cline{2-4}
	  &  $\eta_{}$ & Effectiveness reduction of infected
elements & $\eta\in [0,1]$\\
          \cline{2-4}
          & $p$ & First Lanchester parameter & $p\ge 0$ \\
          \cline{2-4}
          & $q$ & Second Lanchester parameter & $q\ge 0$ \\
          \hline \hline
           \multirow{1}{*}{\parbox[c][0.7cm][c]{2.6cm}{{\textbf{Cyber Combat}}}}& $\alpha$ & Malware attack rate & $\alpha\geq 0$ \\
	\hline \hline
           \multirow{1}{*}{\parbox[c][1.0cm][c]{2cm}{{\textbf{Malware Spreading}}}} & $\beta$ & Infection rate & $\beta\geq 0$ \\
           \cline{2-4}
           & $\gamma$ & Patch rate for vulnerable elements & $\gamma\in [0,1]$ \\
           \cline{2-4}
           & $\tilde \gamma$ & Patch rate for infected elements & $\tilde \gamma\in [0,1]$ \\
	\hline\hline
           \multirow{1}{*}{\parbox[c][1.0cm][c]{2cm}{{\textbf{Timing}}}}
	   & $\Delta t_{\mathsf{att}} $ & Time difference $\Delta t_{\mathsf{att}} =t_{\mathsf{att}} -
	   t_{\mathsf{kin}}$ between the start $t_{\mathsf{att}}$ of 
	   the malware attack and the start $t_{\mathsf{kin}}$ of 
	   kinetic combat.
 & $\Delta t_{\mathsf{att}} \in \mathbb{R}$ \\
           \cline{2-4}
	   & $\Delta t_{\mathsf{mal}}$ & Duration of malware attack & $\Delta t_{\mathsf{mal}}\geq 0$  \\
           \cline{2-4}
	   & $\Delta t_{\mathsf{pat}}$ & Time difference $\Delta t_{\mathsf{pat}}
	   = t_{\mathsf{pat}} - t_{\mathsf{kin}}$ between the start
	   $t_{\mathsf{pat}}$ of the patching process (for both vulnerable and infected elements) and the start
	   $t_{\mathsf{kin}}$ of the kinetic combat & $\Delta t_{\mathsf{pat}}\in \mathbb{R}$ \\
	\hline
\end{tabu}
\caption{Parameters occurring in the equation system
(\ref{lanchmalware}).
\label{definition_param}
}
\end{table}

Model parameters are listed in table~\ref{definition_param}. 
Except of $p,q$, the
given parameters can be chosen independently for Blue and Red.
As initial conditions, the non-negativity $S(0),I(0),R(0),D(0)\ge 0$ of all
compartments is required.
Compared to the models 
of Schramm \cite{schramm2013lanchester} and Yildiz \cite{yildiz2014modeling}, 
the example introduces some generalizations. 
The generalization of the Lanchester exponents $p,q$ from $p=1,q=1$ as in 
\cite{schramm2013lanchester,yildiz2014modeling,yildizmodeling}
to $p,q\in\real_0^+$ gives a significantly improved description of
many historic battles \cite{engel1954verification,lucas2004fitting}.
Concerning the SIR component,
we distinguish the patching rates $\gamma,\tilde \gamma$ of vulnerable and 
already infected force
elements, because we assume that the removal of a malware infection 
takes additional time. Furthermore, the intensity of the malware attack 
is parameterized by $\alpha$.



Table~\ref{definition_param} includes event time parameters, which
extend the model by allowing an onset of kinetic combat, 
malware attack and patching process independent from each other. In this 
way, situations like a preparation of a kinetic battle by a 
supporting malware attack in advance or a late start of countermeasures due to 
a delayed provision of appropriate patches can be modeled.
Usually, the force under malware attack will also need some time to recognize the attack and to initiate corresponding countermeasures. 
Additionally, the duration of the malware attack can be determined by 
$\Delta t_{\mathsf{mal}}$. An early stop of the malware attack can sometimes be 
useful for covering up the source of the malware attack.

Before the start and after the end of an action like malware 
attack, the execution of the action is inhibited by appropriate 
parameter settings. In between, the action is activated by setting
its parameters to their effective values.
Details of the translation of the timing parameters given in 
table~\ref{definition_param} to corresponding event-based changes of the
parameter settings are given in table~\ref{definition_event}. 
Figure~\ref{timeeventsandintervals}
depicts events and time intervals used in the text graphically.

\begin{table}[tbh]
\tabulinesep=0.75mm
\centering
\begin{tabu} to \textwidth{| X[1.3] | X[1] |}
	\hline
	\textbf{Events} & \textbf{Associated Parameter Changes} \\
	\hline \hline
	Initial settings & $\alpha,\delta,\gamma,\tilde \gamma=0$ \\
	Start of kinetic combat at time
	$t_{\mathsf{kin}}$ & $\delta$ set to (effective) parameter value \\
	Start of malware attack at time
	$t_{\mathsf{att}} = t_{\mathsf{kin}} + \Delta t_{\mathsf{att}} $ & $\alpha$ set to
	(effective) parameter value \\
	Stop of malware attack at time
	$t_{\mathsf{att}}+\Delta t_{\mathsf{mal}}$ & $\alpha=0$ \\
	Start of patching process at time
	$t_{\mathsf{pat}} = t_{\mathsf{kin}} + \Delta t_{\mathsf{pat}}$ & 
	$\gamma,\tilde \gamma$ set to (effective) parameter values \\
        \hline
\end{tabu}
\caption{Translation of the events associated with the timing parameters 
given in table~\ref{definition_param} to corresponding changes of the
parameter settings. 
\label{definition_event}
}
\end{table}

\begin{figure}[t!]
\begin{center}
\begin{tikzpicture}[
    thick,
    >=stealth',
    dot/.style = {
      draw,
      fill = white,
      circle,
      inner sep = 0pt,
      minimum size = 4pt
    }
  ]
  \coordinate (O) at (0,0);
  \draw[->] (-0.35,0) coordinate (xmin) -- (8,0) coordinate[label = {below right:$t$}] (xmax);
  \draw (0,-0.2) -- (0,0.2) coordinate[pos=0.4,label = {below right:$0$}]  (ymax);
  \draw (1.5,-0.2) -- (1.5,0.2) coordinate[pos=0.4,label = {below right:$t_{\mathsf{kin}}$}]  (y1);
  \draw (2.5,-0.2) -- (2.5,0.2) coordinate[pos=0.4,label = {below right:$t_{\mathsf{att}}$}]  (y2);
  \draw (3.3,-0.2) -- (3.3,0.2) coordinate[pos=0.4,label = {below right:$t_{\mathsf{pat}}$}]  (y3);
  \draw (4.2,-0.2) -- (4.2,0.2) coordinate[pos=0.4,label = {below right:$t_{\mathsf{att\#end}}$}]  (y4);
  \draw (7.0,-0.2) -- (7.0,0.2) coordinate[pos=0.4,label = {below right:$t_{\mathsf{end}}$}]  (y6);
\draw[->]      (2.5,0.5) -- (4.2,0.5) node[pos=0.5, above] {$\Delta t_{\mathsf{mal}}$};
\draw[<-]      (1.5,0.5) -- (2.5,0.5) node[pos=0.5, above] {$\Delta t_{\mathsf{att}}$};
\draw[<-]      (1.5,-0.6) -- (3.3,-0.6) node[pos=0.5, below] {$\Delta t_{\mathsf{pat}}$};
  \draw (2.5,0.4) -- (2.5,0.6) coordinate  (yx);
\draw[dotted]      (1.5,0.2) -- (1.5,1.0) node {};
\draw[dotted]      (2.5,0.2) -- (2.5,1.0) node {};
\draw[dotted]      (4.2,0.2) -- (4.2,1.0) node {};
\draw[dotted]      (1.5,-0.2) -- (1.5,-1.0) node {};
\draw[dotted]      (3.3,-0.2) -- (3.3,-1.0) node {};
\end{tikzpicture}
\end{center}
\caption{Designations of events and time intervals as used in the text.\label{timeeventsandintervals} }
\end{figure}



\subsection{Model Validation}

%
%
For demonstrating the validity of the model (\ref{lanchmalware}), we
have to show that the model behavior is sufficiently similar to the 
real system. Unfortunately, corresponding real-world data sets are missing
or at least not accessible to the public. We thus show the 
validity of the two main components of the overall model 
by applying a cross-validation with basic models of the pure Lanchester 
and SIR equations.
An explicit validation of the main components of the overall model
seems to be necessary; despite of the already provided arguments 
supporting the validity of the two main components when being 
considered in isolation, additional modes of dynamics may occur in 
the integrated system
(\ref{lanchmalware}) compared to isolated Lanchester or SIR models. 
This validation strategy 
is applicable to more general heterogeneous Lanchester models and 
extensions of SIR models as well.
After showing the validity of the two main components,
the interactions between them have to be considered in a validation task of 
its own. Details of the three validation steps are given in the following.

\paragraph{\rm\em Validation of the Lanchester Component:}
Concerning the Lanchester part of (\ref{lanchmalware}), the authors refer to
\cite{engel1954verification,kamaluregression,lucas2004fitting,fricker1998attrition,hartley1995validating} regarding the usability of the Lanchester model as basic reference model.
For executing the comparison, all model parameters are set 
to zero resp.\ neutral values with exception of 
$\delta_{b},\delta_{r},p,q$. With respect to the effectiveness 
reduction of force elements due to malware infections, 
the values $\eta_{b},\eta_{r}=1$ have been chosen.

\paragraph{\rm\em Validation of the SIR Component:}
Concerning the SIR-model as a description of a malware epidemic, see \cite{schramm2013lanchester} and the references therein. 
The realism of compartmental models for representing a
malware epidemics is supported in \cite{del2015mathematical}. 
The cross-validation uses the model parameters $\beta,\gamma,\tilde \gamma$.
Other parts of the model (\ref{lanchmalware}) are switched off 
by $\delta_{b},\delta_{r}, \alpha_{b},\alpha_{r}=0$. 
This leads to the following comparison model.
\begin{align}
\begin{split}
	S_b'(t) =& -\beta_b S_b I_b / N_b - \gamma_b S_b \\
	I_b'(t) =& \quad\beta_b S_b I_b / N_b - \tilde \gamma_b I_b (S_b + R_b) / N_b \\
	R_b'(t) =&\quad \gamma_b S_b + \tilde \gamma_b I_b (S_b + R_b) / N_b \\
\end{split}
\end{align}
This equation system differs from common 
SIR-models with biological origin by an unusual 'recovery' term 
$\tilde \gamma_b I_b (S_b + R_b) / N_b$. For the cross-validation, we have 
chosen the initial conditions $S_b(0) \in [0.1,1.0]$ and $I_b(0) :=
0.1 \cdot S_b$. Since the parameter settings for the cross-validation
switch off malware attacks, we have to start with a non-zero fraction of
infections.
%

\paragraph{\rm\em Validation of the Component Interactions:}
Our validation approach concerning the interactions between the 
Lanchester and the SIR component relies on confidence building.
For this purpose, plausibility arguments are checked. This includes
the reproduction of essential behavioral characteristics.
The checked properties encompass, that a larger force 
cannot give a worse outcome under --- apart from that --- equal conditions,
that a larger kinetic effectiveness of one side increases the losses of the
other side and that a combined kinetic-cyber attack cannot give a worse
outcome than a pure kinetic attack. Beyond that, the reproducibility of 
model properties formally derived from (\ref{lanchmalware}) 
like the conservation of the numbers 
$N_b(t)+D_b(t)$, $N_r(t)+D_r(t)$ of force elements and the monotonic 
decrease of existing force elements $N_b(t)$, $N_r(t)$ over time
were examined. Other such properties are provided in
\cite{schramm2013lanchester,yildiz2014modeling}.

\section{Simulation-based Risk Assessment}\label{sectwo}

\subsection{Simulation Trajectories}\label{secsimhor}

In order to assess the risk inherent to a scenario $x\in X$,
the outcome of the situation represented by $x$ is determined 
based on the equation system~(\ref{intermediatestep3})
\cite{draeger2015roadmap,draeger2017formalized}. 
For this purpose, a standard differential equation
solver of computational mathematics is applied, which approximates the
time-dependent solution of the differential equation by taking finite 
discrete steps $h$ on the time-axis. For being well-defined, the error
inherently involved in such an approximation must vanish for
$h\rightarrow 0$. This property is assured by Lipschitz continuity,
which is shown in the following proposition.

\begin{proposition}[Lipschitz continuity]
The equation system~(\ref{intermediatestep3}) is Lipschitz continuous.
\end{proposition}

\begin{proof}
Proposition \ref{sympropos}.\ref{propos_ndconst} assures the
conservation of the overall numbers $N_j(t)+D_j(t)$ of both blue and red
elements over time. Due to $N_j(t),D_j(t)\ge 0$, the numbers
$N_j(t), D_j(t)$ are thus bounded from above. This also holds
for the stock levels $b_{jl}(t),r_{jl}(t)$ comprising $N_j(t)$.
Since we have assured
that all flow terms $\phi_{glj}$ remains bounded (see
page~\pageref{assumptbound}), the expressions for the derivatives
in equation system~(\ref{intermediatestep3}) remain bounded as well.
Lipschitz continuity of the  equation system~(\ref{intermediatestep3})
is an immediate consequence.
\end{proof}

The space $X$ of scenarios is the Cartesian product of the 
domains of the input parameters and of the initial values.
It is also called {\em design space}.
Executing the simulation $\simul \colon X \rightarrow Y$ 
assigns an outcome $y\in Y$ to the input $x\in X$. The space 
$Y = \real_0^+ \rightarrow (\real_0^+)^{4n}$ 
of simulation outcomes records the dynamics of the model
(\ref{intermediatestep3})
as the time-dependent variations of the force element numbers.
Thus, an outcome $y\in Y$ consists of $N_{j}(t)$, $D_{j}(t)$ for
both Blue and Red with $t\in T=[0,\infty[$ as the simulation time.
Since for function spaces usually no {\em canonical} ordering 
'$<$' exists, it is reasonable to ask how outcomes can be compared
with each other.
Proposition~\ref{dynamicsatinf} improves the situation, because
it shows that the annihilation dynamics in
equation system (\ref{intermediatestep3}) is fading out for 
$t\rightarrow\infty$. This justifies to identify an
outcome with the values
$\lim_{t\rightarrow\infty} N_{j}(t)$, $\lim_{t\rightarrow\infty} D_{j}(t)$ 
of its constituents in the far future for both Blue and Red.
Since these values are real numbers, a canonical ordering '$<$' 
is available for comparison purposes then. 

The fading dynamics may also be used to trigger a stop 
of the simulation, if 
\begin{equation}\label{stopcriterion}
	|N_j(t) - N_j(t + \Delta \tau)| < \varepsilon \,\,\,\text{for all $j$}
\end{equation}
holds for a 'long' time period $\Delta \tau$. 
The end time of the simulation given by the 
stopping criterion is designated as $t_{\mathsf{end}}$. Since
special events like a malware attack may modify parameter values
and thus change the considered situation 
in a fundamental way, we have to assure that no such event has still 
to be processed when a stop of the simulation run is declared. 
In the example of section~\ref{secexamplemodel}, we thus start to check
criterion~(\ref{stopcriterion}) only for 
$t > t_{\mathsf{kin}}$, $t> t_{\mathsf{att}} +\Delta t_{\mathsf{mal}}$,
and $t> t_{\mathsf{pat}}$ simultaneously.
In order to work properly, $\Delta \tau$
has to be chosen sufficiently large and $\varepsilon$ sufficiently small.
Unfortunately, for each choice of $\Delta \tau$ and
of $\varepsilon$ there exist scenarios with an arbitrary slow dynamics 
leading to large approximation errors. For them, the simulation
stops early providing intermediate instead of 'final'
results. These exceptions are tolerable, as long as they are so
rare that the outcome statistics remains unaffected.

In the following proposition,
we state that an almost vanished force will not change its own
size significantly anymore and will also be unable to change the size 
of the opposing force significantly because of its almost vanished
fighting power. The proposition is a generalization of
proposition~\ref{veryfirstpropos}.\ref{labnoinfluence}.

\begin{proposition}[Effects of a Destroyed Force] \label{dynamicsatinf2}
If $N_r \rightarrow 0$ or if $N_b \rightarrow 0$, 
then $N_b', N_r', D_b', D_r' \rightarrow 0$.
\end{proposition}

\begin{proof}
W.l.o.g one can assume $N_b(t) \rightarrow 0$; otherwise exchange blue
and red side.
The claim $N_b(t) \rightarrow 0 \Rightarrow N_b'(t),N_r'(t) \rightarrow 0$
holds according to proposition~\ref{veryfirstpropos}.\ref{labnoinfluence}.
Proposition~\ref{sympropos}.\ref{propos_ndconst} states that
$N_r(t)+D_r(t)$ is constant over time; thus, $N_{b/r}'(t) \rightarrow 0$
gives $D_{b/r}'(t) \rightarrow 0$ as well.
\end{proof}

The application of the criterion for triggering the stop of the simulation 
leads to the simulation algorithm \ref{algoverlab}.


\begin{algorithm}[htb!]
\caption{Simulation Algorithm}\label{algoverlab}
\begin{algorithmic}[1]
\Procedure{Simulation}{}
\Loop
\State Calculate system state for new time step $t$
\State \textbf{exit if}
$\forall j\colon |N_j(t) - N_j(t + \Delta \tau)| < \varepsilon \wedge 
	|D_j(t) - D_j(t + \Delta \tau)| < \varepsilon $
\State $t\gets t+\delta t$
			\Comment{Transition to next Euler step}
\EndLoop
	\State $t_{\mathsf{end}} \gets t$ \Comment{End time of simulation}
\EndProcedure
\end{algorithmic}
\end{algorithm} 

\subsection{Observables}\label{evalmeas}

Survivors $N_b(t),N_r(t)$ and losses $D_b(t), D_r(t)$ are 
recording the effects of a malware infection.
Since $N_b(t),N_r(t)$ are non-negative, the 
relative number $\Delta N:= N_b(t_{\mathsf{end}}) - N_r(t_{\mathsf{end}})$ 
of surviving force elements can be used as an assessment criterion
of the final outcome. 
Analogously, the relative number $\Delta D:= D_b(t_{\mathsf{end}}) - 
D_r(t_{\mathsf{end}})$ of destroyed elements can be applied. 
If interested in absolute numbers, the cumulative losses
$L_b:= D_b(t_{\mathsf{end}})$ of the blue force may be preferred.
A list of the observables used in this paper are given in 
table~\ref{definition_outcome}. 
Hereafter, the set of 
values of an observable, say, $\Delta N$, for a set $\tilde X\subseteq
X$ of scenarios is designated as $\Delta N(\tilde X)$.
Analogously, the set of values of an input parameter, say, $\alpha_b$
occurring in $\tilde X$ is designated as $\alpha_b(\tilde X)$.

The interpretation of the relative assessment
criteria $\Delta N$ and $\Delta D$ is straightforward.  
The case $\Delta N>0$ indicates a win of 
Blue, whereas the case $\Delta N<0$ indicates a win of Red. A situation 
with $\Delta N=0$ could be judged as a draw. Analogously, $\Delta D > 0$ 
indicates an advantage for Red concerning the involved risks,
whereas $\Delta D < 0$ indicates an 
advantage for Blue. Again, $\Delta D=0$ could be judged as a draw because 
the losses of both sides have the same amount. The inclusion of both
$\Delta N$ and $\Delta D$ is justified, because results with e.g.
$\Delta N > 0$ and $\Delta D > 0$ are possible due to different force
sizes and force effectivenesses.

\begin{definition}[Wins and Losses]
Based on the relative assessment criteria $\Delta N$ and $\Delta D$, we provide the following definitions of wins and losses, given from the perspective of Blue.
	\begin{center}
\setlength{\tabcolsep}{3pt}		
\begin{tabular}{|l||l|l|}\hline
	&& \\[-2mm]
	Defined Notion (for Blue) & Condition on $\Delta N$ & Condition on $\Delta D$ \\[0.5mm] \hline 
	&& \\[-3mm]
	Strong win & $\Delta N > 0$ & $\Delta D < 0$ \\[0.5mm] 
	Weak win & $\Delta N > 0$ & $\Delta D > 0$ \\[0.5mm] 
	Weak loss & $\Delta N < 0$ & $\Delta D < 0$ \\[0.5mm] 
	Strong loss & $\Delta N < 0$ & $\Delta D > 0$ \\[0.5mm] \hline 
\end{tabular}
	\end{center}
\end{definition}

We supplement the above notions with corresponding definitions
at the pure kinetic level.
This will allow us to quantify the influence of a malware
epidemic by e.g. measuring the fraction of situations, in which
cyber warfare gives an advantage despite of losing a pure
kinetic battle.

\begin{definition}[Kinetic Superiority and Kinetic Inferiority]
\label{defkinsup}
Let $g_{\mathsf{kin}}\colon X \rightarrow X; x\mapsto x'$ designate a
mapping between scenarios, which set the kinetic effectiveness $\delta$ 
of the force elements to values assigned to force elements, which are
not infected by malware. Thus, $\delta$ becomes independent on the malware 
infection states of attacker and defender.
All other model parameter settings are left unchanged. 
The mapping $g_{\mathsf{kin}}$ provides a pure kinetic scenario $x'$ 
resulting from $x$ by hiding all malware effects. 
With respect to $x'$, we give the following definitions.
	\begin{center}
\setlength{\tabcolsep}{3pt}		
\begin{tabular}{|l||l|l|}\hline
	&& \\[-2mm]
	Defined Notion (for Blue) & Condition on $\Delta N$ & Condition on $\Delta D$ \\[0.5mm] \hline 
	&& \\[-3mm]
	Strong kinetic superiority & $(\Delta N) \circ g_{\mathsf{kin}}  > 0$ & $(\Delta D) \circ g_{\mathsf{kin}}   < 0$ \\[0.5mm] 
	Weak kinetic superiority & $(\Delta N) \circ g_{\mathsf{kin}}   > 0$ & $(\Delta D) \circ g_{\mathsf{kin}}   > 0$ \\[0.5mm] 
	Weak kinetic inferiority & $(\Delta N) \circ g_{\mathsf{kin}}   < 0$ & $(\Delta D) \circ g_{\mathsf{kin}}   < 0$ \\[0.5mm] 
	Strong kinetic inferiority & $(\Delta N) \circ g_{\mathsf{kin}}   < 0$ & $(\Delta D) \circ g_{\mathsf{kin}}   > 0$ \\[0.5mm] \hline 
\end{tabular}
	\end{center}
\end{definition}



\begin{table}[b!]
\setlength{\tabcolsep}{3pt}		
\tabulinesep=1.5mm
\centering
\begin{tabu} to \textwidth{| L{4.1cm}  X[1]  L{1.9cm} |}
	\hline
	\textbf{Observables} & \textbf{Description} & \textbf{Codomain} \\
        \hline \hline
	   $\Delta N:= N_b(t_{\mathsf{end}}) - N_r(t_{\mathsf{end}})$ & Relative number of
	   existing force elements at $t_{\mathsf{end}}$  & $\Delta N\in
	   \mathbb{R}$ \\
           \hline
	   $\Delta D := D_b(t_{\mathsf{end}}) - D_r(t_{\mathsf{end}})$ &
		Relative number of destroyed elements at $t_{\mathsf{end}}$
		& $\Delta D\in
	   \mathbb{R}$ \\ \hline 
	   $L_b:= D_b(t_{\mathsf{end}})$ &
		Destroyed elements of the blue force at $t_{\mathsf{end}}$
		& $L_b\in \mathbb{R}_0^+$ \\ \hline 
\end{tabu}
\caption{List of the observables used for assessment purposes. 
%
\label{definition_outcome}
}
\end{table}


\begin{proposition}[Extrema of Observables]\label{eoo}
Let us assume $D_b(0)$, $D_r(0)=0$.
\begin{enumerate}[label=\emph{\alph*}),nosep,leftmargin=*]
\item \label{rangeofobs}
It holds $\max (\Delta N)= \max (\Delta D)=N_b(0)$ and
$\min (\Delta N)= \min (\Delta D)=-N_r(0)$. 
\item \label{rangeofloss} 
It holds $\max (L_b)= N_b(0)$ and $\min (L_b)= 0$. 
\end{enumerate}
\end{proposition}

\begin{proof}[Sketch of Proof]
~
\begin{enumerate}[label=\emph{\alph*}),nosep,leftmargin=*]
\item Due to the definition of $\Delta N$ and the 
monotonic decrease of $N_b$, the observable $\Delta N$ can not have a value
larger than $N_b(0)$. At the end of the simulation, it reaches this value 
if all force elements of Blue survive e.g. due to $ \delta_{r,ikjl} =0$ and
if no elements of Red survive due to $ \delta_{b,ikjl} > 0$.
Proposition~\ref{sympropos}.\ref{proposfirst} leads to 
$\min (\Delta N)=-N_r(0)$. The corresponding statements for 
$\Delta D$ are a consequence of the preservation of 
$N(t)+D(t)$ over time according to
proposition~\ref{sympropos}.\ref{propos_ndconst}.
\item Proof analogous to a)
\end{enumerate}
\end{proof}


Quantitative assessments enable comparisons of scenario outcomes. For 
additionlly being able to judge a {\em single} outcome as especially 
'good' or 'bad',  the ranges of possible assessment values have to be known. 
They are provided by Proposition~\ref{eoo}.

\subsection{Sampling of Simulation Outcomes}\label{secthree}

The details of a future (or current) malware attack may not be
known exhaustively. In such a case, the scenario to be 
considered in the risk assessment process will not be uniquely
determined. The existing uncertainties have to be included in the approach.
The usually infinite size of the scenario space $X$ compatible
with the existing uncertainties rules out a brute-force processing.
Instead, a small randomly selected subset
$\tilde X \subset^{\mathsf{fin}}X$ of scenarios is analyzed. 
This leads to a Monte-Carlo approach for calculating the 
risk with inclusion of uncertainties approximatively.
In the following, the subset $\tilde X$ is called the {\em simulation 
design}. The set of outcomes of the simulation
scenarios $\tilde X$ is designated as $\tilde Y \subseteq Y$.

In order to enable the intended random selection, the scenario space 
$X$ is enriched by a notion of probability.
This gives the so-called {\em Monte-Carlo design space} $(X,\Pr)$.
For assuring the appropriateness of the random sample $\tilde X$,
the probability distribution $\Pr$ has to encode the 
frequency of occurrence of the different scenarios $x\in X$.
In this way, the probability distribution $\Pr$ represents 
the knowledge about the considered situation.

The encoding of available and missing knowledge to a corresponding 
probability distribution $\Pr$ follows information-theoretic principles.
We have to avoid that $\Pr$
contains more information about the situation than actually given;
consequently, we choose the probability distribution $\Pr$ with the highest 
entropy among all  distributions compatible with the given information.
According to \cite{jaynes1957information},
the entropy $H(\Pr)$ of a continuous probability distribution $\Pr$ is given 
by $H(\Pr) = -\int \Pr(x) \log \Pr(x) \dd x$. 

If the initially available knowledge is minimal in the sense of
information-theory, only a lower and an upper limit of the domain of 
support is given. 
Renouncing any knowledge would be a problem, because in this case 
the finiteness of the entropy $H(\Pr)$ can not be assured anymore.
The resulting uniform distribution \cite{conrad2004probability} 
(see table~\ref{inftheotab}) may raise some scepticism, because 
uniformly distributed parameters are not justified by observations.
We have to take care when interpreting the role of the uniform
distribution, however. The uniform probability distribution of
the input parameters does not represent a single specific situation,
but a kind of 'superposition' of the variety of all {\em possible}
distributions compatible with the available knowledge.
%
Table~\ref{inftheotab} gives some maximum entropy probability distributions for different kinds of basic knowledge. We restrict ourselves to distributions used in the example presented in section~\ref{secfour}. Concerning parameters with a finite support, the uniform distribution is used if no further knowledge is available. If additionally expected values are known, the Beta distribution results. For a support equal to $\real^0_+$, the exponential distribution results in the case of a known expected value. Additional knowledge of the standard deviation leads to a lognormal distribution. In the case of a support equal to $\real$, the normal distribution has to be selected for known expected value and standard deviation.

\begin{table}[tbh!]
\setlength{\tabcolsep}{3pt}		
\tabulinesep=1.5mm
\centering
\begin{tabu} to \textwidth {| L{1.8cm} | L{1.6cm} | L{3.6cm} | X[1] |}
	\hline
	\textbf{Type} & \textbf{Support} & \textbf{Constraints} & \textbf{Probability Density} \\ \hline \hline
	Uniform & $[a,b]$ & & $f(x)=\frac{1}{b-a}$  with $a\neq b$
        \\ \hline
	Exponential & $[0,\infty)$ & $\E(x)=1/\lambda$ & $f(x)=\lambda\exp(-\lambda x)$ 
        \\ \hline
	Normal & $(-\infty,\infty)$ & $\E(x)=\mu$, $\Var(x)=\sigma^2$ & $f(x)= \frac{1}{\sigma \sqrt{2\pi}}\exp\left( - \frac{1}{2} \left(\frac{x-\mu}{\sigma}\right)^2 \right)$  
        \\ \hline
	Beta & $[a,b]$ & $\E(x)=\frac{\alpha\cdot b + \beta\cdot a}{\alpha+\beta}$ & 
	with $a\neq b$, $f(x)=$  \\
	& & $\E(b-x)=\frac{\beta\cdot b + \alpha\cdot a}{\alpha+\beta}$ & $\qquad$ $\frac{1}{(b-a)B(\alpha,\beta)}(\frac{x-a}{b-a})^{\alpha-1} (\frac{b-x}{b-a})^{\beta -1}$ \\
        \hline
	Lognormal & $[0,\infty)$ & $\E(\ln(x))=\mu$ & $f(x)=\frac{1}{\sigma x \sqrt{2\pi}}\exp\left( - \frac{(\ln(x)-\mu)^2}{2\sigma^2}\right)$ \\
	& & $\E((\ln(x)-\mu)^2)=\sigma^2$ & 
        \\ \hline
\end{tabu}
	\caption{Maximum entropy probability distributions resulting from basic constraints according to \cite{park2009maximum}. In the table, $B$ designates the Beta function.
\label{inftheotab}
}
\end{table}

The inclusion of the knowledge resp.\ uncertainty perspective gives the following picture.  Input parameters of the model subject to uncertainties may vary, taking on different values. Their variation is described by the probability distribution $\Pr$ of the space $(X,\Pr)$.  Thus we can state, that in the presented approach fixed parameter values may be replaced by fixed probability distributions characterizing their variations. The accompanying transition from a point estimate to a statistics of estimates gives a richer description of the model behavior. 

Though a random selection of the scenarios
successfully avoids a selection bias, it may introduce
a possible statistical bias. For monitoring the statistical quality,
we compare the measured properties of the sampled subset 
$\tilde X\subset X$ with the exact properties of the full design 
$(X,\Pr)$. The outcomes $\tilde Y$ can be used for
monitoring purposes as well.
Numerical results of such comparisons are shown in
figure~\ref{errorrates}.

\begin{figure}[tbh!]
\begin{center}
\begin{tabu} to \textwidth{X[1]X[1]}
	\multicolumn{1}{c}{Input Bias of blue force size $N_b(0)$} &
\multicolumn{1}{c}{Outcome Bias of Survivors $\Delta N$} \\ 
\includegraphics[width=0.45\textwidth]{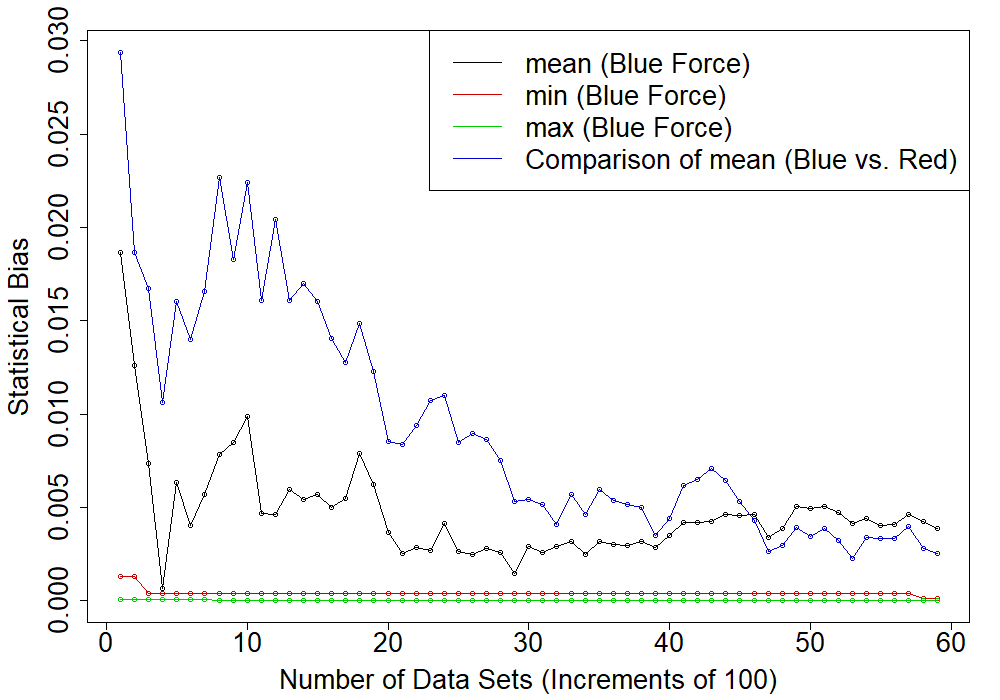}
&
\includegraphics[width=0.45\textwidth]{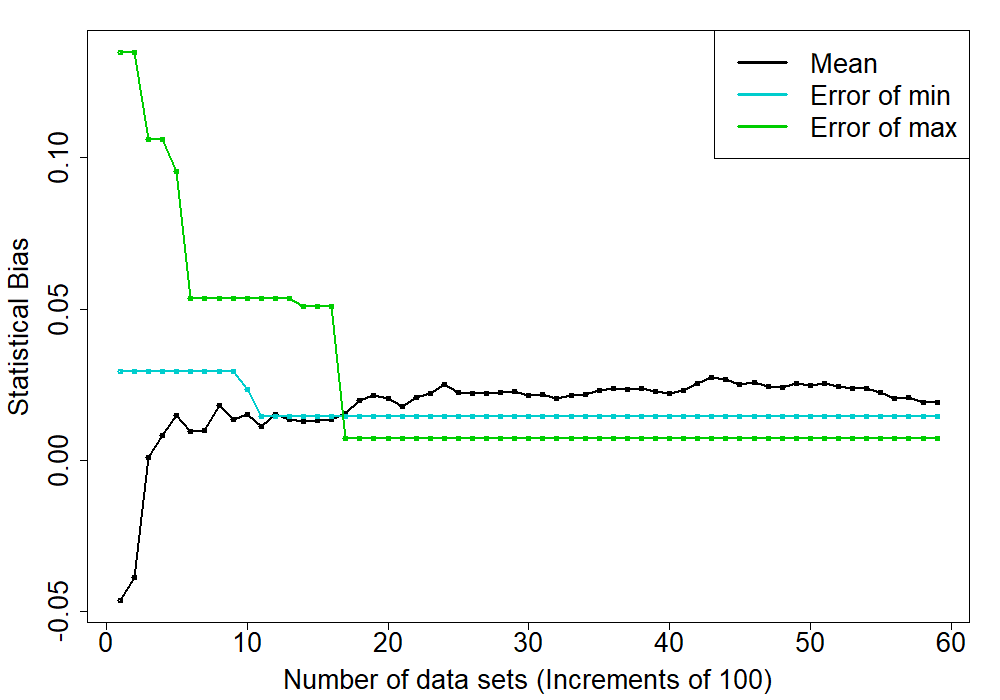}
\end{tabu}
\end{center}
\caption[Statistical Bias]{
Statistical bias of random sampling dependent on the number of simulation runs.
The comparison of statistical quantities of the simulation design $\tilde X$ 
(and of the associated outcomes $\tilde Y$) with 
theoretical predictions for the corresponding full design allows a
monitoring of the statistical bias.
If the bias turns out to be too high, the size of the simulation
design has to be increased accordingly. 

As an example, we consider uniformly distributed force sizes
$N_b(0), N_r(0)\in [0,1]$ (left) and the resulting behavior of 
$\Delta(N)$ (right) in a symmetric situation.
Let us write $\hat N := N_b(0)$ for short. 
Concerning the input, the plot includes the
difference between the measured value $\mean(\hat N(\tilde X))$ and the
theoretically expected value 0.5, analogously the deviation of
$\min(\hat N(\tilde X))$ resp.  $\max(\hat N(\tilde X))$ 
from 0 resp. 1, and finally the difference of $\mean(\hat N(\tilde X))$ 
between Blue and Red.
Concerning the outcome, the plot includes $\mean(\Delta N(\tilde X))$ and the
deviation of $\max(\Delta N(\tilde X))$ resp.
$\min(\Delta N(\tilde X))$ from 
$1.0$ resp. $-1.0$ as given by proposition~\ref{eoo}.\ref{rangeofobs}.
All considered input and outcome quantities will have the value zero for a completely bias-free 
statistics.
\label{errorrates}
}
\end{figure}

\subsection{Risk Assessment by Sampling}\label{statana}



Uncertainties of the input parameters may especially cause
a variation of the losses $L = D(t_{\mathsf{end}})$.
This leads to the notion of risk in a straightforward way.

\begin{definition}[Risk]
The {\em risk} $\mathcal{R}$ assigned to a scenario $x\in X$ is defined as
$\mathcal{R}:= L(x)$. 
If due to epistemic uncertainties a set of scenarios has to be considered
given by the Monte-Carlo design $(X,\Pr)$, the risk
$\mathcal{R}$ is more generally defined as expected value $\mean(L(X))$.
\end{definition}

From a computational perspective, the risk may be determined 
approximatively as the mean loss of the random sample of scenarios 
in accordance with \cite{berger2013statistical}. In the following 
proposition, some properties of risk are derived.

\begin{proposition}[Properties of Risk]\label{proppropor}
Let a simulation design $\tilde X$ be given. We assume $D(0)=0$.
\begin{enumerate}[label=\emph{\alph*}),nosep,leftmargin=*]
\item \label{propria} It holds 
$\mean(N(0)) = \mean(N(t_{\mathsf{end}})) + \mean(D(t_{\mathsf{end}})) =
\mean(N(t_{\mathsf{end}})) + \mathcal{R}$
\item $\mathcal{R} \le \max(N(0))$.
\end{enumerate}
\end{proposition}

\begin{proof}
~
\begin{enumerate}[label=\emph{\alph*}),nosep,leftmargin=*]
\item According to proposition~\ref{sympropos}.\ref{propos_ndconst},
one gets $N(0) = N(t_{\mathsf{end}}) + D(t_{\mathsf{end}})$. 
Applying the operator $\mean$ 
and taking its linearity into account, the 
definition of $\mathcal{R}$ gives the claim.
\item According to the proof of part a), it holds 
$\mean(N(0)) = \mean(N(t_{\mathsf{end}})) + \mathcal{R}$ and thus
$\max(N(0)) \ge \mean(N(0)) \ge \mathcal{R}$,
\end{enumerate}
\end{proof}


\section{Example: Analyzing Homogeneous Lanchester-SIR}\label{secfour}

\subsection{A Sequence of Situations with Decreasing Uncertainty}

\begin{table}[tbh!]
{\tabulinesep=1.0mm
\setlength{\tabcolsep}{3pt}		
\begin{tabu} to \textwidth {| L{2.8cm} | X[1] | L{2.5cm} |}
	\hline
	\textbf{Parameter} & \textbf{Situational Description} & \textbf{Constraints} \\ \hline \hline
	Initial size of forces $N_b$, $N_r$ & The scale of forces involved in combat is unknown. We only assume a plausible lower and upper limit. & Finite support, $N(0)\in [0.1,1]$ \\ \hline
	Kinetic effectiveness $\delta_b$, $\delta_r$ & Expected values of kinetic effectiveness exist, but the available data do not suffice for providing the standard deviation. & $\E({\delta})=1.2$ for Blue and Red\\ \hline
	Lanchester parameters $p$, $q$ & As for the kinetic effectiveness, expected values exist, but not standard deviations. The probability distributions of $p$ and $q$ coincide, because we do not know which force is attacking and which defending. & $\E(p)=\E(q)=2.5$ \\ \hline
\end{tabu}
	}
\caption{Initial uncertainty constraints on the kinetic parameters of the example. 
\label{examplekinparam}
}
\end{table}



\begin{table}[tbh!]
{\tabulinesep=1.0mm
\setlength{\tabcolsep}{3pt}		
\begin{tabu} to \textwidth {| L{2.7cm} | X[1] | L{3.8cm} |}
	\hline
	\textbf{Parameter} & \textbf{Situational Description} & \textbf{Constraints} \\ \hline \hline
	Time difference $\Delta t_{\mathsf{att}}$ between kinetic and cyber attack & 
	The start times of malware attack and kinetic combat will not differ substantially. A late start is uncommon due to the limited influence on kinetic combat. An early start, on the other hand, gives the opponent time to develop a patch. 	
	& $\E({\Delta t_{\mathsf{att}}})=0$, $\Var({\Delta t_{\mathsf{att}}})=25^2$ \\ \hline
	Duration $\Delta t_{\mathsf{mal}}$ of malware attack & We only assume a plausible lower and upper bound. & Finite support, $\Delta t_{\mathsf{mal}}\in [0.1,10]$  \\ \hline
	Malware attack rate $\alpha$ & We assume a plausible upper bound. & Finite support, $\alpha\in [0,1]$ \\ \hline
	Initial number $R(0)$ of non-vulnerable force elements &
	Blue uses a large variety of software systems, which lowers the probability that a malware epidemics covers a large fraction of its force. It holds $I(0)=0$ and thus $S(0)=N(0)-R(0)$. & We assume $\E({R_b(0)})=5N_b(0)/7$, $\E(S_b(0))=2N_b(0)/7$, $\E({R_r(0)})=N_r(0)/2$, $\E(S_r(0))=N_r(0)/2$ \\ \hline
	Infection rate $\beta$ & Experience may provide the expected infection rate and its standard deviation. & $\E(\ln({\beta}))=\mu=1$ and $\E((\ln({\beta})-\mu)^2)=\sigma^2=1$ \\ \hline
	Effectiveness reduction $\eta$ of infected elements & The support of $\eta \in [0,1]$ is given by definition. Experience gives expected values as a supplement. & Finite support $\eta \in [0,1]$ with $\E({\eta})=0.5$  and $\E(1-{\eta})=0.5$ \\ \hline
\end{tabu}
	}
\caption{Initial uncertainty constraints on the cyber attack parameters of the example. 
\label{exampleattparam}
}
\end{table}

\begin{table}[tbh!]
{\tabulinesep=1.0mm
\setlength{\tabcolsep}{3pt}		
\begin{tabu} to \textwidth {| L{2.5cm} | X[1] | L{4.0cm} |}
	\hline
	\textbf{Parameter} & \textbf{Situational Description} & \textbf{Constraints} \\ \hline \hline
	Time difference $\Delta t_{\mathsf{pat}}$ between patching and kinetic attack &
	Again, we assume a start time $\Delta t_{\mathsf{pat}}$  
	of patching close to the start of kinetic combat.
	Patching will have no influence when applied after kinetic combat.
	Patching early then again usually just leads to the exploit of another vulnerability than the already patched one.  
	& $\E({\Delta t_{\mathsf{pat}}})=0$, $\Var({\Delta t_{\mathsf{pat}}})=25^2$ \\ \hline
	Patch rate $\gamma$ for vulnerable elements & As for the infection rate, we can provide expected values based on previous experience. The value for Blue and Red differ. & We assume $\E({\gamma_b})=1/4$, $\E(1-{\gamma_b})=3/4$, $\E({\gamma_r})=1/6$, $\E(1-{\gamma_r})=5/6$ \\ \hline
	Patch rate $\tilde \gamma$ for infected elements & Again, we can provide expected values based on previous experience. We adopt a smaller patch rate for infected elements than for vulnerable elements, i.e. $0\le \tilde \gamma \le \gamma$. & We assume $\E({\tilde\gamma_b})=\gamma_b/2$, $\E(\gamma_b-{\tilde\gamma_b})=\gamma_b/2$, $\E({\tilde\gamma_r})=\gamma_r/2$, $\E(\gamma_r-{\tilde\gamma_r})=\gamma_r/2$ \\ \hline
\end{tabu}
}
\caption{Initial uncertainty constraints on the cyber patching parameters of the example. 
\label{examplepatparam}
}
\end{table}

The evaluation approach is applied to the 
Lanchester-SIR model introduced in section~\ref{secexamplemodel}. 
Initially, we consider a situation without specific conflict
ahead. Technological breakthroughs, economic collapse, 
and other unforeseen events may lead to an large variety of
future scenarios. The resulting uncertainties
are given in the tables~\ref{examplekinparam}, \ref{exampleattparam}, 
and \ref{examplepatparam}.
Over time, these uncertainties about the expected combat situation 
may be eliminated stepwise due to an increasing amount of knowledge 
and the decisions that have been made. This leads to a 
sequence of decreasing uncertainty. The corresponding sequence of 
risk assessments is given in section~\ref{ssra}. In this
exemplary course of action, the analysis results sometimes 
influence the hypothetical decisions.

The considerations contain some simplifications.
First, the start times $\Delta t_{\mathsf{att}}$,
$\Delta t_{\mathsf{pat}}$ and other parameters of cyber actions
result in practice 
from control decisions. These are not part of the model and 
replaced by random parameter settings (see section~\ref{sscy2}). 
This is justified by the fact,
that watching out for vulnerabilities and patching them 
is a continuously executed process. Accordingly,
patching is not necessarily a consequence of a malware
attack. Exploiting information provided by intelligence can foil 
the principal causal order between malware attack and patching action
as well.
Furthermore, just intending to develop malware with a given
infection rate or to close a vulnerability with a required patch 
rate does not mean that such a project can indeed be realized 
due to possible technical limitations and organizational frictions. 
Second, the uncertainties of the model parameters are assumed to be 
essentially independent from each other. The authors have made several
exceptions, however, which are indicated in the tables~\ref{examplekinparam}, 
\ref{exampleattparam}, and \ref{examplepatparam}.
Now, we present the uncertainty sequence going to be examined.


\smallskip

\textbf{Step 1 (Start Situation):} The parameters are uncertain at the beginning, because the circumstances of the next conflict are unknown (see 
tables~\ref{examplekinparam}, \ref{exampleattparam}, 
\ref{examplepatparam}).

\smallskip

\textbf{Step 2 (Force Selection):} A conflict emerges with a foreign state. 
An expeditionary force is sent. Accordingly, the opposing kinetic forces are known and the values of the uncertain parameters given in table~\ref{examplekinparam} can be fixed for Blue and Red. The new knowledge consists of $N_b(0)=0.15$, $N_r(0)=0.3$, $\delta_b=2.0$, $\delta_r=0.8$, $p=2.5$, $q=1.7$.

\smallskip

\textbf{Step 3 (Cyber Attack):} Since the enemy turns out to be superior from the kinetic perspective, a cyber attack concurrent to kinetic combat is planned. Here, own superiority was assumed from the beginning (see table~\ref{exampleattparam}), but with large uncertainties due to intelligence gaps.
Accordingly, the cyber part of Red is a focus of blue intelligence.
As a result, the cyber attack parameters of table~\ref{exampleattparam} can be determined for Blue. 
Force elements, which cannot be infected by the malware in principle, are
represented by an initial number $R(0)\ge 0$.
The blue cyber attack is characterized by the parameters settings
$\Delta t_{b,\mathsf{att}}=-15.0$,
$\Delta t_{b,\mathsf{mal}}=1.0$,
$\alpha_b=0.5$.
The effect of the malware attack on Red is described by
$R_r(0)=0.15$,
$S_r(0)=0.15$,
$\beta_r=1.75$, 
$\eta_r=0.15$.

\smallskip

\textbf{Step 4 (Enemy Attack):} When the blue intelligence effort, which has been started in the last step, is continued, a planned red cyber attack is uncovered. This determines the parameters of table~\ref{exampleattparam} also for Red.
The new information turns out to be invaluable for Blue. 
Since the vulnerability going to be exploited by Red has become clear, an approach for patching the own force elements can be developed well in advance. This specifies the values of the parameters of table~\ref{examplepatparam} for Blue as well. 
The parameter values describing the red malware attack are as follows: 
$\Delta t_{r,\mathsf{att}}=2.5$, 
$\Delta t_{r,\mathsf{mal}}=5.0$,
$\alpha_r=0.5$.
The red malware propagates across the blue force with
$R_b(0)=0.75N_b(0)$,
$S_b(0)=0.25N_b(0)$,
$\beta_b=0.5$, 
$\eta_b=0.5$.
For the countermeasures mounted by Blue it holds
$\Delta t_{b,\mathsf{pat}}=-0.25$,
$\gamma_b=0.5$,
$\tilde \gamma_b=0.5\gamma$.

\smallskip

\textbf{Step 5 (Enemy Recovery):} Despite of all blue efforts,
Red can still win the outcome as bottom line (see figure~\ref{sequofhisto},
second last row). His chances of success depend on a fast and efficient
recovery from the blue cyber attack. Accordingly, Blue tries to hamper
the development of an efficient patch by Red. This is realized successfully.
The parameters of table~\ref{examplepatparam} are now also given for Red, 
which removes the last uncertainties. We get
$\Delta t_{r,\mathsf{pat}}$=45.0,
$\gamma_r=0.2$,
$\tilde \gamma_r=0.005\gamma_r$
as characterization of the red countermeasures.
The fast infection of red force elements together with the slow patch rate 
of infected systems maximizes the effect of the blue malware attack.


The analysis of the given situation (and uncertainty) sequence is organized in the following way.  In section~\ref{ssra}, the risk in the various stages of the sequence is determined. After demonstrating the principal applicability and usefulness of the framework, the value of various contributions to the framework is shown. This concerns the underlying model in section~\ref{ssim}, the effectiveness of random resp.\ optimized cyber actions in section~\ref{sscy} resp.\ \ref{sscy2}, 
and, finally, the extension from point estimates to the inclusion of uncertainties in section~\ref{ssun}.

\subsection{Risk Assessment Results}\label{ssra}

About 100000 Monte-Carlo simulation runs were executed for each step 
of the uncertainty sequence.
The simulation runs were stopped as soon as the change of all
compartment levels was smaller than $1.0 \cdot 10^{-5}$
for a simulation time interval equal to 1.
%
Figure~\ref{sequofhisto} gives the histograms of the observables
$\Delta N$, $\Delta D$, and $L_b$ assessing the simulation outcomes
for each step of the complete uncertainty sequence.

In the histograms, the variation of outcomes is significantly reduced
as soon as the kinetic parameters are fixed after step 1. 
Much more prominent cyber effects can be expected for cyber actions
being optimized concerning their effect on kinetic combat
(see section~\ref{sscy2}).
The small number ($2/7 \cdot N_b$) of blue force elements being
susceptible to the red malware attack also contributes
to the small scale of cyber effects. 
The losses $L_b$ have a distinct peak close to zero only in the
start situation. In the step 'force selection', a now even more
pronounced peak is located at $L_b\approx 0.1$. Later, in the step
'enemy attack', the losses could be reduced further.
Simultaneously, the peak is located from now on at $L_b\approx 0.06$.
This change notices its coming by a corresponding distortion of
the outcome distribution in the intermediate step 'cyber attack'.
Similar variations can be detected in $\Delta N$ and $\Delta D$.

\begin{figure}[tbhp!]
\tabulinesep=0.7mm
\begin{center}
\begin{tabu} to \textwidth{| C{0.4cm} | C{0.3cm} | X[1] | X[1] | X[1] |}
\hline
	&& \multicolumn{1}{|c|}{\Large $\Delta N$} & \multicolumn{1}{c|}{\Large $\Delta D$} & \multicolumn{1}{c|}{\Large $L_b$} \\ \hline \hline
	\multirow{5}{*}{\adjustbox{angle=90}{\hspace*{+0.0cm} \Large Frequency}} &
	\multirow{1}{*}{\adjustbox{angle=90}{\hspace*{+0.0cm} \large Start Situation}} &
\includegraphics[trim = 15mm 15mm 0mm 0mm,scale=0.12]{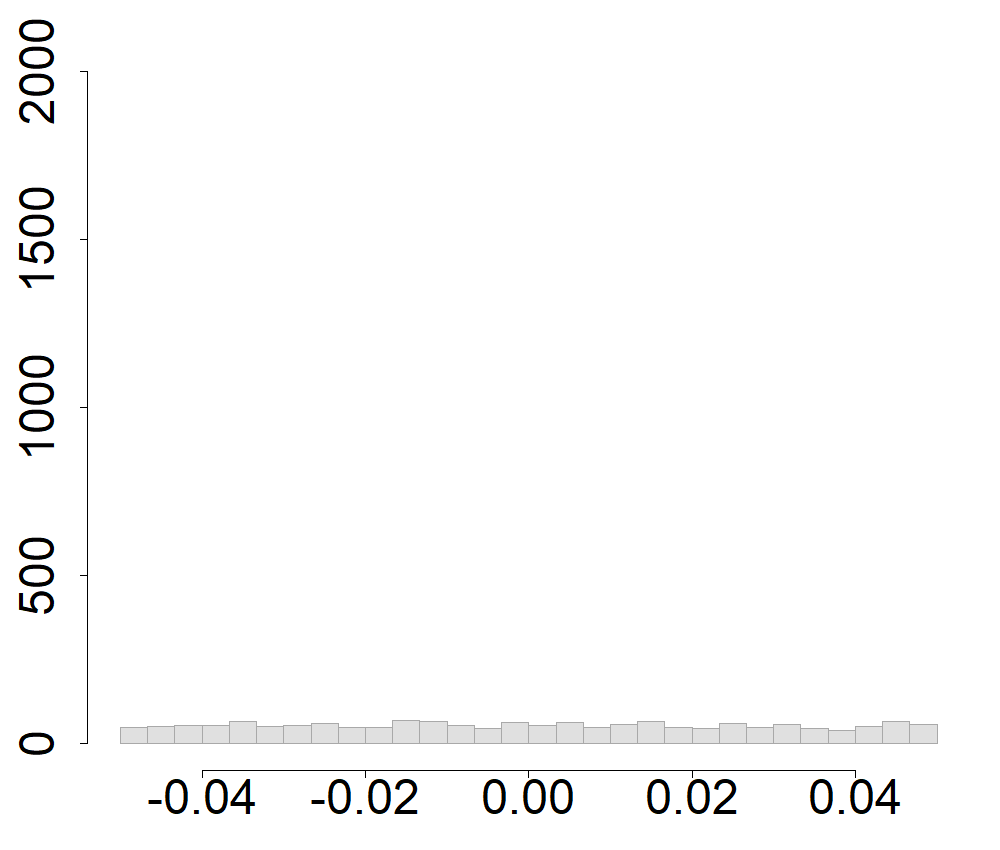}
	&
\includegraphics[trim = 15mm 15mm 0mm 0mm,scale=0.12]{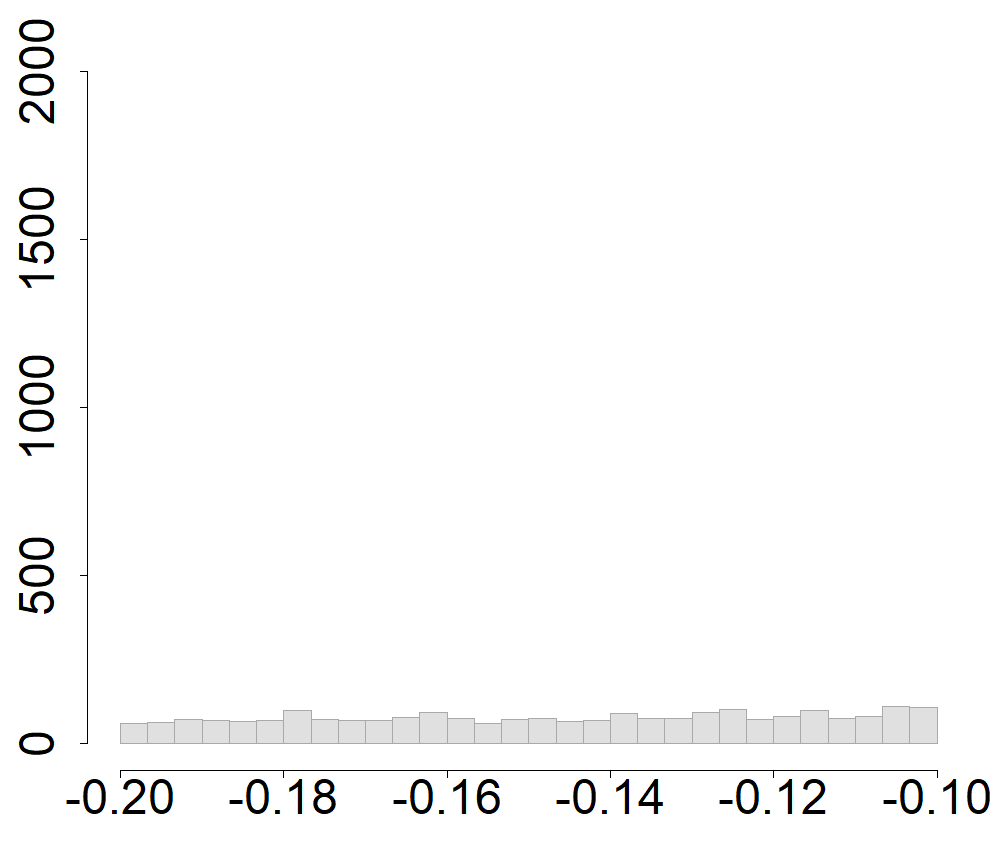}
	&
\includegraphics[trim = 15mm 15mm 0mm 0mm,scale=0.12]{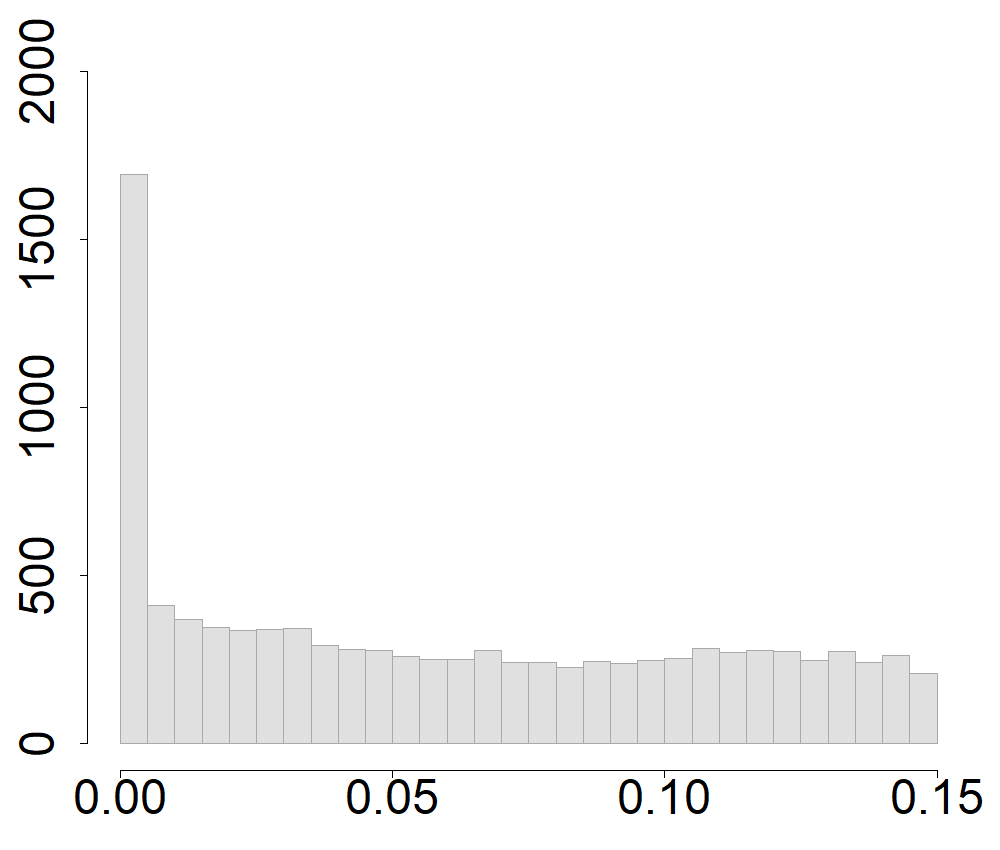}
\\
	\cline{2-5}
	&
	\multirow{1}{*}{\adjustbox{angle=90}{\hspace*{+0.0cm} \large Force Selection}} &
\includegraphics[trim = 15mm 15mm 0mm 0mm,scale=0.12]{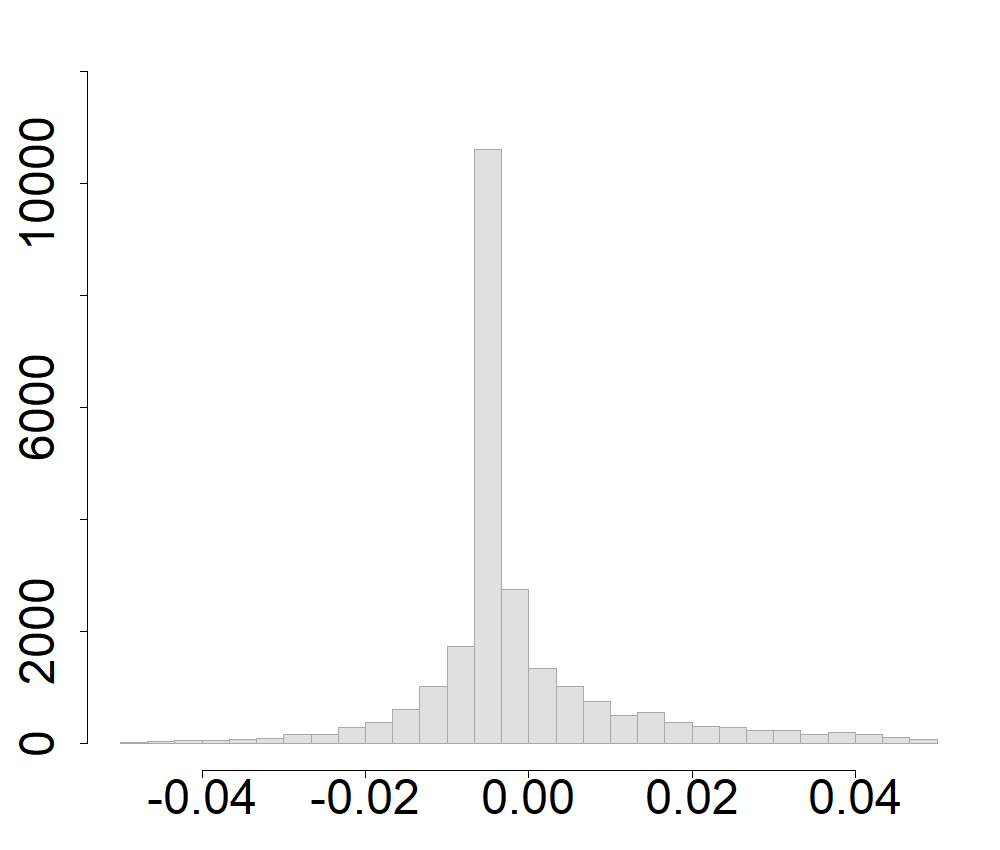}
	&
\includegraphics[trim = 15mm 15mm 0mm 0mm,scale=0.12]{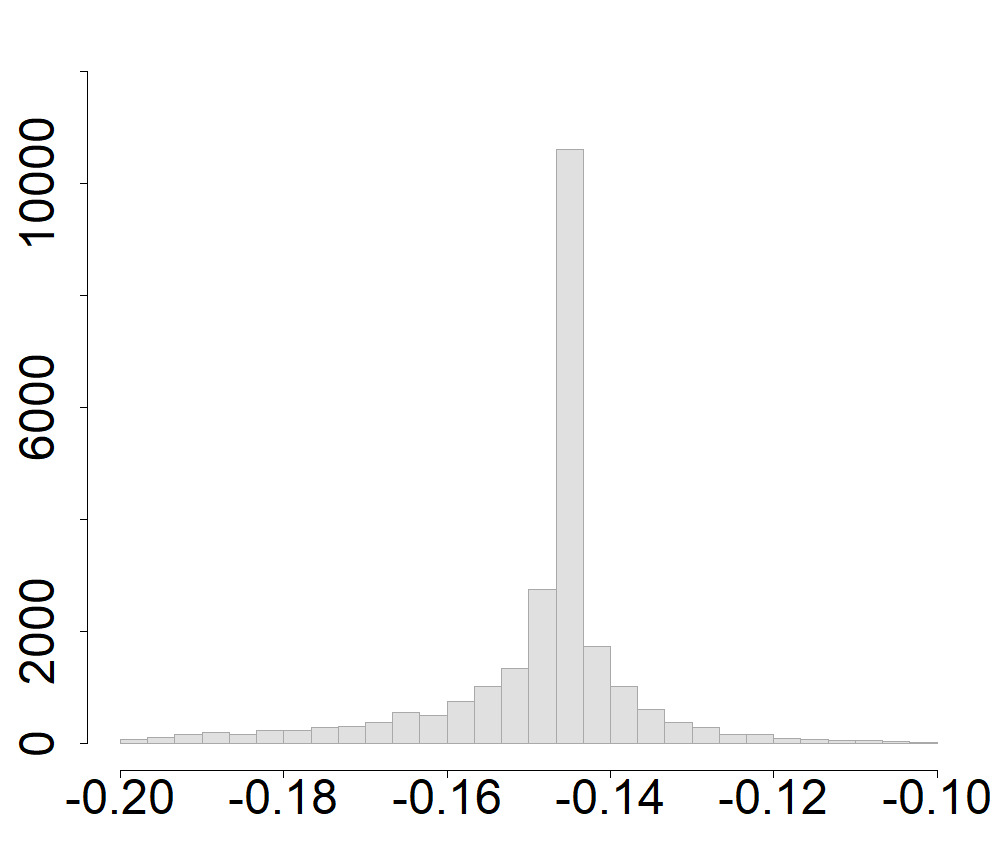}
	&
\includegraphics[trim = 15mm 15mm 0mm 0mm,scale=0.12]{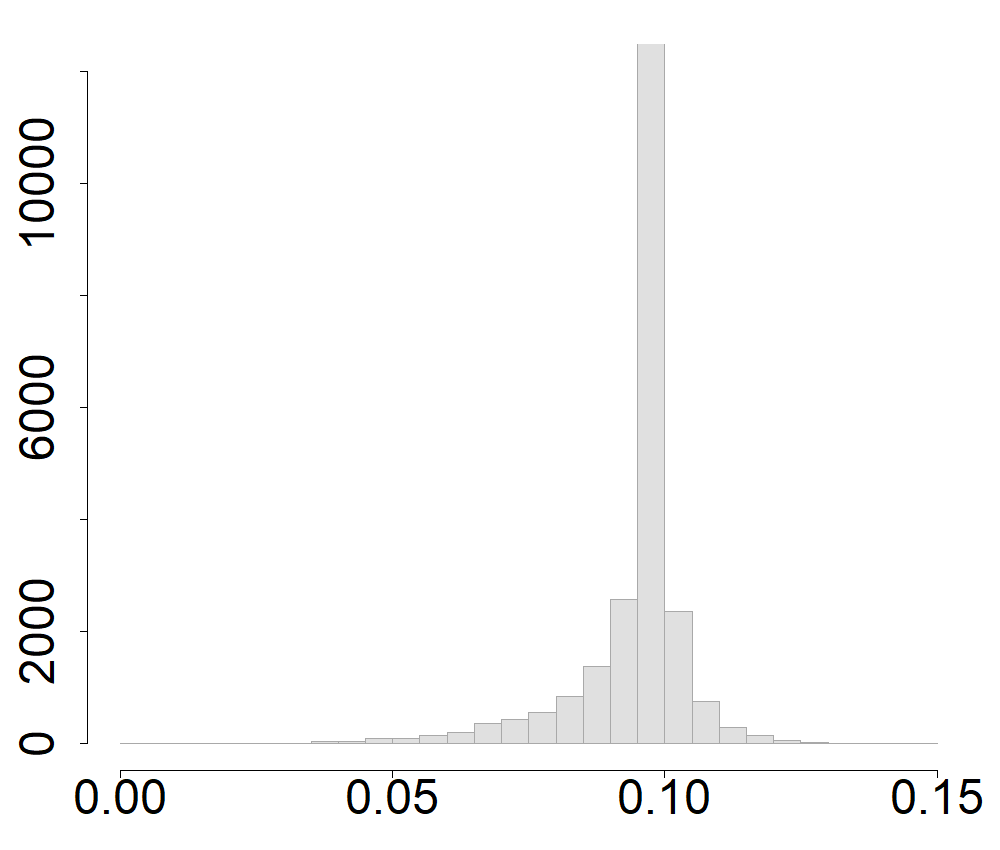}
\\
	\cline{2-5}
	&
	\multirow{1}{*}{\adjustbox{angle=90}{\hspace*{+0.0cm} \large Cyber Attack}} &
\includegraphics[trim = 15mm 15mm 0mm 0mm,scale=0.12]{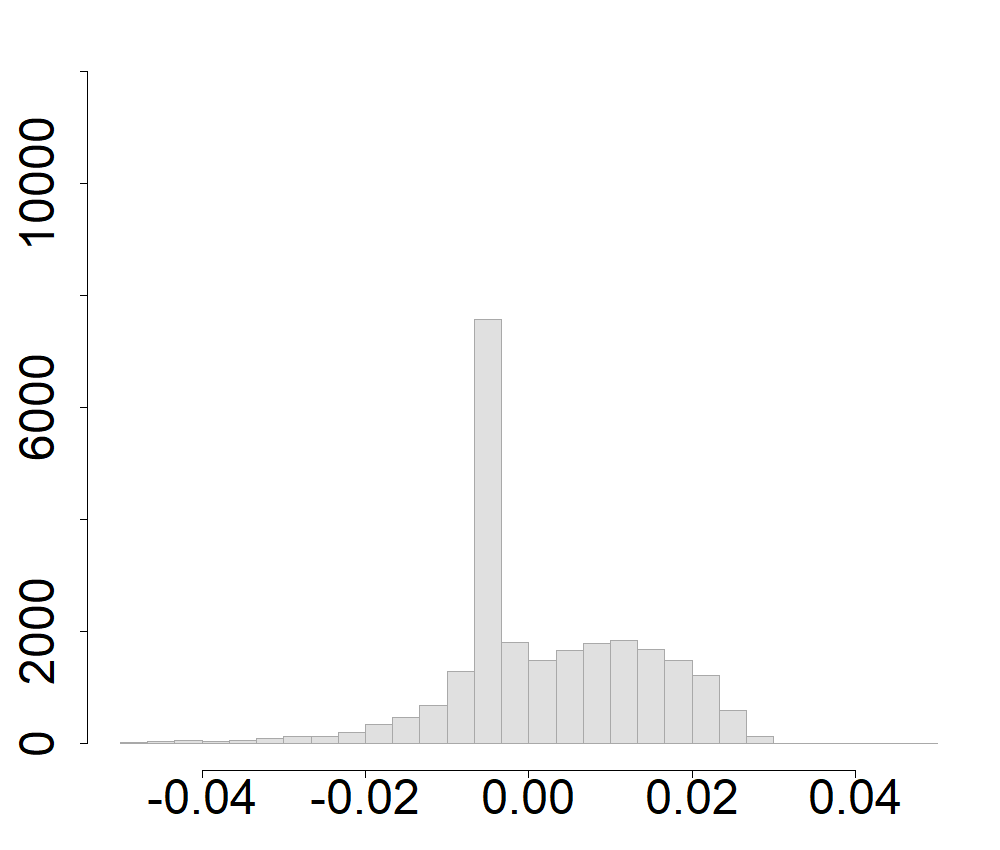}
	&
\includegraphics[trim = 15mm 15mm 0mm 0mm,scale=0.12]{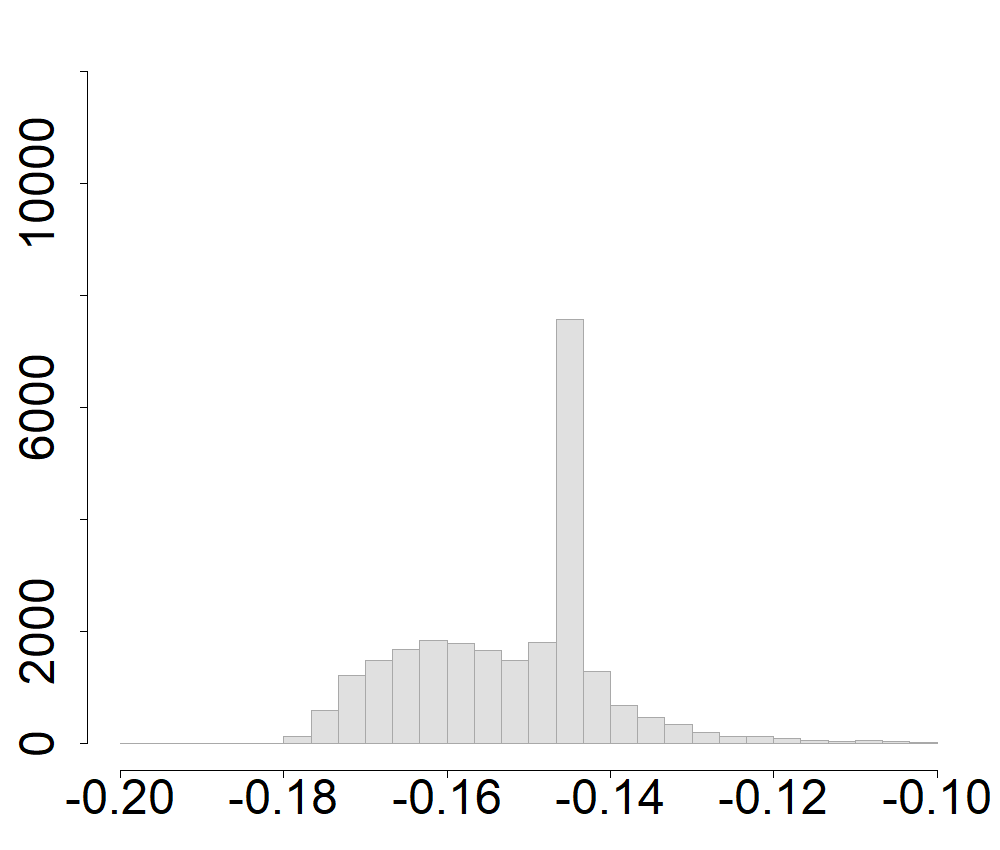}
	&
\includegraphics[trim = 15mm 15mm 0mm 0mm,scale=0.12]{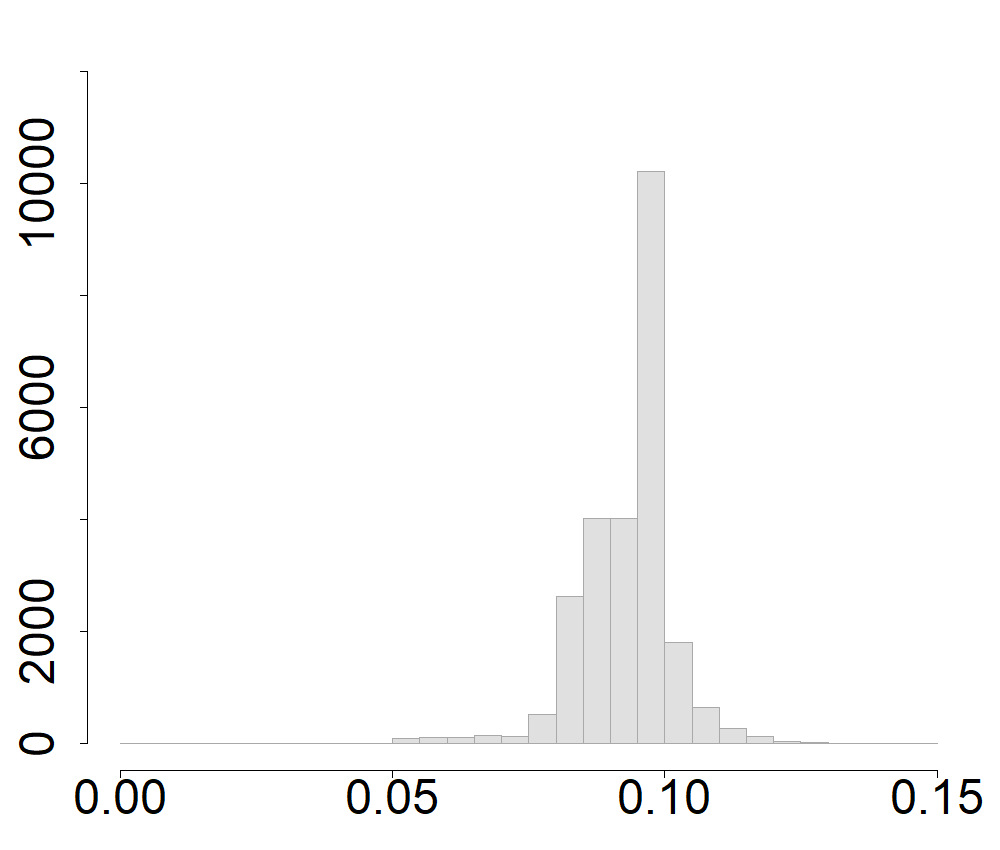}
\\
	\cline{2-5}
	&
	\multirow{1}{*}{\adjustbox{angle=90}{\hspace*{+0.0cm} \large Enemy Attack}} &
\includegraphics[trim = 15mm 15mm 0mm 0mm,scale=0.12]{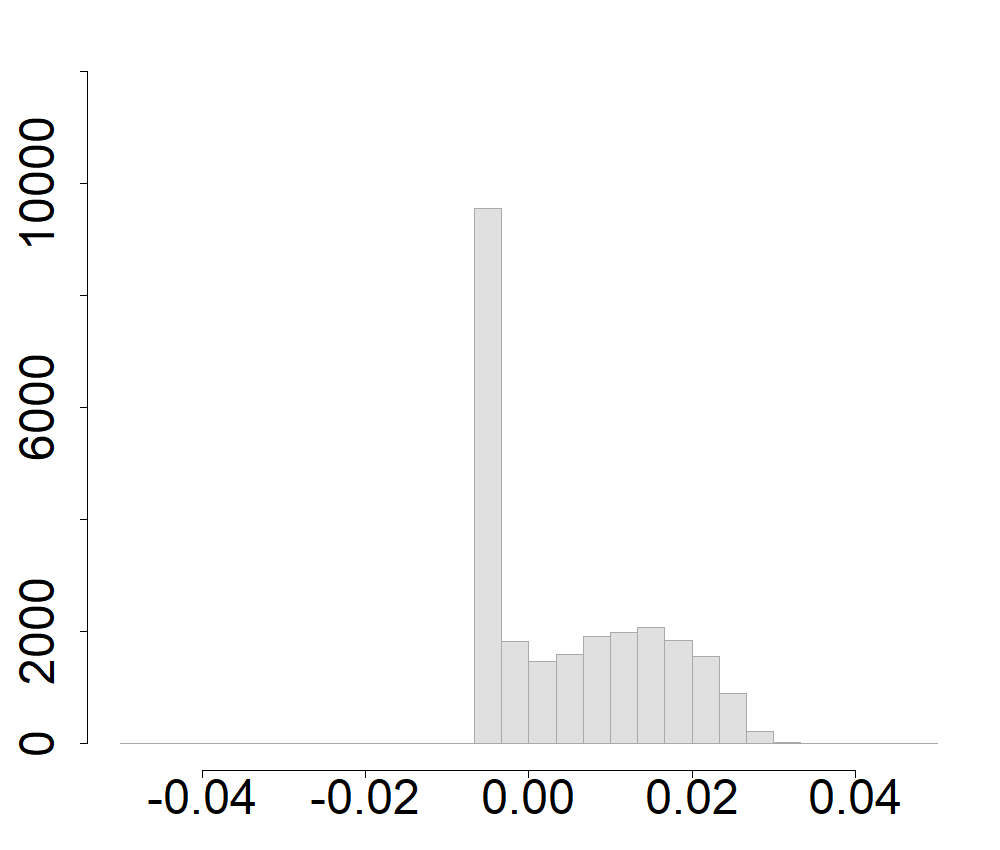}
	&
\includegraphics[trim = 15mm 15mm 0mm 0mm,scale=0.12]{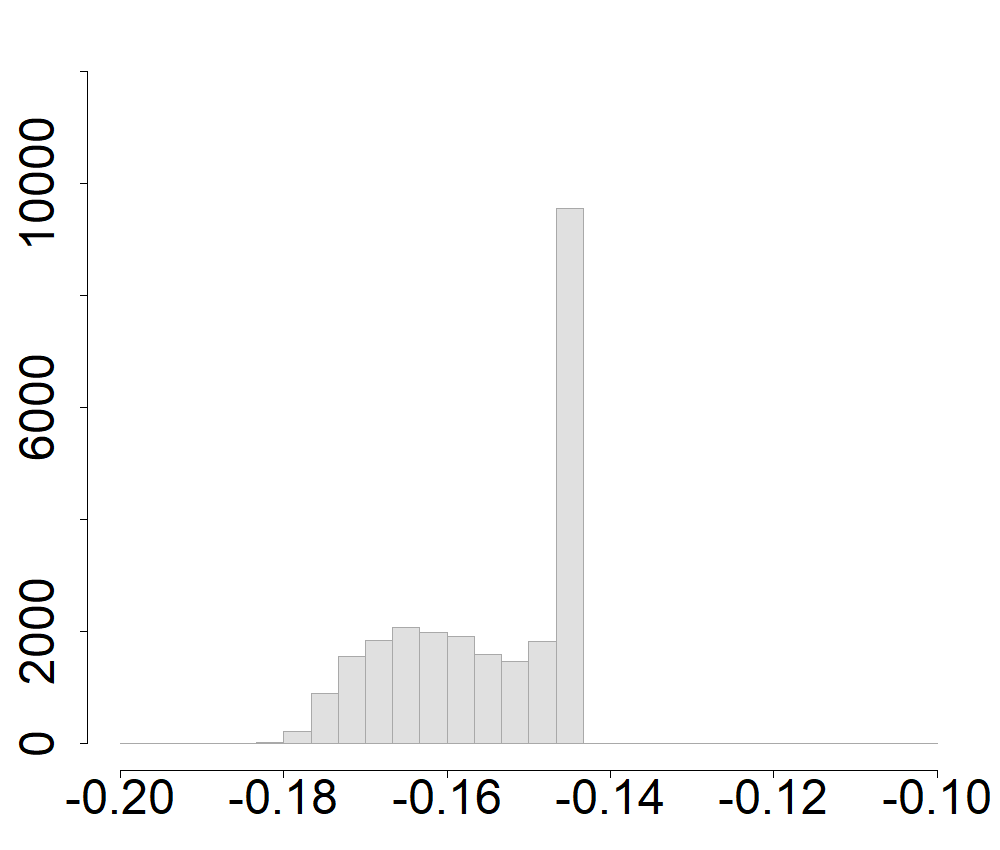}
	&
\includegraphics[trim = 15mm 15mm 0mm 0mm,scale=0.12]{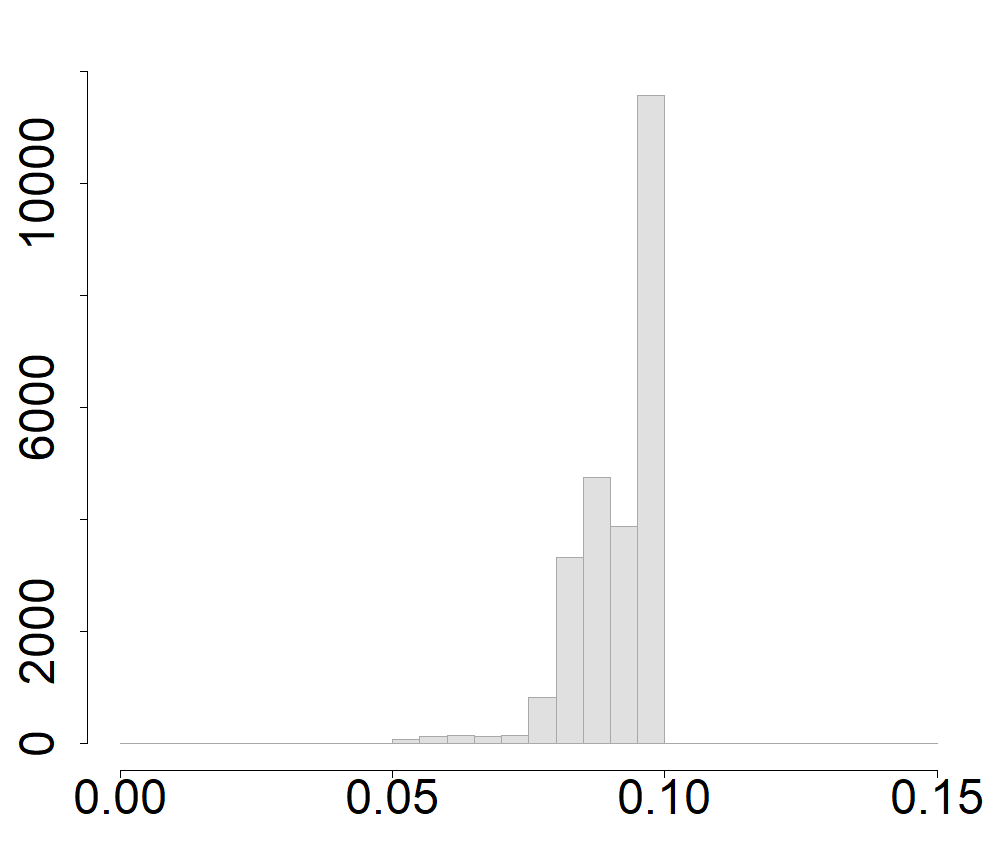}
\\
	\cline{2-5}
	&
	\multirow{1}{*}{\adjustbox{angle=90}{\hspace*{+0.0cm} \large Enemy Recovery}} &
\includegraphics[trim = 15mm 15mm 0mm 0mm,scale=0.12]{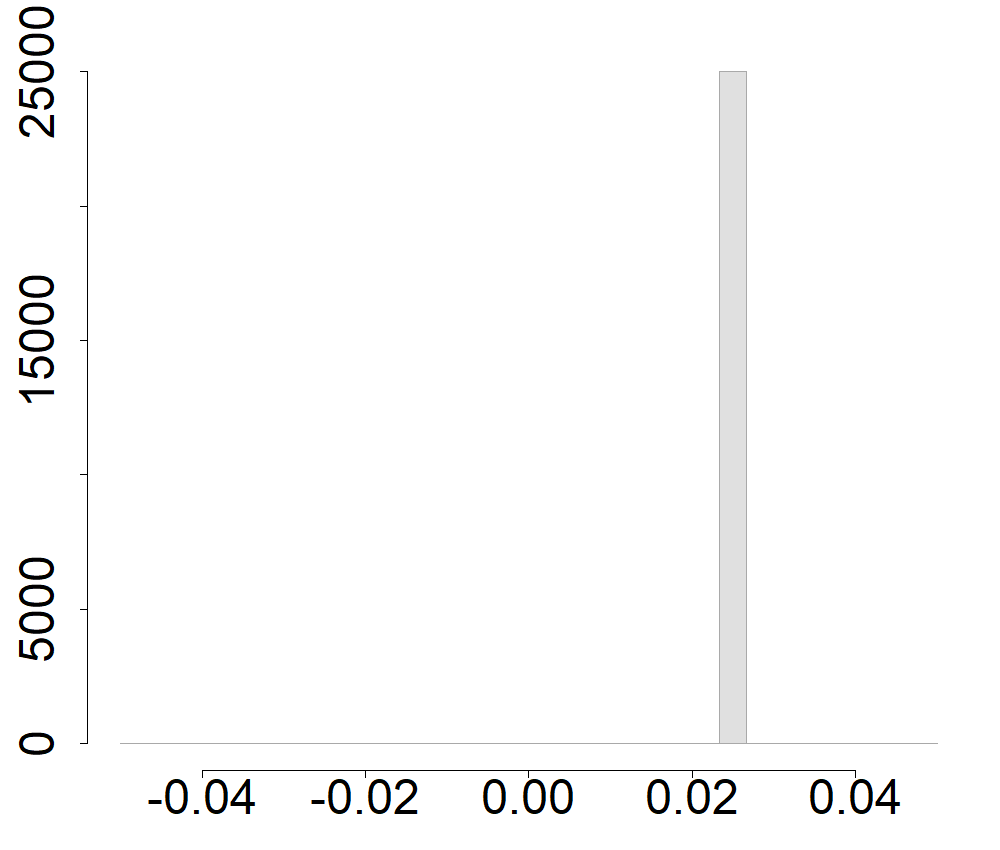}
	&
\includegraphics[trim = 15mm 15mm 0mm 0mm,scale=0.12]{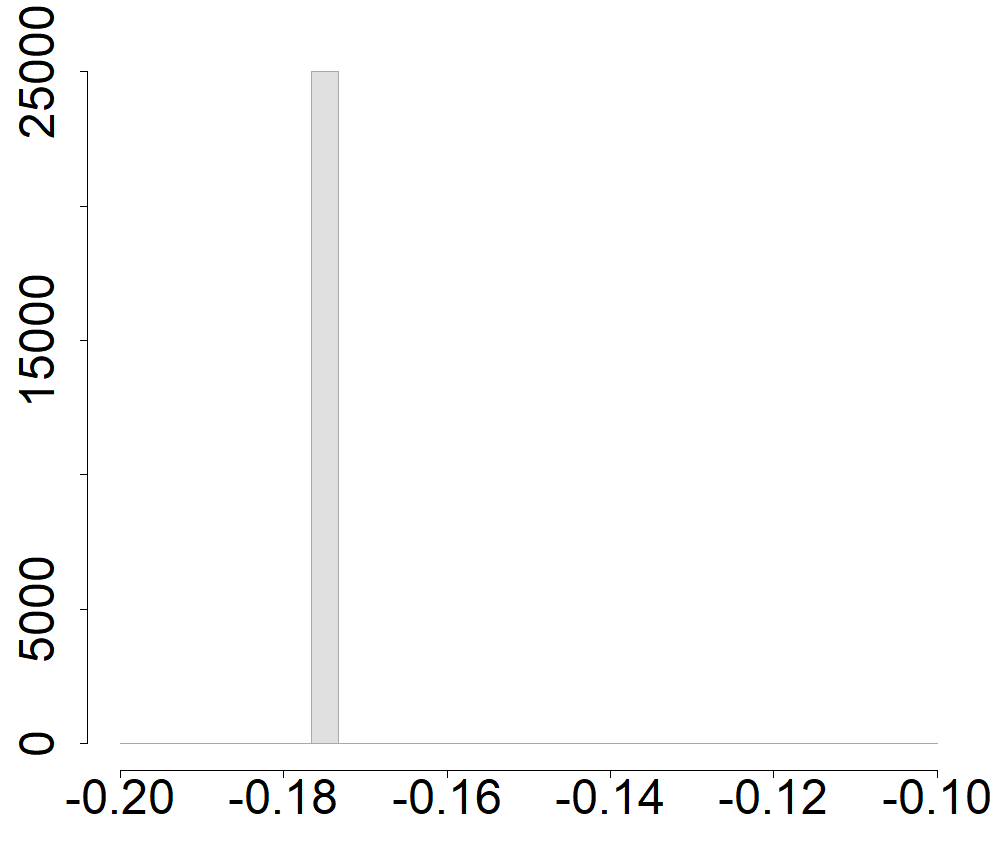}
	&
\includegraphics[trim = 15mm 15mm 0mm 0mm,scale=0.12]{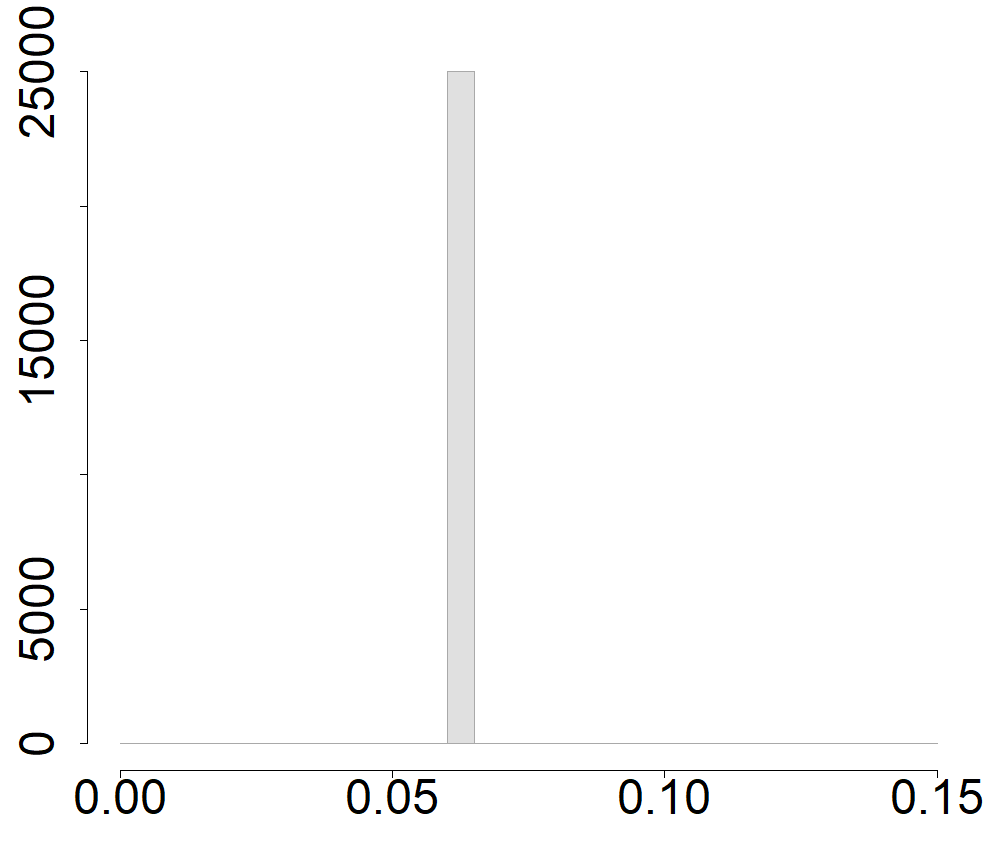}
\\
\hline
\end{tabu}
\end{center}
\caption{
Statistics of the outcomes for the given uncertainty sequence. The outcomes are characterized by the observables $\Delta N$, $\Delta D$, and $L_b$.
\label{sequofhisto}
}
\end{figure}

As a supplement, table~\ref{sequastable} gives means and variances.
The small size of its expeditionary force is a clear disadvantage for Blue.
After fixing force sizes and kinetic parameters in step 2 of the
uncertainty sequence, Blue will in fact even typically lose the battle
according to $\mean(\Delta N) \approx -0.001$. Simultaneously,
the higher kinetic effectiveness of Blue leads to significantly
higher losses of Red in the mean --- from $\mean(\Delta D) \approx -0.016$
in step 1 of the uncertainty sequence to $\mean(\Delta D) \approx -0.149$
in step 2. The continuous decrease of  $\mean(\Delta D)$ resp.\ 
$\mathcal{R} =\mean(L_b)$ in later steps indicates that Blue makes
better progress
in exploiting the cyber aspects of the situation than Red. The decrease
of $\mean(\Delta D) $ is accompanied with an increase of $\mean(\Delta N)$.
Starting with step 3, the value of $\mean(\Delta N)\approx 0.002$
is back in the positive range and increases further to
$\mean(\Delta N)\approx 0.024$ in step 5. 
Concerning the variances, the numbers drop to a value in the range
$10^{-4}-10^{-5}$ in step 2 'force selection', continuously
decreasing further in course of the following steps. Fixing
the force sizes to comparatively small values limits the ranges
of the observables $\Delta N$, $\Delta D$, and $L_b$, which
contributes to the initial drop of the variances.

\begin{table}[tb!]
\tabulinesep=1.0mm
\begin{tabu} to \textwidth {| L{2.7cm} || X[1] | X[1] | X[1] | X[1] | X[1] | X[1] | }
\hline
	Sequence Step & $\mean(\Delta N)$ &  $\mean(\Delta D)$ &  $\mean(L_b)$ &  $\var(\Delta N)$ &  $\var(\Delta D)$ &  $\var(L_b)$   \\ \hline \hline
Start Situation &  0.013 & -0.016 & 0.292 & $2.14\cdot 10^{-1}$ & $1.34\cdot 10^{-1}$ & $6.54\cdot 10^{-1}$ \\ \hline
Force Selection & -0.001 & -0.149 & 0.095 & $2.48\cdot 10^{-4}$ & $2.48\cdot 10^{-4}$ & $1.06\cdot 10^{-4}$ \\ \hline
Cyber Attack    &  0.002 & -0.152 & 0.093 & $1.71\cdot 10^{-4}$ & $1.71\cdot 10^{-4}$ & $7.50\cdot 10^{-5}$ \\ \hline
Enemy Attack    &  0.005 & -0.155 & 0.091 & $9.66\cdot 10^{-5}$ & $9.66\cdot 10^{-5}$ & $5.55\cdot 10^{-5}$ \\ \hline
Enemy Recovery  &  0.024 & -0.175 & 0.063 & $0.000$ & $0.000$ & $0.000$ \\ \hline
\end{tabu}
\caption{
Mean and variance of the observables $\Delta N$, $\Delta D$, and $L_b$ for the given uncertainty sequence.
The risk $\mathcal{R}$ associated with a sequence step is given by $\mathcal{R}=\mean(L_b)$.
\label{sequastable}
}
\end{table}

\begin{table}[tb!]
\tabulinesep=1.0mm
\begin{tabu} to \textwidth {| L{2.7cm} || X[1] | X[1] | X[1] | X[1] | }
\hline
Sequence Step & Strong Wins  & Weak Wins  & Weak Losses  & Strong Losses    \\ \hline \hline
Start Situation & 0.413 & 0.098 & 0.103 & 0.385 \\ \hline
Force Selection & 0.265 & 0.000 & 0.734 & $2.0 \cdot 10^{-4}$ \\ \hline
Cyber Attack    & 0.482 & 0.000 & 0.518 & $1.0 \cdot 10^{-4}$ \\ \hline
Enemy Attack    & 0.545 & 0.000 & 0.455 & 0.000 \\ \hline
Enemy Recovery  & 1.000 & 0.000 & 0.000 & 0.000 \\ \hline
\end{tabu}
\caption{
The fractions of wins and losses --- considered from the perspective of Blue --- along the given sequence of decreasing uncertainty about the expected combat situation.
\label{sequastable2}
}
\end{table}

The fractions of wins and losses are shown in table~\ref{sequastable2}. Blue has cyber superiority from the beginning, and significantly higher kinetic effectiveness starting with step 2 of the uncertainty sequence. Accordingly, weak wins and strong losses become very rare --- the overall technical superiority of Blue and the only moderately larger force of Red make it unlikely that the losses of Blue are higher than the losses of Red. The potential of red options causes a clear drop of the number of strong wins in step 2. Blue demonstrates a clear recovery, however, with a significantly increasing fraction of strong wins along the remaining steps of the uncertainty sequence.

Figure~\ref{sequofhisto} shows only a small part of the 
histograms of the observables $\Delta N$, $\Delta D$, and $L_b$.
The complete picture for the start situation of the uncertainty 
sequence is provided in figure~\ref{sequofhistoforcyb}.
The distribution of $\Delta D$ has a pronounced peak 
close to $\Delta D\approx 0$ contrary to $\Delta N$.
This preference of small values may be caused by a cancellation effect.
Typically, both winner and loser will suffer losses, whereas
the loser has no surviving force elements at all. Accordingly,
the difference $\Delta D$ of the losses of the opponents will
typically cancel out large numbers. For $\Delta N$
representing the difference of survivors, such a cancellation effect 
does not occur.

\begin{figure}[tb!]
\tabulinesep=0.7mm
\begin{center}
\begin{tabu} to \textwidth{| p{0.32\textwidth} |p{0.32\textwidth} |  p{0.32\textwidth} | }
\hline
	\multicolumn{1}{|c|}{\Large $\Delta N$} & \multicolumn{1}{c|}{\Large $\Delta D$} & \multicolumn{1}{c|}{\Large $L_b$} \\ \hline 
\includegraphics[scale=0.12]{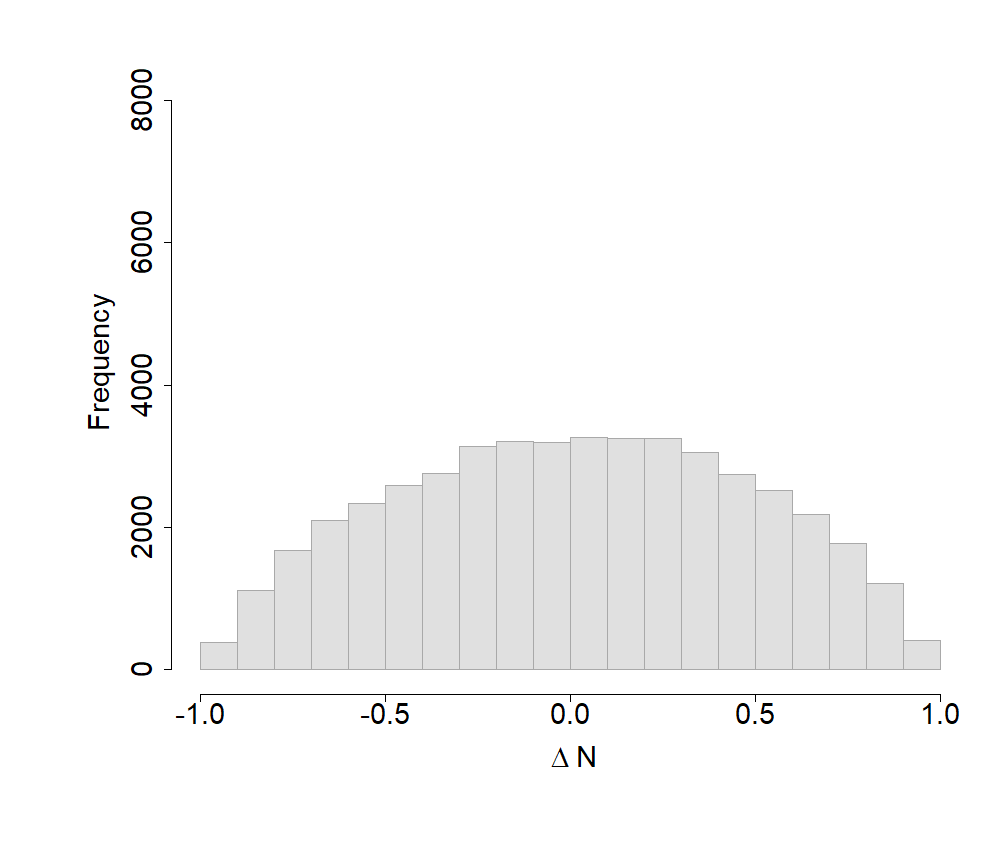}
	&
\includegraphics[scale=0.12]{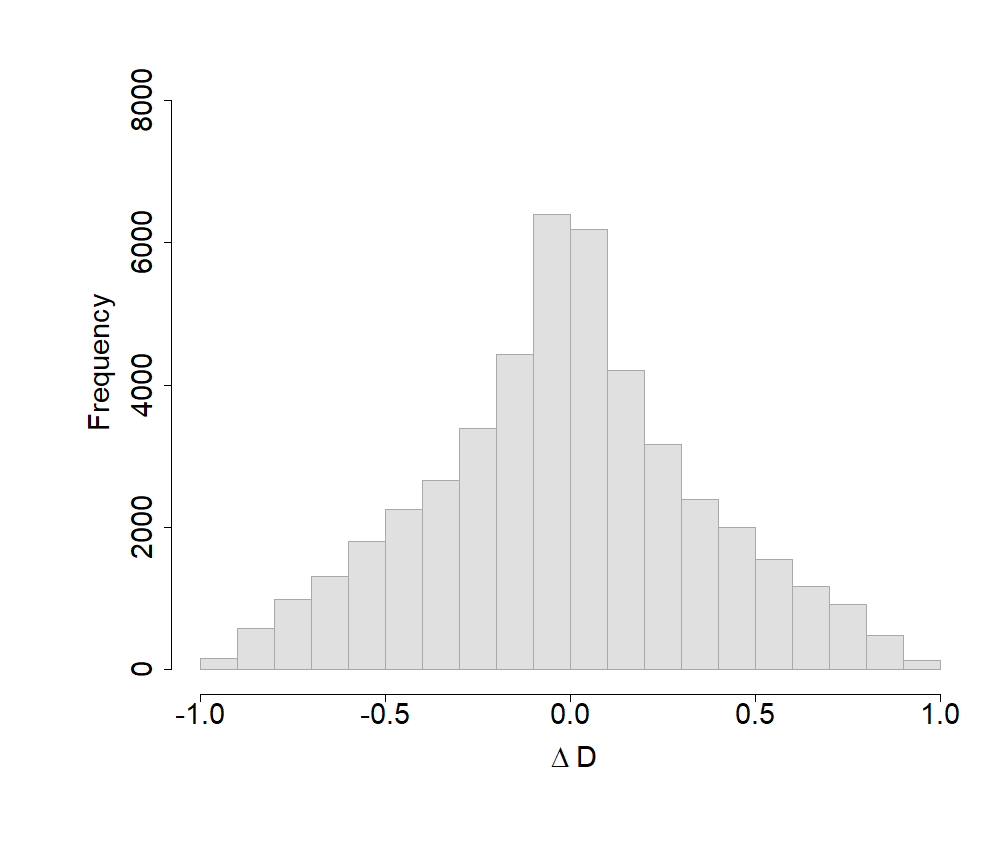}
	&
\includegraphics[scale=0.12]{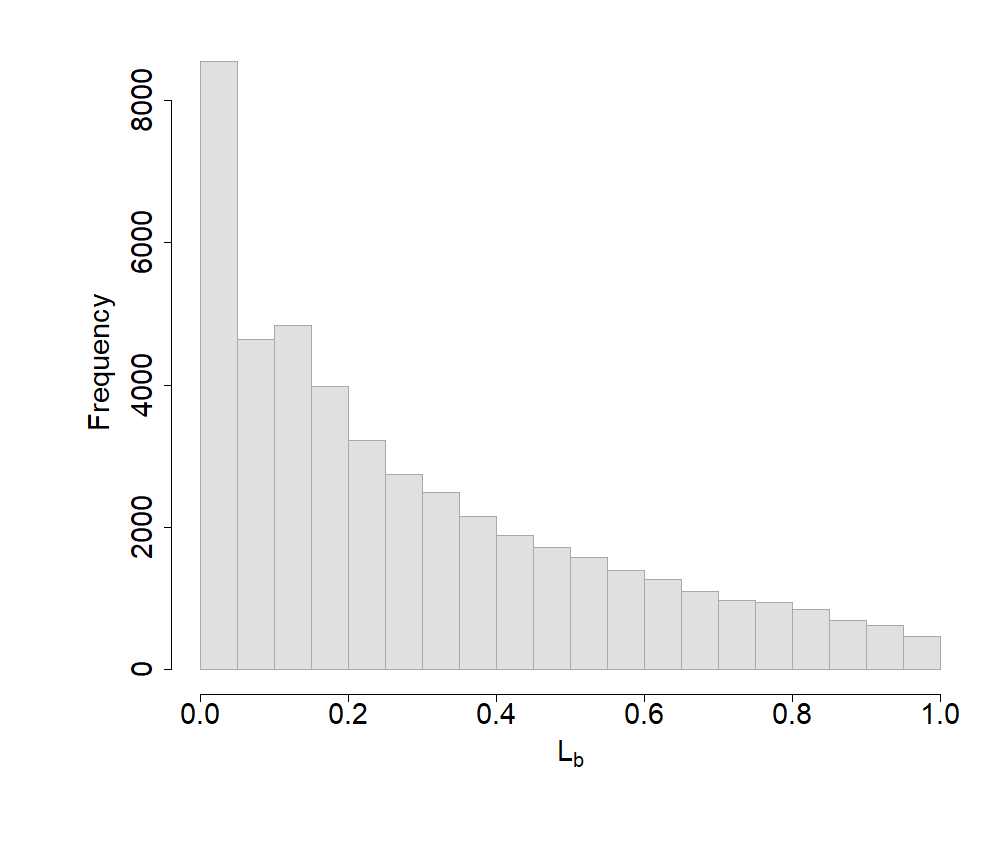}
\\
\hline
\end{tabu}
\end{center}
\caption{
	Histograms of the values of the observables $\Delta N$, $\Delta D$, and $L$ for the start situation of the uncertainty sequence. 
\label{sequofhistoforcyb}
}
\end{figure}

\subsection{Influence of the Model}\label{ssim}

Modifications of the model $M$ representing the given situation may 
lead to distinct changes of the results.
For a demonstration, table~\ref{compastable} compares
the model of section~\ref{secfour} (designated as 'gen') and 
the slightly simpler model of Yildiz \cite{yildiz2014modeling}
(designated as 'ref') based on the means of the observables 
$\Delta N$, $\Delta D$, and $L_b$.
Due to the explicit representation of the malware attack in (\ref{intermediatestep3}), the model of Yildiz can only essentially, but not fully, embedded in (\ref{intermediatestep3}). Taking this restriction into account, a scenario $x\in X$ can be approximatively mapped to the model of Yildiz by applying the parameter settings $\alpha=1$, $\Delta t_{\mathsf{mal}}= 1$, $p,q=1$, $\tilde\gamma=\gamma$, and $\Delta t_{\mathsf{att}}, \Delta t_{\mathsf{pat}} =0$.

The differences between the results of 'gen' and 'ref' are notable in some stages of the uncertainty sequence. The largest difference can be found with the predicted blue losses at the end of the sequence --- $\mathcal{R} = \mean(L_b) \approx 0.064$ for 'gen' vs.\ $\mathcal{R}=\mean(L_b) \approx 0.112$ for 'ref'. Beyond that, $\mean(\Delta N) <0$ for 'gen' and $\mean(\Delta N) >0$ for 'ref' in step 2 leads to a win/loss classification mismatch (at least in the average). 

We can state that using a refined and more detailed model $M$ may have notable consequences for the analysis results. The proposed framework offers support by allowing to choose the most suitable model $M\in \mathcal{M}$ instead of using a fixed single model. 
The results for point estimates also contained in table~\ref{compastable} are discussed in section~\ref{secinfuncertain}.

\begin{table}[tbhp!]
\tabulinesep=1.0mm
\begin{tabu} to \textwidth {| L{2.7cm} | L{2.2cm} || X[1] | X[1] | X[1] | X[1] | X[1] | X[1] | }
\hline
	Sequence Step &Uncertainties? & \multicolumn{2}{c|}{$\mean(\Delta N)$} &  \multicolumn{2}{c|}{$\mean(\Delta D)$} &  \multicolumn{2}{c|}{$\mean(L_b)$} \\ \cline{3-8}
	&& gen & ref & gen & ref & gen & ref \\
	\hline \hline
	\multirow{2}{*}{Start Situation} & sampling      &  0.013 &  0.011 & -0.016 & -0.014 &  0.292 &  0.338 \\ \cline{2-8}
	                                 & point estimate&  0.006 &  0.020 & -0.010 & -0.024 &  0.447 &  0.525 \\ \hline
	\multirow{2}{*}{Force Selection} & sampling      & -0.001 &  0.032 & -0.149 & -0.183 &  0.095 &  0.116 \\ \cline{2-8}
	                                 & point estimate& -0.004 &  0.032 & -0.146 & -0.182 &  0.097 &  0.117 \\ \hline
	\multirow{2}{*}{Cyber Attack}    & sampling      &  0.001 &  0.034 & -0.152 & -0.185 &  0.093 &  0.113 \\ \cline{2-8}
	                                 & point estimate& -0.003 &  0.032 & -0.147 & -0.182 &  0.097 &  0.116 \\ \hline
	\multirow{2}{*}{Enemy Attack}    & sampling      &  0.006 &  0.039 & -0.156 & -0.189 &  0.091 &  0.109 \\ \cline{2-8}
	                                 & point estimate& -0.002 &  0.035 & -0.148 & -0.185 &  0.096 &  0.113 \\ \hline
	\multirow{2}{*}{Enemy Recovery}  & sampling      &  0.025 &  0.037 & -0.175 & -0.187 &  0.064 &  0.112 \\ \cline{2-8}
	                                 & point estimate&  0.025 &  0.037 & -0.175 & -0.187 &  0.064 &  0.112 \\ \hline
\end{tabu}
\caption{
Importance of the modeling details and of the inclusion of uncertainties.
\label{compastable}
}
\end{table}


\subsection{Influence of Random Cyber Actions}\label{anares1}\label{sscy}

\begin{table}[tbh!]
\tabulinesep=1.0mm
\setlength{\tabcolsep}{3pt}
\centering
\begin{tabu} to \textwidth {| L{0.5cm} | L{1.9cm} | L{3.7cm} | X[1] |}
	\hline
	\textbf{Id} & \textbf{Parameter Settings} & \textbf{Designation} & 
		\textbf{Description} \\
        \hline 
        gen & - & General case, i.e. two-sided cyber support & No additional constraints \\
	\hline
	cyb & $\eta_{b}=1$ & One-sided cyber support &
	The blue side is not affected by malware, i.e. 
	Red is fighting only at the kinetic level\\
        \hline
	kin & $\eta_{b},\eta_{r}=1$ & Pure Lanchester case &
	Pure kinetic combat without cyber component at all \\
	\hline
\end{tabu}
	\caption{List of scenario classes.
\label{definition_cases}
}
\end{table}

For assessing the influence of cyber combat, results reached with
one- or two-sided cyber support are compared with results of
pure kinetic scenarios.
Corresponding scenario classes 'gen' (two-sided cyber warfare),
'cyb' (one-sided cyber warfare), and 'kin' (without cyber warfare),
which define subsets of the 
overall universe of possible situations with help of parameter 
constraints, are given in table~\ref{definition_cases}.
As a first application, the simulation outcomes 
for the start situation of the uncertain sequence.
are broken down to the four kinetic win/loss classes defined
in definition~\ref{defkinsup} in table~\ref{startastable3}.
As one may expect, $\Delta N$ decreases along the sequence
'strong kinetic superiority', 'weak kinetic superiority',
'weak kinetic inferiority', 'strong kinetic inferiority'.
This holds for the scenario classes 'gen', 'cyb', and 'kin' likewise. 
For $\Delta D$, the values are correspondingly increasing
with the two middle positions interchanged. 
According to the definition of weak kinetic inferiority
with $\Delta D <0$ resp.\ superiority with $\Delta D >0$,
the risk for Blue is significantly higher in situations with weak
kinetic superiority than in situations with weak kinetic inferiority. 
As expected, Blue profits from one-sided cyber capabilities
in the scenario class 'cyb'. The losses are smaller and
the fraction of surviving force elements higher when compared to
scenario classes 'gen' and 'kin'.  

\begin{table}[tb!]
\setlength{\tabcolsep}{5pt}
\tabulinesep=1.0mm
\begin{tabu} to \textwidth {| L{2.3cm} || X[1] | X[1] | X[1] | X[1] | X[1] | X[1] | X[1] | X[1] | X[1] | }
\hline
Win/Loss & \multicolumn{3}{c|}{$\mean(\Delta N)$}  & \multicolumn{3}{c|}{$\mean(\Delta D)$}  & \multicolumn{3}{c|}{$\mean(L_b)$}    \\ \cline{2-10}
Class 	& gen & cyb & kin & gen & cyb & kin & gen & cyb & kin \\ \hline \hline
Strong kinetic superiority & 0.423 & 0.431 & 0.415 & -0.322 & -0.331 & -0.315 & 0.145 & 0.137 & 0.153 \\ \hline
Weak kinetic superiority   & 0.282 & 0.295 & 0.265 & 0.124 & 0.112 & 0.142 & 0.325 & 0.314 & 0.342 \\ \hline
Weak kinetic inferiority   & -0.238 & -0.226 & -0.260 & -0.165 & -0.176 & -0.143 & 0.198 & 0.196 & 0.205 \\ \hline
Strong kinetic inferiority & -0.406 & -0.398 & -0.423 & 0.295 & 0.286 & 0.312 & 0.455 & 0.454 & 0.460 \\ \hline 
\end{tabu}
\caption{
The means of the observables $\Delta N$, $\Delta D$, and $L_b$ for the start situation of the given uncertainty sequence for the scenario classes 'gen', 'cyb', and 'kin'.
\label{startastable3}
}
\end{table}

Situations, in which the winning side changes due to the usage of
malware, are of special interest. Table~\ref{startastable2} takes up
this issue.
As expected, malware can turn around a situation.
Strong wins are possible for all kinetic win/loss classes
for both 'gen' and 'cyb'. Cyber support is obviously able
to compensate kinetic inferiority. Since one-sided cyber support is a 'true' support --- the situation cannot worsen by the usage of malware impeding hostile force elements --- some transitions of the win/loss classification are excluded. Kinetic superiority, for example, cannot lead to lost outcomes anymore. Furthermore, strong kinetic superiority in 'cyb' assures a strong win.  Thus, a larger number of winning changes can be observed for the scenario class 'cyb' than for 'gen'.

\begin{table}[tb!]
\tabulinesep=1.0mm
\begin{tabu} to \textwidth {| L{2.7cm} || X[1] | X[1] | X[1] | X[1] | X[1] | X[1] | X[1] | X[1] | }
\hline
	Win/Loss Class & \multicolumn{2}{c|}{Strong Wins}  & \multicolumn{2}{c|}{Weak Wins}  & \multicolumn{2}{c|}{Weak Losses}  & \multicolumn{2}{c|}{Strong Losses}    \\ \cline{2-9}
	& gen & cyb & gen & cyb & gen & cyb & gen & cyb \\ \hline \hline
Strong kinetic superiority & 0.980 & 1.000 & 0.009 & 0.000 & 0.008 & 0.000 & 0.003 & 0.000 \\ \hline
Weak kinetic superiority   & 0.085 & 0.093 & 0.891 & 0.907 & 0.000 & 0.000 & 0.023 & 0.000 \\ \hline
Weak kinetic inferiority   & 0.071 & 0.075 & 0.000 & 0.000 & 0.891 & 0.925 & 0.037 & 0.000 \\ \hline
Strong kinetic inferiority & 0.015 & 0.015 & 0.012 & 0.016 & 0.023 & 0.025 & 0.950 & 0.944 \\ \hline
\end{tabu}
\caption{
	The fractions of wins and losses (from the perspective of Blue) in  the start situation of the given uncertainty sequence. 
\label{startastable2}
}
\end{table}

\subsection{Influence of Optimized Cyber Actions}\label{sscy2}

In the considered example, cyber support leads to fundamental
changes of the outcome (see table~\ref{startastable2})
only in a comparatively small number of cases. This is partially 
caused by a randomized parameterization of cyber actions.
Unexpected frictions or technical breakthroughs, the hardly 
predictable time required for developing malware and patches, and
unforeseen activities of the opponent motivate this simplification.
Notwithstanding, military forces will try to exploit any situation
to their own advantage in the best possible way. Accordingly, 
one may expect some kind of optimization in practice.
We will thus evaluate its potential in the following. At first, we take a
closer look at the changes of the outcomes due to cyber-related actions
in dependence on their parameterization. For this purpose, we define 
$\Delta\Delta N := 
\Delta N (x) - \Delta N (g_{\mathsf{kin}} (x))$ and analogously
$\Delta\Delta D$, $\Delta L_b$. 

\begin{figure}[tbhp!]
\tabulinesep=0.7mm
\begin{center}
\begin{tabu} to \textwidth{| C{0.5cm} | X[1] | X[1] | X[1] |}
\hline
	& \multicolumn{1}{c|}{\Large $\Delta \Delta N$} & \multicolumn{1}{c|}{\Large $\Delta \Delta D$} & \multicolumn{1}{c|}{\Large $\Delta L_b$} \\ \hline \hline
	\multirow{1}{*}{\adjustbox{angle=90}{\hspace*{+0.0cm} \Large Scenario Class 'gen'}} &
\includegraphics[scale=0.12]{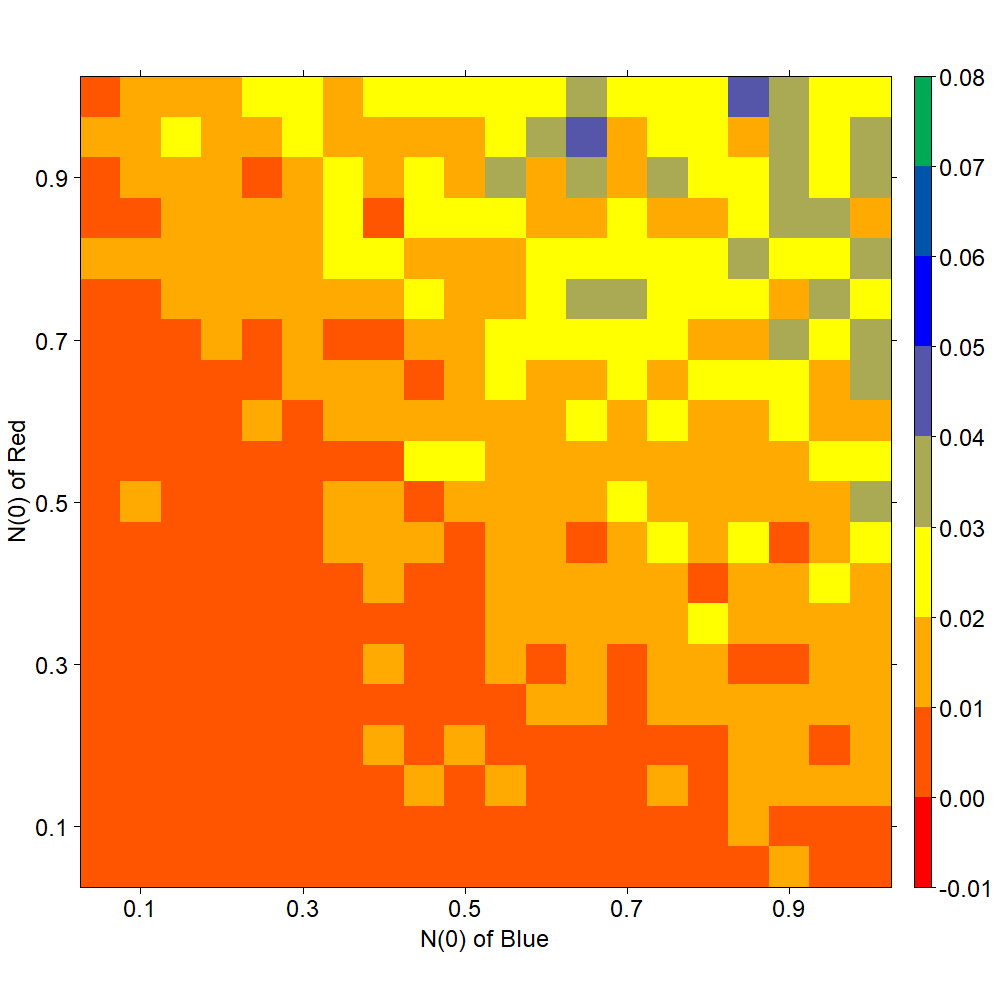}
	&
\includegraphics[scale=0.12]{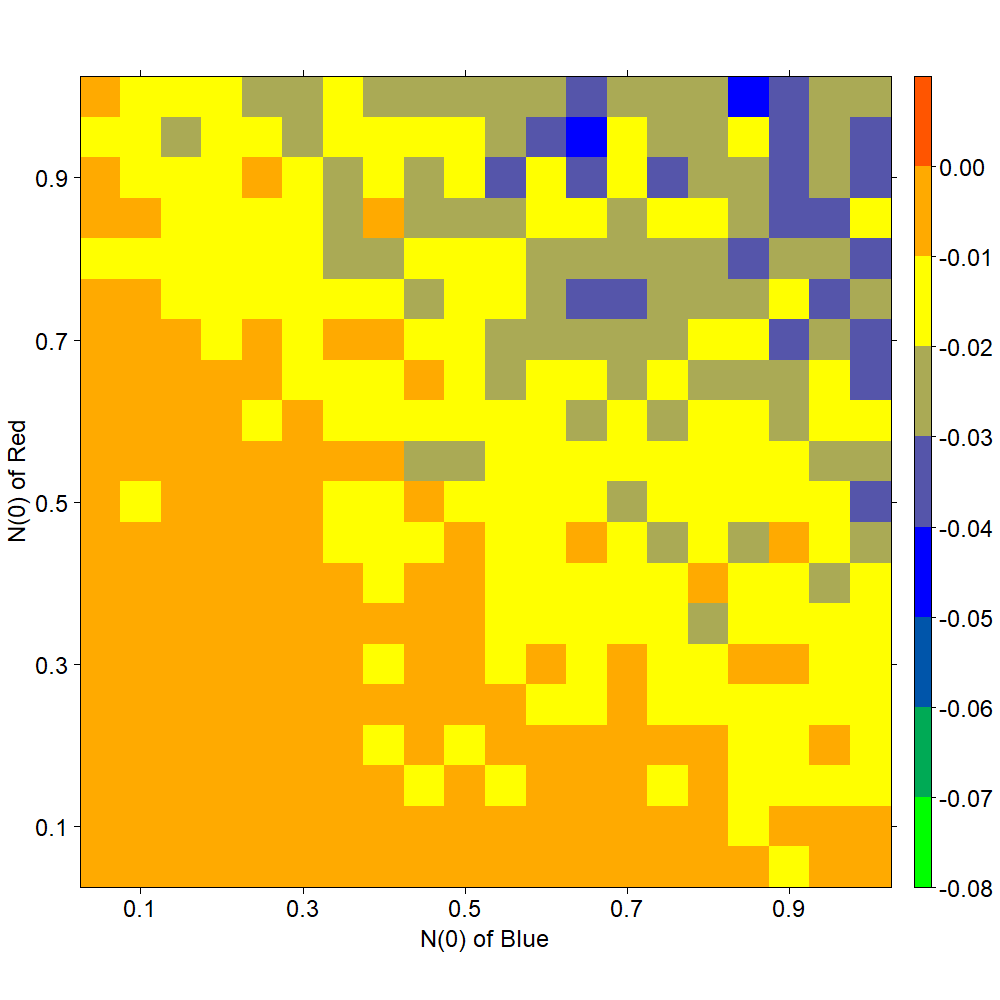}
	&
\includegraphics[scale=0.12]{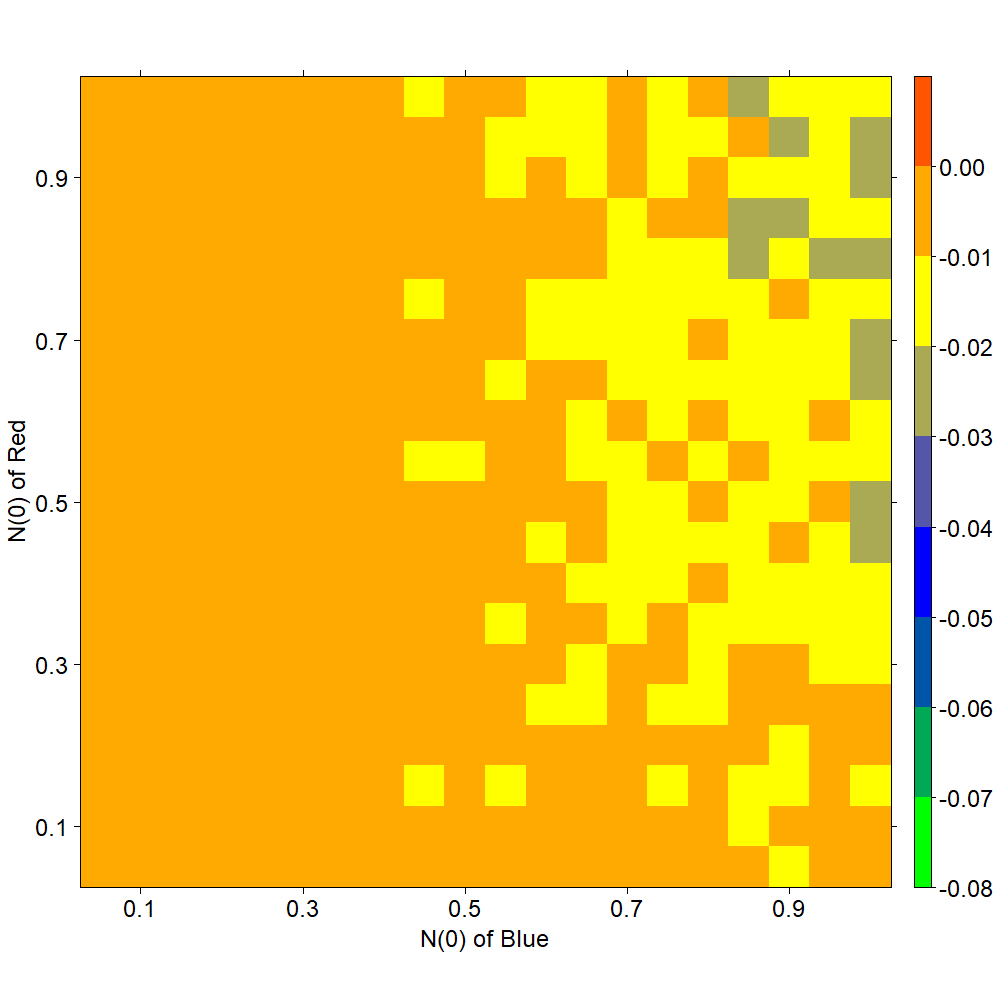}
\\
\hline
	\multirow{1}{*}{\adjustbox{angle=90}{\hspace*{+0.0cm} \Large Scenario Class 'cyb'}} &
\includegraphics[scale=0.12]{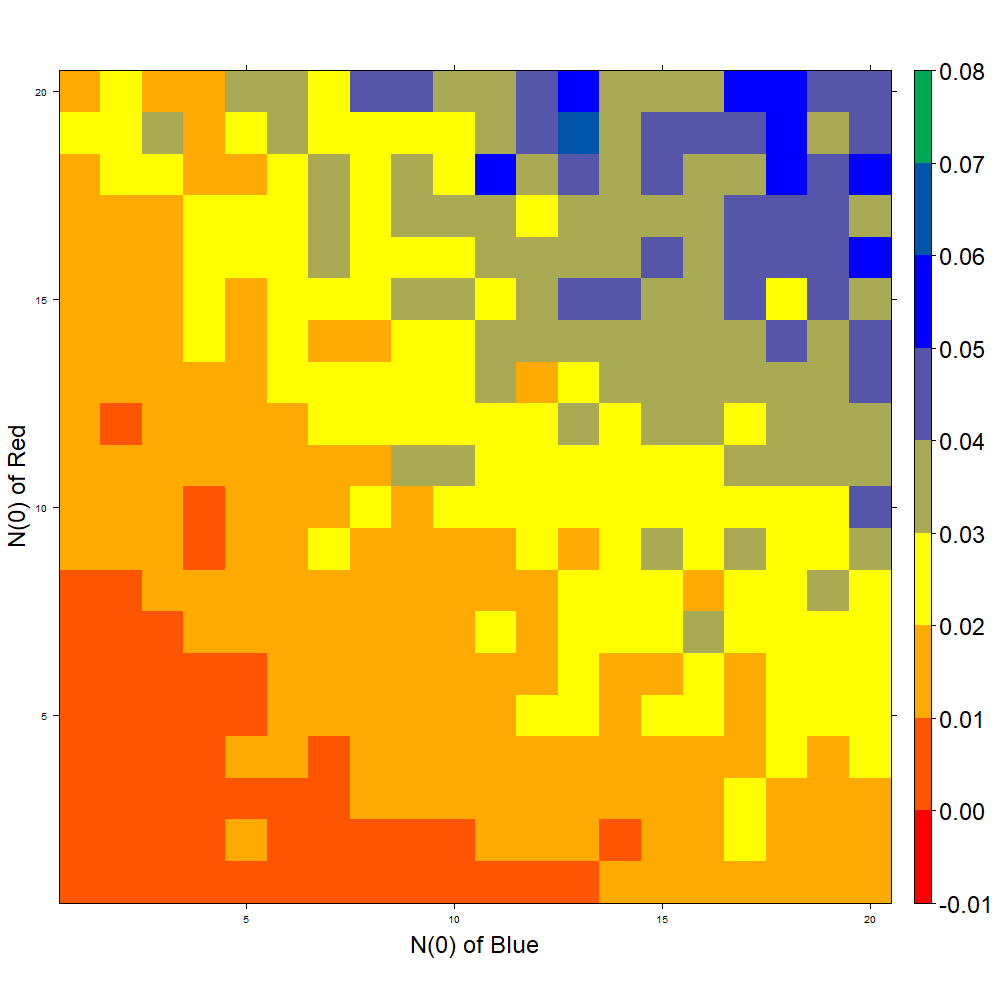}
	&
\includegraphics[scale=0.12]{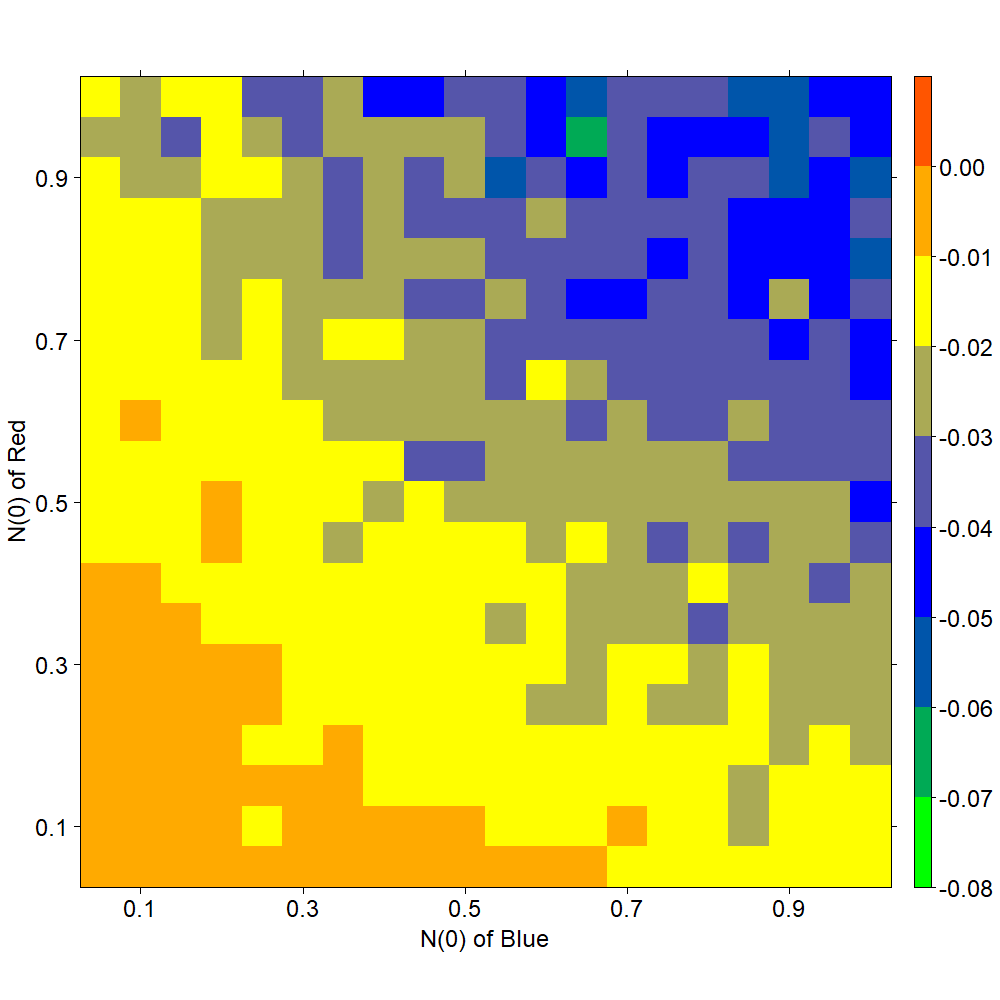}
	&
\includegraphics[scale=0.12]{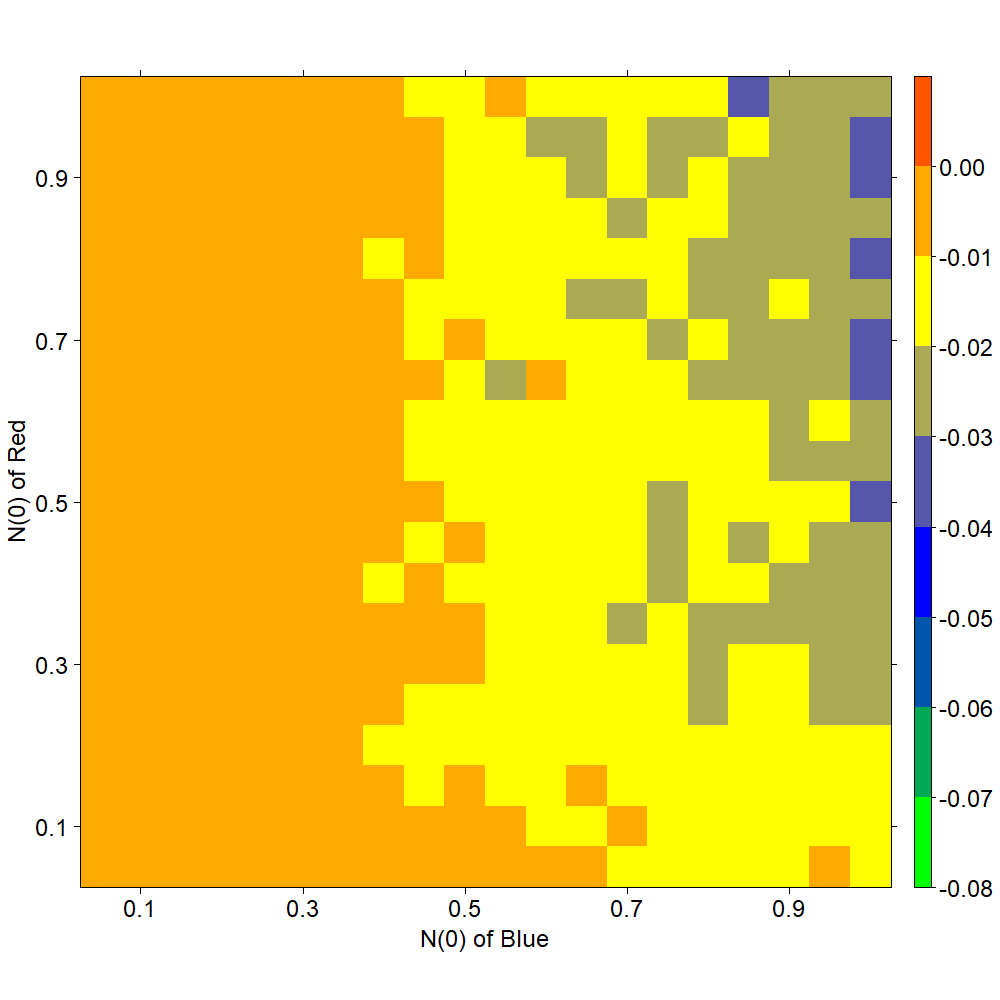}
\\
\hline
\end{tabu}
\end{center}
\caption{
	Cyber support as an option for compensating kinetic inferiority. The heatmaps show the means of $\Delta \Delta N$, $\Delta \Delta D$, and $\Delta L_b$ dependent on the force sizes $N_b(0)$ vs.\ $N_r(0)$ for one- and two-sided cyber support ('cyb' and 'gen').
\label{heatofoneside}
}
\end{figure}

\begin{figure}[tbhp!]
\tabulinesep=0.7mm
\begin{center}
\begin{tabu} to \textwidth{| C{0.5cm} | X[1] | X[1] |}
\hline
	& \multicolumn{1}{c|}{\Large $\Delta \Delta N$} & \multicolumn{1}{c|}{\Large $\Delta L_b$} \\ \hline \hline
	\multirow{1}{*}{\adjustbox{angle=90}{\hspace*{+0.0cm} \Large $\beta_b$}} &
\includegraphics[scale=0.12]{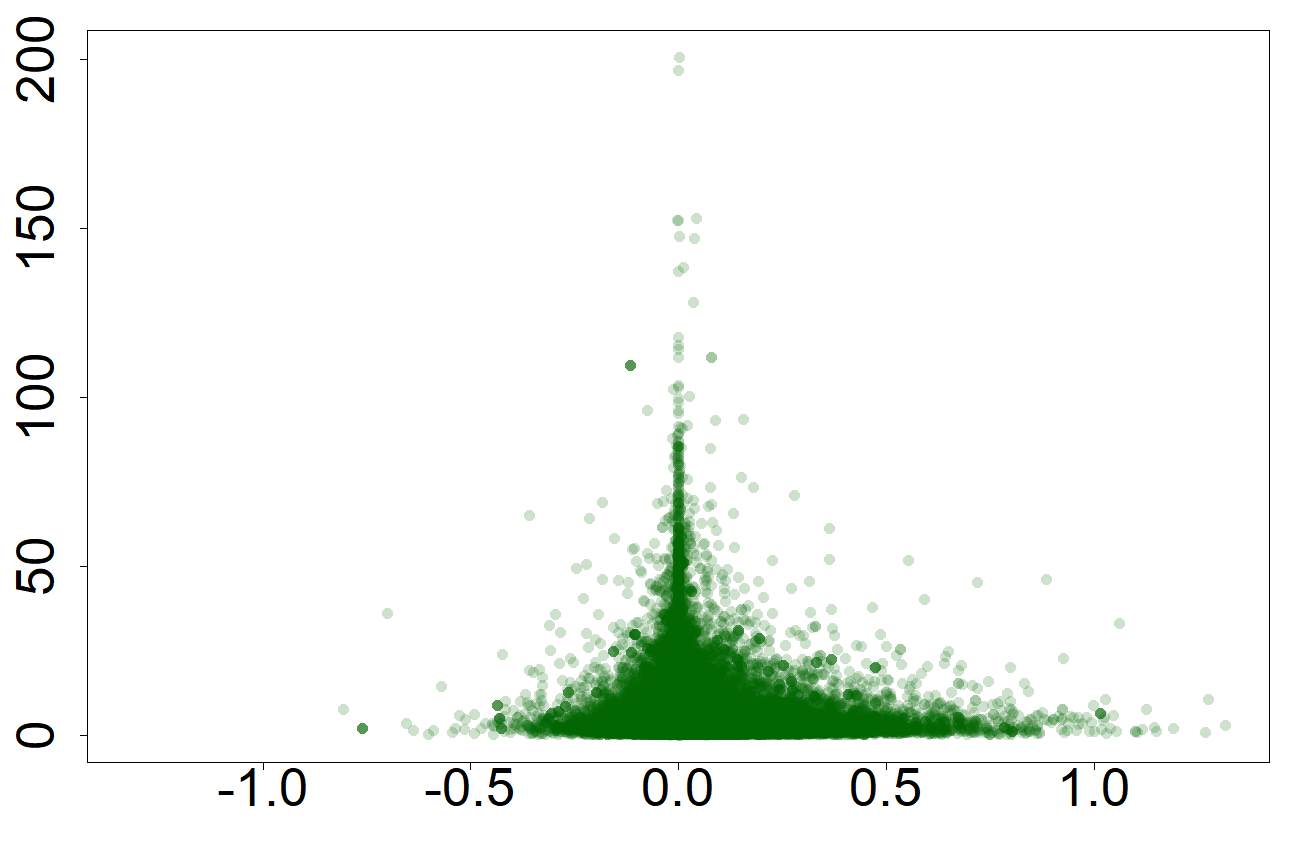}
	&
\includegraphics[scale=0.12]{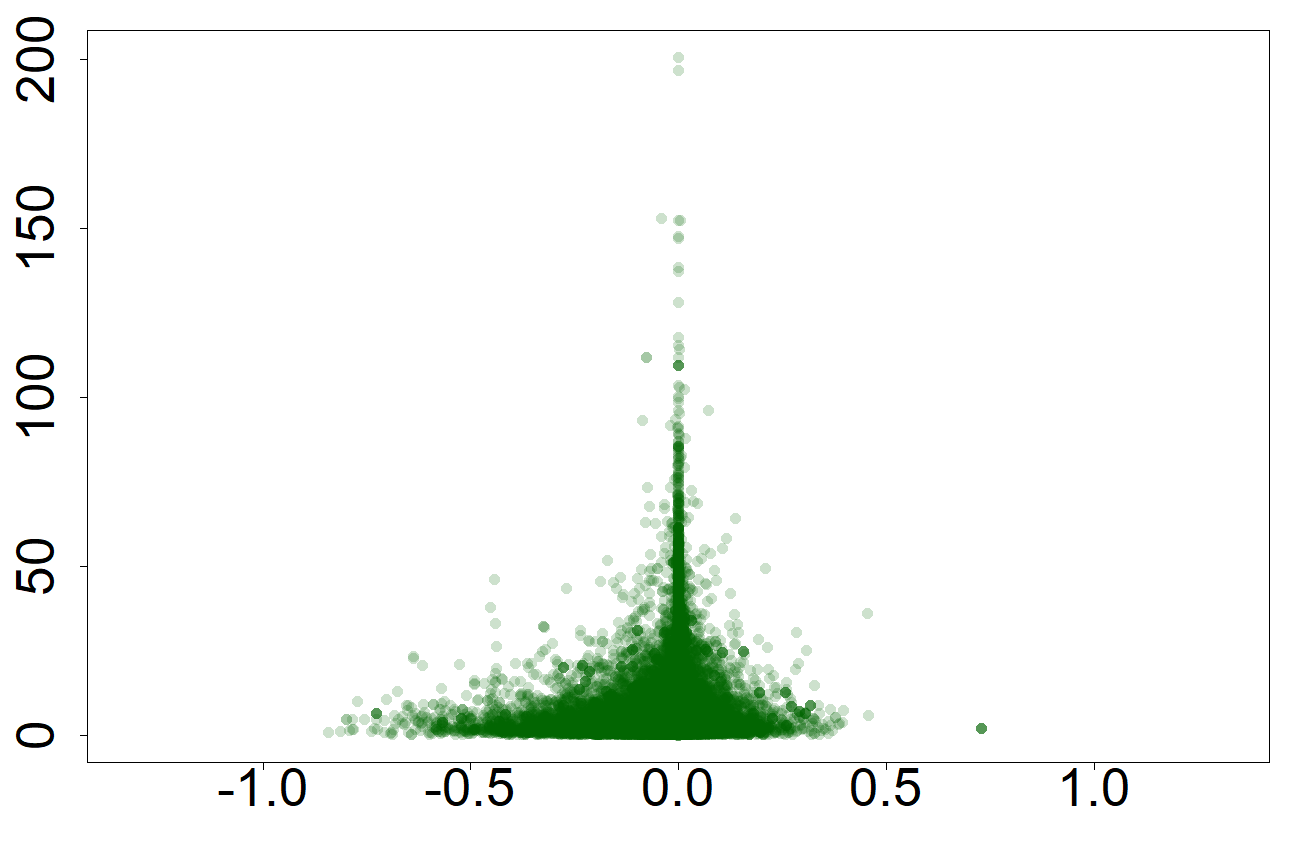}
\\ \hline
	\multirow{1}{*}{\adjustbox{angle=90}{\hspace*{+0.0cm} \Large $\Delta t_{b,\mathsf{att}}$}} &
\includegraphics[scale=0.12]{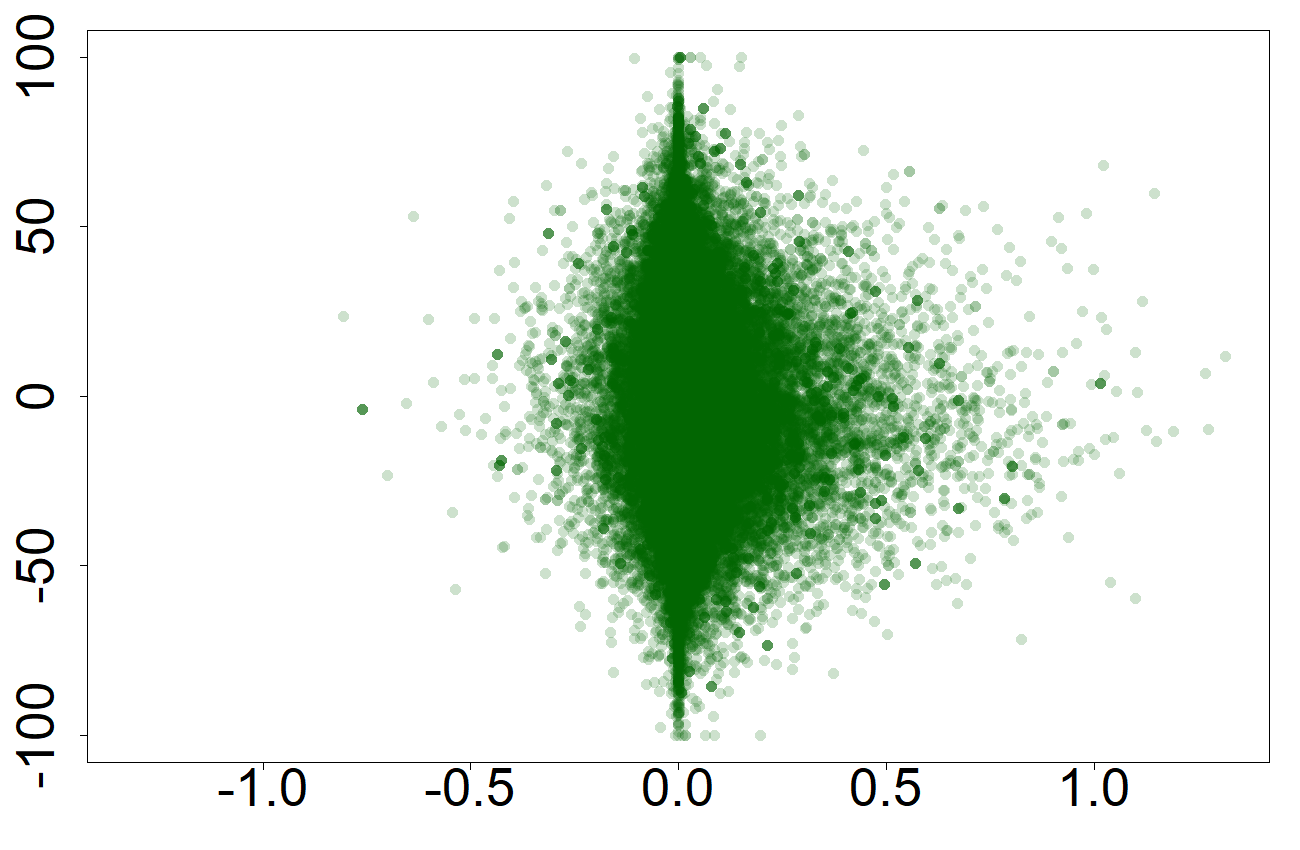}
	&
\includegraphics[scale=0.12]{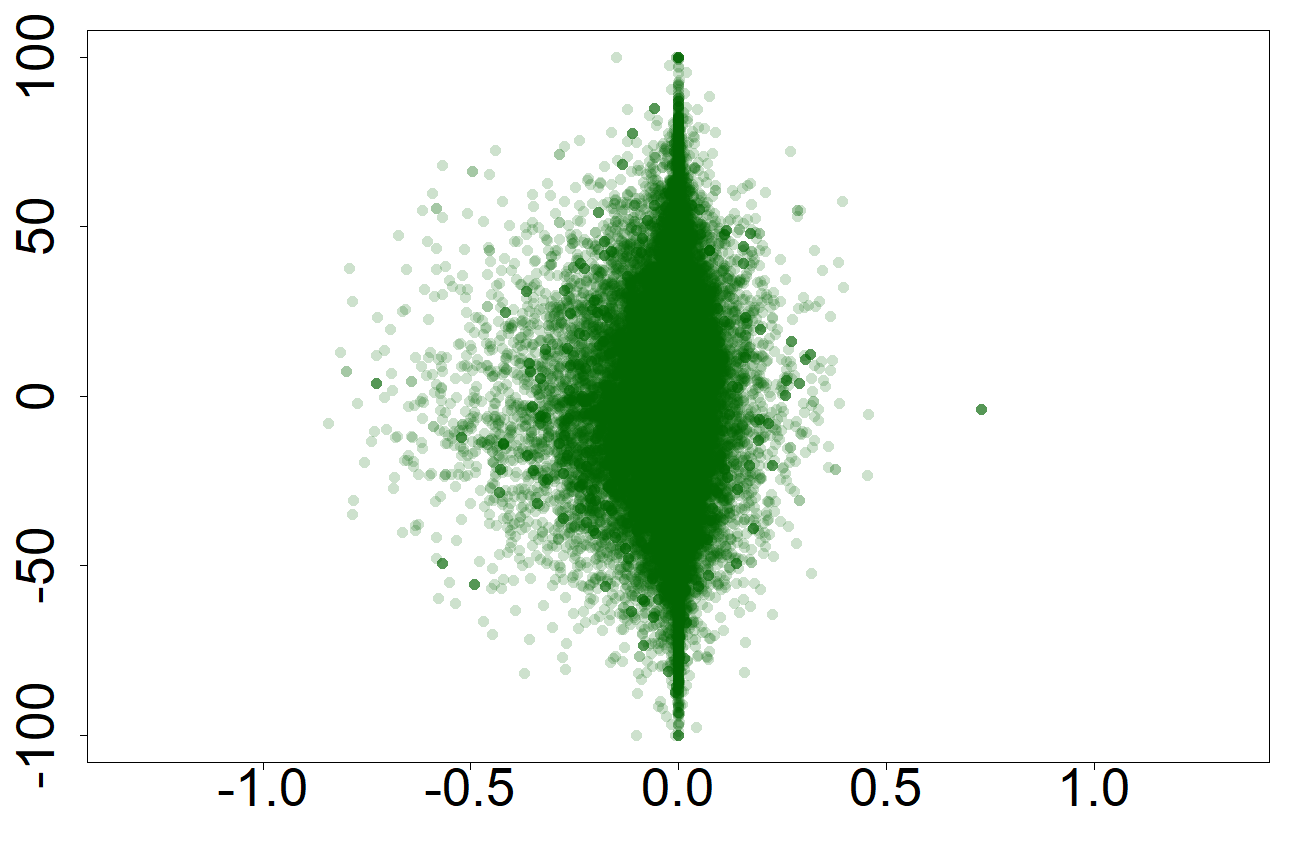}
\\
\hline
\end{tabu}
\end{center}
\caption{
Cyber support may change the outcome of kinetic combat. The scatterplots demonstrate the influence of the input parameter $\beta_b$ resp. $\Delta t_{b,\mathsf{att}}$ on the magnitude of the changes $\Delta \Delta N$ and $\Delta L_b$ in the scenario class 'gen'.
\label{scatterforgenandcyb}
}
\end{figure}

Figure~\ref{heatofoneside} shows the means of 
$\Delta\Delta N$, $\Delta\Delta D$, and $\Delta L_b$ dependent on 
the initial force sizes $N_b(0)$, $N_r(0)$.
The changes are distinctively larger for one-sided cyber 
support (scenario class 'cyb') than for two-sided support 
(scenatio class 'gen').
More important, some settings of the 
parameters $N_b(0)$, $N_r(0)$ give higher chances of a
large change than others. This observation is supported by
figure~\ref{scatterforgenandcyb}, which shows the influence of
infection rate $\beta_b$ and timing $\Delta t_{b,\mathsf{att}}$ 
on $\Delta \Delta N$ resp.\
$\Delta L_b$. Though in large regions of the parameter space the
influence is typically small or even neglectable, $\Delta \Delta N$
and $\Delta L_b$ can reach large values in others. Summing up,
we can confirm comparatively small effects of randomly configured
cyber actions, but we can expect potentially large effects on the outcome with
appropriately refined parameter values. 

\begin{table}[tb!]
\tabulinesep=0.3mm
\setlength{\tabcolsep}{1.15mm}
\begin{tabular}{| p{0.8cm} || l l l | l l l | l l l | l l l | }
\hline
\multicolumn{13}{|c|}{Definition of Subregions of the Scenario Space 'gen'} \\
\hline
	gen0 & $p,q$ & $\in$ & $[0, \infty)$ & $\delta_b,\delta_r$ & $\in$ & $[0,\infty)$ & & & & & & \\ 
	gen1 &&& $[0,7]$ &&& $[0,3.5]$ & & & & & & \\ 
	gen2 &&& $[0,5]$ &&& $[0,2.5]$ & & & & & & \\ \hline
	gen0 & $\beta_b,\beta_r$ & $\in$ & $[0,\infty)$ & $\eta_b$ & $\in$ & $[0,1]$ & $\eta_r$ & $\in$ & $[0,1]$ & & & \\ 
	gen1 &&& $[0,17.5]$ &&& $[0.3,1]$ &&& $[0,0.7]$ & & & \\ 
	gen2 &&& $[0,12.5]$ &&& $[0.5,1]$ &&& $[0,0.5]$ & & & \\ \hline
	gen0 & $\Delta t_{b,\mathsf{att}}$ & $\in$ & $(-\infty,\infty)$ & $\Delta t_{r,\mathsf{att}}$ & $\in$ & $(-\infty,\infty)$ & $\Delta t_{b,\mathsf{pat}}$ & $\in$ & $(-\infty,\infty)$ & $\Delta t_{r,\mathsf{pat}}$ & $\in$ & $(-\infty,\infty)$  \\ 
	gen1 &&& $[-80,120]$ &&& $[-120,80]$ &&& $[-90,60]$ &&& $[-60,90]$  \\ 
	gen2 &&& $[-40,80]$ &&& $[-80,40]$ &&& $[-60,30]$ &&& $[-30,60]$ \\ \hline
	gen0 & $R_b(0)$ & $\in$ & $[0,N_b(0)]$ & $R_r(0)$ & $\in$ & $[0,N_r(0)]$ & $\gamma_b,\gamma_r$ & $\in$ & $[0,1]$ & $\gamma_b',\gamma_r'$ & $\in$ & $[0,1]$ \\ 
	gen1 &&& $[0,1]$ &&& $[0,0.35]$ &&& $[0,0.7]$ &&& $[0,0.35]$ \\ 
	gen2 &&& $[0.25,1]$ &&& $[0,0.25]$ &&& $[0,0.5]$ &&& $[0,0.25]$  \\ \hline
\end{tabular}
\caption{
\label{defsubregions}
Restrictions of the scenario space of the start situation of
the given uncertainty sequence to specific subregions.
Whereas the scenario set 'gen0' covers the complete scenario space,
the scenario sets 'gen1', 'gen2' are restrictions of 'gen0' subject
to the indicated constraints.
For the parameters $N(0)$, $\Delta t_{\mathsf{mal}}$, and $\alpha$, no additional constraints are raised.
}
\end{table}

\begin{table}[tb!]
\tabulinesep=0.3mm
\setlength{\tabcolsep}{1.7mm}
\tabulinesep=1.0mm
	\begin{tabu} to \textwidth {| L{1.8cm}  | L{1.3cm} || X[1] | X[1] | X[1] | X[1] | X[1] | X[1] | X[1] | X[1] | }
\hline
Win/Loss & Scenario&  \multicolumn{2}{c|}{Strong Wins}  & \multicolumn{2}{c|}{Weak Wins}  & \multicolumn{2}{c|}{Weak Losses}  & \multicolumn{2}{c|}{Strong Losses}    \\ \cline{3-10}
Class &Set&  gen & cyb & gen & cyb & gen & cyb & gen & cyb \\ \hline \hline
Strong  			& gen0  & 0.980 & 1.000 & 0.009 & 0.000 & 0.008 & 0.000 & 0.003 & 0.000 \\ 
kinetic 		        & gen1 & 0.987 & 1.000 & 0.008 & 0.000 & 0.005 & 0.000 & 0.000 & 0.000 \\ 
superiority 						        & gen2 & 1.000 & 1.000 & 0.000 & 0.000 & 0.000 & 0.000 & 0.000 & 0.000 \\ \hline
Weak 			     	& gen0 &  0.085 & 0.096 & 0.891 & 0.904 & 0.000 & 0.000 & 0.023 & 0.000 \\ 
kinetic 		     	& gen1 & 0.139 & 0.141 & 0.846 & 0.859 & 0.000 & 0.000 & 0.015 & 0.000 \\ 
superiority   		     				        & gen2 & 0.172 & 0.172 & 0.828 & 0.828 & 0.000 & 0.000 & 0.000 & 0.000 \\ \hline
Weak 				& gen0 &  0.071 & 0.077 & 0.000 & 0.000 & 0.891 & 0.923 & 0.037 & 0.000 \\ 
kinetic 			& gen1 & 0.203 & 0.206 & 0.000 & 0.000 & 0.748 & 0.794 & 0.048 & 0.000 \\ 
inferiority   						        & gen2 & 0.750 & 0.750 & 0.000 & 0.000 & 0.250 & 0.250 & 0.000 & 0.000 \\ \hline
Strong				& gen0 &  0.015 & 0.016 & 0.012 & 0.016 & 0.023 & 0.026 & 0.950 & 0.942 \\ 
kinetic 			& gen1 & 0.041 & 0.042 & 0.025 & 0.029 & 0.052 & 0.052 & 0.882 & 0.877 \\ 
inferiority 						        & gen2 & 0.141 & 0.141 & 0.047 & 0.055 & 0.055 & 0.055 & 0.758 & 0.750 \\ \hline
\end{tabu}
\caption{
	\label{subregionsdata}
	The influence of cyber actions summarized as
	fractions of wins and losses for the scenario sets 'gen0', 'gen1', 'gen2' defined in table~\ref{defsubregions}.
}
\end{table}

The influence of appropriately chosen cyber action parameters
should be examined more closely. For this purpose, subregions 'gen0',
'gen1' of the scenario space of the start situation of the 
uncertainty sequence are defined in table~\ref{defsubregions}. 
The influence of cyber actions is clearly visible in
table~\ref{subregionsdata} with the number of outcomes, which belong
to a different win/loss class (table columns) 
than the corresponding pure kinetic situation (table rows).
The fraction of strong wins increases from 'gen0' to 'gen1' 
in all four kinetic win/loss classes significantly. 
The increase in case of weak kinetic
inferiority could even be considered as dramatic: From 0.071 
(for 'gen') resp. 0.077 (for 'cyb') to 0.750 
(for 'gen' and 'cyb'). Remarkable as well is the increase of 
strong wins for strong kinetic inferiority: From 0.015 resp. 0.016
to 0.141. 
The small differences between 'gen' and 'cyb' may be 
caused by the assumption of cyber superiority for Blue from 
the beginning.
Based on these observations, we conclude that the small number 
of changing win/loss classifications due to cyber actions in
table~\ref{startastable2} seems indeed be caused by the random
parameterization of the cyber actions.

\subsection{Influence of Uncertainties}\label{secinfuncertain}\label{ssun}

The inclusion of uncertainties resulting from missing
knowledge is an important contribution towards practicability. It
requires the consideration of a whole set of scenarios instead of
only a single
one. The statistical properties of the corresponding set of outcomes
provides information not accessible otherwise. Its value results from 
the incompatibility between statistics and dynamics. This means that
the typical outcome does not necessarily coincide with the so-called
point estimate, i.e.\ the outcome of the typical scenario, as
demonstrated in table~\ref{compastable}. 
%

The point estimate of a Monte-Carlo design space $(X,\Pr)$ is 
defined as follows.
According to section~\ref{secsimhor}, a scenario $x\in X$ is the Cartesian
product $x= \bigtimes_k x_k$ of the model parameters (see
table~\ref{definition_param}) and of the initial values of equation
system (\ref{lanchmalware}). If we neglect the few 
dependencies between the parameters $x_k$ in the example, 
the probability distribution $\Pr$ is decomposable as well according
to $\Pr= \bigtimes_k \Pr_k$. This facilitates the definition of
the 'typical' scenario as $x_{\Pr} := \bigtimes_k \mean(x_{k})$.
Then, the {\em point estimate} of $(X,\Pr)$ is given as $\simul(x_{\Pr})$.

The differences between risk assessments with inclusion of uncertainties
(indicated by term 'sampling', being the mean of all sampled outcomes)
and corresponding 'point estimates'
in table~\ref{compastable} can typically be neglected for small 
uncertainties. In the case of large uncertainties, however, the
deviation may be significant. This is demonstrated by the losses
$L_b$ in the start situation of the example. 
Beyond that, modified win/loss classifications due to the inclusion
of uncertainties cannot be excluded. In the table, this occurs 
in the steps 'cyber attack' and 'enemy attack' due to different
signs of $\mean(\Delta N)$. 

\section{Application of the framework} \label{secapplic}

After the presentation of the proposed framework, some remarks about its
application should be made.
This especially concerns, how information about the situation,
which should be analyzed, is extracted and handed over to the framework. 
With regard to ensuring an adequate representation of the situation by the model used in the framework, we refer to the related discussions in the sections~\ref{modelingstructure} and \ref{secexamplemodel}.
In the following, we will additionally address the quantitative
determination of the values of the model parameters.
We will also discuss, how the results computed with help of the framework 
can be utilized for achieving better decisions.
These considerations demonstrate that the framework is not only scientifically correct but also practically useful. 

\subsection{Providing the Inputs of the Framework}

In order to apply the framework, the parameter values characterizing 
the scenario(s) to be processed with the framework model have to be
determined \cite{schweidtmann2020obey}.
This translates the maybe quite abstract 
properties of the real-world force elements into numeric input of the 
model. For the purpose of our discussion, we will again focus on
the example in section~\ref{secfour}.

The values can be determined by measurement executed either
in reality or in corresponding computer-based simulation experiments
\cite{hou2013modeling}. The latter is especially of interest, if the 
values have to be known {\em before} the situation unfolds in reality.
Applying these usually quite elaborated simulation models also to risk
assessment purposes cannot be recommended, though; their large runtime and lack
of important theoretical properties (see section~\ref{secmodelextensions} 
below) exclude such usage. 
Measurement errors can be 
taken into account as uncertainties in accordance
with the design of the framework.

The values of the parameters $\delta_{}$, $p$, $q$ of the Lanchester component --- resp.\ their statistics \cite{lucas2004effect,dinges2001exploring} due to their dependence on the specific circumstances of an engagement \cite{ancker1987validity} --- can be determined by fitting a Lanchester model to the dynamic behavior of a discrete event simulation model \cite{taylor1980lanchester}. The parameter $\eta_{}$ is accessible analogously based on the effects of a malware infection on force element capabilities. The framework integrates value statistics seamlessly as value uncertainties.
Concerning the parameters $\beta$ and $\gamma$ (resp.\ $\tilde \gamma$) of the SIR component, we follow \cite{hall2006thwarting}. According to the differential equation system (\ref{lanchmalware}), the infection rate $\beta$ can be determined by measuring the number $I_{\Delta T}$ of new malware infections in a (small) time period $\Delta T$ based on $I_{\Delta T}= \beta S I N^{-1} \Delta T $ \cite{hosseini2016agent}. An observation of the time period $\tau$, in which a force element is infected with malware, immediately gives the recovery rate $\gamma = \tau^{-1}$ \cite{yao2018epidemic,tiensin2007transmission,hosseini2016agent}. Alternatively, parameters like $\beta$ can be determined based on malware properties. According to \cite{zou2006performance}, it holds $\beta = \xi/O $, where $\xi$ is the scanning rate of the malware and $O$ the size of the scanned subset of potential infection targets.
Finally, the values of the parameters $\Delta t_{\mathsf{att}}$, $\Delta t_{\mathsf{mal}}$, $\Delta t_{\mathsf{pat}}$, and $\alpha$ result from decisions of the military leaders.
As such, they belong to the what-if part of the application of the framework. They cannot be predicted in a deterministic way. 




\subsection{Applying the Results of the Framework}

The framework computes the value of the risk $R$ associated with a
conflict situation $S$ based on a model $M\in\mathcal{M}$ of $S$.
Since we cannot completely rule out unjustified assumptions about
the situation at hand 
in principle, the determination of the value of $R$ can be  sometimes
worthless in spite of the correctness of the computation itself.
In the end, it depends on the judgement of the user whether the
information about $R$ is considered as valuable. 
Should this be the case, he also determines how it is then exploited.
He may modify the planned course of action and/or decide about actions
in order to change the situation $S$ towards a lower risk value, but 
he may also choose to still disregard the risk estimate completely.
The framework {\em supports} human decision making, but does not replace it.
It is the human decider who selects the scenarios he is especially interested in. He does so by setting the knowledge encoded into the underlying model and into the parameter values in an overall situational context. 
Analogously, he disposes about the usage of the computed values of risk assigned 
to these scenarios. 

\subsection{Utility of the Framework}

The utility of the framework relates to 
the quality of the decisions of the military leaders 
applying it. In this respect, 'quality' means the tendency 
to avoid situations with a high number of own losses.
(Un)Fortunately, such an observation-based utility adjudication is
hampered by an insufficient amount of real-world data.

To get around this difficulty, we discuss the utility of the 
framework resting upon theoretical arguments instead of observations.
Without doubt, long-term predictions with inclusion of uncertainties 
are a task hardly manageable by humans in general.
The provision of measures like the risk $R$ with help of the 
framework can thus be considered as highly valuable knowledge, 
as long as an extensive validation effort supports the correctness 
of the value of $R$. For rationally acting users, the knowledge
about $R$ can thus be viewed as a principal advantage. It simplifies
pending decisions by a reduction of existing unknowns. 
The framework may therefore be judged as practically useful.

\section{Discussion and Outlook}\label{secdiscussion}

In the following, we will point out limitations and restrictions 
of the proposed framework as well as opportunities of future research. 

\subsection{Extension of the Class $\mathcal{M}$ of Models}\label{outoptsec}

The faithfulness of situation representations may be increased
by extending the underlying class $\mathcal{M}$ of models due
to the additional modeling options. 
System dynamics components of kinetic combat and of malware propagation 
more general than utilized in the model class $\mathcal{M}$ are 
proposed in e.g.\ \cite{bull2010system,colizza2007epidemic}.
Especially the Lanchester component may profit from such a generalization.
The basic Lanchester model assumes, for example, that incapacitated 
opponents must be known at once, so that fire is distributed only 
against active opponents \cite{hughes1995two}. For being more realistic, 
the compartment $D$ representing the already destroyed elements may attract
a certain fraction of the fire of the opposing force as well.

The investigation of further examples will clarify, whether a risk 
analysis of more detailed --- especially heterogeneous --- models 
can proceed according to the template given by section~\ref{secfour}. 
Tractability may be a possible concern, since more details will 
usually be accompanied with a higher number of parameters.
It seems hardly feasible to determine the accurate values of a 
large number of model parameters, though. 
Significant uncertainties will be the result, which in turn require
the execution of a large number of Monte Carlo simulation runs for 
assuring a sufficient statistical significance of the computed risk value.
It may thus be difficult to make straightforward use of the 
details of a more elaborated model without corresponding knowledge
about the values of the model parameters.

The restriction of $\mathcal{M}$ to deterministic models should be 
abandoned, as soon as small force numbers play a role. Such a situation 
occurs in the start phase of an infection with self-replicating malware,
in which the number of infected elements is small, and for overall small 
forces \cite{del2015mathematical}.  
Then, stochastic fluctuations eventually influencing the system dynamics
will not cancel out anymore. 
The extension of $\mathcal{M}$ towards stochastic models may thus
contribute to the validity of results.
A stochastic model dynamics can be perfectly integrated in the
Monte Carlo approach of the framework. 

\subsection{Improving the Utility of the Framework}


\begin{figure}[bthp]
\tabulinesep=0.7mm
\begin{center}
\includegraphics[scale=0.12]{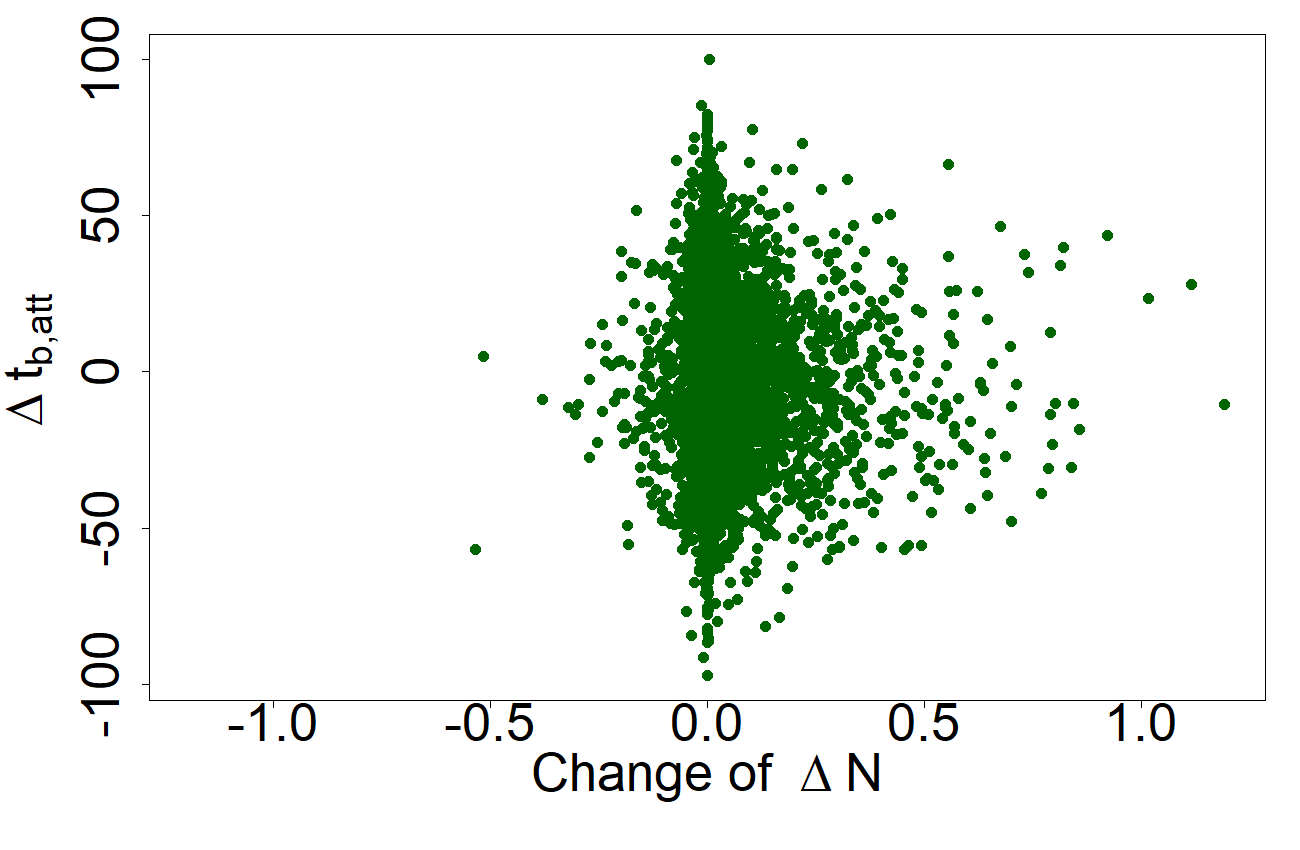}
\end{center}
\caption{
The distribution of $\Delta \Delta N$ dependent on $\Delta t_{b,\mathsf{att}}$ indicates, which values of $\Delta t_{b,\mathsf{att}}$ may cause large cyber effects.
\label{attacktiming}
}
\end{figure}

The proposed framework can be applied for other purposes 
than decision support as well.
Based on the calculated risk, various attack methods, infection 
mechanisms, propagation pathways, and countermeasures can be 
analyzed and compared with each other. This permits an 
identification of preferable system design principles.
Experimentation with the consequences of different knowledge
(or assumptions) about a given situation is enabled. 

Beyond that, the framework can help for configuring the cyber actions.
Let us take a look at the timing $\Delta t_{b,\mathsf{att}}$ of the 
blue cyber attack in the analyzed example. In the step 'Cyber Attack' of 
the uncertainty sequence, the kinetic aspects are already fixed. 
The parameter settings $R_r(0) \approx 0.15$, 
$\beta_r\approx 1.75$ may also have already become known.
For scenarios satisfying the given
technical malware parameters, the changes of $\Delta N$ dependent
on the timing $\Delta t_{b,\mathsf{att}}$ of the blue malware attack 
are distributed as shown in figure~\ref{attacktiming}. 
According to the plot, we may eventually expect a large advantageous 
influence on the outcome of the combat for
$\Delta t_{b,\mathsf{att}}\approx -15.0$. There, the changes of $\Delta N$
reach their maximum both absolutely and in the mean though several 
singular data points representing large positive changes of $\Delta N$
can also be observed for $\Delta t_{b,\mathsf{att}} > 0$.
Hence, the parameter setting $\Delta t_{b,\mathsf{att}}\approx -15.0$
seems to be a good working hypothesis.
As can be seen in section~\ref{ssra}, 
the situation indeed has been improved with this choice.
We are not going to discuss these additional benefits further, though.
They may be topic of future research.

\section{Acknowledgements}

The authors are grateful for many helpful remarks 
provided by Filipa Campos-Viola, Elisa Canzani, Stefan Hahndel, 
Thomas Rieth, Harald Schaub, Christine Schwarz-Hemmert and
Iris Berthold.
Additionally, we want to thank Simon Hurst for his support 
for validating the implementation of the computational model. 

\bibliographystyle{plain}
\bibliography{overall_work_references}

\end{document}